\documentclass[12pt]{article}
\usepackage{amsmath, bm, physics, multirow, BOONDOX-cal,amsthm}
\usepackage{amssymb}
\usepackage{lscape}

\DeclareMathOperator*{\argmax}{\arg\!\max}

\usepackage{graphicx,psfrag,epsf}
\usepackage{enumerate}
\usepackage{natbib}
\usepackage{url} 
\usepackage{amsfonts}

\newtheorem{corollary}{Corollary}

\newtheorem{lemma}{Lemma}
\newtheorem{proposition}{Proposition}

\newtheorem{theorem}{Theorem}
\newtheorem{assumption}{Assumption}

\numberwithin{corollary}{section}
\numberwithin{definition}{section}
\numberwithin{equation}{section}
\numberwithin{lemma}{section}
\numberwithin{proposition}{section}
\numberwithin{remark}{section}
\numberwithin{theorem}{section}

\begin{document}

\def\spacingset#1{\renewcommand{\baselinestretch}%
{#1}\small\normalsize} \spacingset{1}


  \begin{center}
    {\Large \bf Time-Varying Multivariate Causal Processes}
    \bigskip
    
    $^{\ast}${\sc Jiti Gao} and $^{\ast}${\sc Bin Peng} and $^{\dag}${\sc Wei Biao Wu} and $^{\ast}${\sc Yayi Yan}
\medskip

    $^{\ast}$Department of Econometrics and Business Statistics,  Monash University,\\
    and $^{\dag}$Department of Statistics, University of Chicago
\end{center}
  \medskip

\begin{abstract}
In this paper,  we consider a wide class of time-varying multivariate causal processes which nests many classic and new examples as special cases. We first  prove the existence of a weakly dependent stationary approximation for our model which is the foundation to initiate the theoretical development. Afterwards, we consider the QMLE estimation approach, and provide both point-wise and simultaneous inferences on the coefficient functions. In addition, we demonstrate the theoretical findings through both simulated and real data examples. In particular, we show the empirical relevance of our study using an application to evaluate the conditional correlations between the stock markets of China and U.S. We find that the interdependence between the two stock markets is increasing over time.
\medskip

\noindent%
\smallskip


\end{abstract}

\vfill

\newpage
\spacingset{1.9} 

\section{Introduction}\label{Sec1}

The family of vector autoregressive (VAR) models and the family of multivariate (G)ARCH models are among some of the most popular frameworks for modelling dynamic interactions of multiple variables. The VAR family usually captures the dynamic by imposing structures on the time series itself, while the (G)ARCH family imposes restrictions on the conditional second moments. We acknowledge the vast literature of both families, and have no intention to exhaust all relevant studies in this paper for the sake of space. We refer interested readers to \cite{stock2001vector} and \cite{bauwens2006multivariate} for excellent review on both families. 

Although both families have rich literature on their own, to the best of the authors' knowledge not many works have been done to bridge them. Among limited attempts (e.g., \citealp{ling2003asymptotic,bardet2009asymptotic}), most (if not all) of these studies rely on the stationarity assumption. While the stationarity assumption comes in handy when deriving asymptotic properties, it may not be very realistic in practice (\citealp{preuss2015detection,chen2021inference}). For example, economic and financial data always include different macro shocks, as a consequence the behaviour can be quite volatile; the climate data may contain certain time trend which recently has attracted lots of attention due to greenhouse emission; etc. Anyway, certain nonstationarity may always occur. 

To account for nonstationarity, locally stationary processes have received considerable attention since the seminal work of \cite{dahlhaus1996kullback}, \cite{dette2011measure}, \cite{zhang2012inference}, \cite{truquet2017parameter}, \cite{dahlhaus2019towards}, among others. In contrast to the unit root process, the locally stationary process nicely balances stationarity and nonstationarity by allowing for the simultaneous presence of both types of behaviours in one time series process. In a very recent paper, \cite{karmakar2021simultaneous} consider simultaneous inference for a general class of univariate $p$-Markov processes with time-varying coefficients, which covers several time-varying versions of the classical univariate models (e.g., AR, ARCH, AR-ARCH) as special cases. Despite its generality, their study still rules out the time-varying versions of some widely used models (e.g., ARMA, GARCH, ARMA-GARCH). Also, it is worth mentioning this line of research heavily focuses on univariate time series, which somewhat limits the popularity of locally stationary processes.

That said, it is reasonable to call for a framework which can marry the VAR family and the (G)ARCH family while allowing for nonstationarity. To provide a concrete example, consider a time-varying multivariate GARCH model, which can model the co-movements of financial returns. Detailed investigation on such a model can help answer research questions like (i). Is the volatility of a market leading the volatility of other markets? (ii) Whether the correlations between asset returns change over time? (iii). Are they increasing in the long run, perhaps because of the globalization of financial markets? These are of great practical importance for both investors and policymakers (\citealp{bauwens2006multivariate,diebold2009measuring}).

To allow for flexibility as much as possible from the modelling perspective, we consider a class of multivariate causal processes as follows:
\begin{eqnarray}\label{Eq2.1}
	\mathbf{x}_t =\left\{\begin{array}{ll}
	 \bm{\mu}\left(\mathbf{x}_{t-1},\mathbf{x}_{t-2},\ldots;\bm{\theta}(\tau_t)\right) + \mathbf{H}\left(\mathbf{x}_{t-1},\mathbf{x}_{t-2},\ldots;\bm{\theta}(\tau_t)\right)\bm{\varepsilon}_t, & \text{for}\quad t=1,\ldots, T\\
	\bm{\mu}\left(\mathbf{x}_{t-1},\mathbf{x}_{t-2},\ldots;\bm{\theta}(0)\right) + \mathbf{H}\left(\mathbf{x}_{t-1},\mathbf{x}_{t-2},\ldots;\bm{\theta}(0)\right)\bm{\varepsilon}_t & \text{for}\quad t\le 0
	\end{array}\right. ,
\end{eqnarray}
where $\tau_t= t/T$, $\bm{\mu}\left(\cdot\right)$ is an $m$-dimensional random vector, $\mathbf{H}\left(\cdot\right)$ is an $m\times m$-dimensional random matrix, $\bm{\theta}(\tau)$ is a $d\times 1$ time-varying parameter of interest with each element belonging to $C^3[0,1]$, and $\{\bm{\varepsilon}_t\}$ is a sequence of independent and identically distributed (i.i.d.) random vectors. Note that the value of $d$ usually depends on the value of $m$, and the connection becomes clear once a specific model is considered.  As far as we are concerned, both of $m$ and $d$ are fixed throughout the paper. Notably, both $\bm{\mu}(\cdot)$ and $ \mathbf{H}(\cdot)$ are known, and share the same unknown parameter $\bm{\theta}(\cdot)$. The setting for $t\le 0$ regulates the time series for the periods that we do not observe, which is commonly adopted when certain nonstationarity gets involved (e.g., \citealp{vogt2012nonparametric}). Essentially, it requires the initial time period does not have a diverging behaviour.

Before proceeding further, we provide two examples to briefly illustrate the rationality behind \eqref{Eq2.1}, and leave the detailed investigation on these examples to Section \ref{Sec2.4}.  We refer interested readers to  \cite{ling2003adaptive}, \cite{ling2003asymptotic} and \cite{bardet2009asymptotic} for extensive investigation on the parametric counterparts of these examples.

\noindent \textbf{Example 1}: Consider the time-varying VARMA($p,q$) model
\begin{equation}\label{Eq4.1}
\mathbf{x}_t = \mathbf{a}(\tau_t) + \sum_{j=1}^{p}\mathbf{A}_j(\tau_t)\mathbf{x}_{t-j} + \bm{\eta}_t + \sum_{j=1}^{q}\mathbf{B}_j(\tau_t)\bm{\eta}_{t-j} \quad \text{with}\quad \bm{\eta}_t = \bm{\omega}(\tau_t)\bm{\varepsilon}_t.
\end{equation}
It is not hard to show that \eqref{Eq4.1} admits a presentation in the form of \eqref{Eq2.1}, and 
\begin{eqnarray}\label{Eq4.1.1}
\bm{\theta}(\tau)=\mathrm{vec}(\mathbf{a}(\tau),\mathbf{A}_1(\tau),\ldots,\mathbf{A}_p(\tau),\mathbf{B}_1(\tau),\ldots,\mathbf{B}_q(\tau),\bm{\Omega}(\tau)),
\end{eqnarray}
where $\bm{\Omega}(\cdot):=\bm{\omega}(\cdot)\bm{\omega}^\top(\cdot)$.

\smallskip

\noindent \textbf{Example 2}: Consider the time-varying multivariate GARCH($p,q$) model 
\begin{eqnarray}\label{Eq4.4}
\mathbf{x}_t &=& \mathrm{diag} (h_{1,t}^{1/2},\ldots,h_{m,t}^{1/2} ) \bm{\eta}_t,\nonumber \\
\mathbf{h}_t  &=& \mathbf{c}_0(\tau_t) + \sum_{j=1}^{p} \mathbf{C}_j(\tau_t) \left(\mathbf{x}_{t-j}\odot\mathbf{x}_{t-j}\right) + \sum_{j=1}^{q} \mathbf{D}_j(\tau_t)\mathbf{h}_{t-j},
\end{eqnarray}
where $h_{j,t}$ stands for the $j^{th}$ element of  $\mathbf{h}_t$, and $\bm{\eta}_t = \bm{\Omega}^{1/2}(\tau_t) \bm{\varepsilon}_t$. The model \eqref{Eq4.4} generalizes the models of \cite{bollerslev1990modelling} and \cite{jeantheau1998strong}. Similar to Example 1, we show that \eqref{Eq4.4} admits a representation in the form of \eqref{Eq2.1}, and 
\begin{eqnarray}\label{Eq4.4.1}
\bm{\theta}(\tau)=\mathrm{vec}(\mathbf{c}_0(\tau),\mathbf{C}_1(\tau),\ldots,\mathbf{C}_p(\tau),\mathbf{D}_1(\tau),\ldots,\mathbf{D}_q(\tau),\bm{\Omega}(\tau)).
\end{eqnarray}

\smallskip

In view of the development of Example 1 and Example 2 in Section \ref{Sec2.4}, one may further show the time-varying counterparts of the parametric models mentioned in \cite{bardet2009asymptotic} are also covered by \eqref{Eq2.1}. To this end, we argue that \eqref{Eq2.1} does not only allows for nonstationarity and conditional heteroskedasticity, but also provides sufficient flexibility to cover many well adopted models in the literature.

In this paper, our contributions are in the following four-fold: (1). we consider a wide class of time-varying multivariate causal processes which nests many classic and new examples as special cases; (2). we prove the existence of a weakly dependent stationary approximation for the model \eqref{Eq2.1} at any given time of interest (i.e., $\forall\tau\in[0,1]$), which is the foundation in order to establish asymptotic properties associated with the model; (3). we establish the estimation theory, and provide both point-wise and simultaneous inferences on the coefficient functions of which both are important for practical works (\citealp{zhou2010simultaneous}); (4). we demonstrate the theoretical findings through both simulated and real data examples.

The paper is organized as follows. Section \ref{Sec2} presents the theoretical findings associated with the stationary approximation, estimation and  inferences.  In Section \ref{Sec3}, we conduct extensive simulation studies to examine the theoretical findings, and further investigate the time-varying conditional correlations between the Chinese and U.S. Stock market. Section \ref{Sec4} concludes. Due to space limit, we give the proofs of the main results to the online appendices of the paper.

Before proceeding further, it is convenient to introduce some notation: the symbol $|\cdot|$ denotes the Euclidean norm of a vector or the spectral norm for a matrix; $\|\mathbf{v}\|_q:=\left(E|\mathbf{v}|^q\right)^{1/q}$ and $\|\cdot\|:=\|\cdot\|_2$ for short; $\otimes$ denotes the Kronecker product; $\odot$ denotes the Hadamard product; $\mathbf{I}_a$ stands for an $a\times a$ identity matrix; $\mathbf{0}_{a\times b}$ stands for an $a\times b$ matrix of zeros, and we write $\mathbf{0}_a$ for short when $a=b$; for a function $g(w)$, let $g^{(j)}(w)$ be the $j^{th}$ derivative of $g(w)$, where $j\ge 0$ and $g^{(0)}(w) \equiv g(w)$; $K_h(\cdot) =K(\cdot/h)/h$, where $K(\cdot)$ and $h$ stand for a nonparametric kernel function and a bandwidth respectively; let $\tilde{c}_k =\int_{-1}^{1} u^k K(u) \mathrm{d}u$ and $\tilde{v}_k= \int_{-1}^{1} u^k K^2(u) \mathrm{d}u$ for integer $k\ge 0$; $\mathrm{diag}(\mathbf{a})$ is a diagonal matrix with the vector $\mathbf{a}$ on its main diagonal, while $\mathrm{diag}(\mathbf{A})$ creates a vector from the diagonal of matrix $\mathbf{A}$; finally, let $\to_P$ and $\to_D$ denote convergence in probability and convergence in distribution, respectively.

\section{Estimation and Asymptotics}\label{Sec2}

In this section, we first prove the existence of a weakly dependent stationary approximation for the model \eqref{Eq2.1} in Section \ref{Sec2.1}; we then provide the estimation approach using the local linear quasi-maximum-likelihood estimation and establish the asymptotic properties of the proposed estimator in Section \ref{Sec2.2}; Section \ref{Sec2.3} provides results on both point-wise and simultaneous inferences; Section \ref{Sec2.4} gives some detailed examples to justify the usefulness of our study.

\subsection{Stationary Approximation}\label{Sec2.1}

To study \eqref{Eq2.1}, the first challenge lies in the fact that the model may not be stationary. Therefore, for $\forall\tau\in [0,1]$, we initial our analysis by finding a stationary approximation for each $\mathbf{x}_t$ with $t\ge 1$. By doing so, we are able to measure the weak dependence of $\{\mathbf{x}_t \}$ using the nonlinear system theory in \cite{wu2005nonlinear}, which then provides us a framework to derive the asymptotic properties  accordingly.

To be clear on the dependence measure, consider an example in which $\mathbf{e}_t$ is a stationary process, and admits a causal representation $\mathbf{e}_t = \mathbf{J}(\bm{\varepsilon}_t,\bm{\varepsilon}_{t-1},\ldots)$ with $\mathbf{J}(\cdot)$ being a measurable function. See \cite{tong1990non} for discussion on nonlinear time series of this kind. For $k\geq 0$, we define the following dependence measure:
\begin{eqnarray}\label{DefDM}
\delta_{r}^{\mathbf{e}}(k)=\left\|\mathbf{J}(\bm{\varepsilon}_k,\bm{\varepsilon}_{k-1},\ldots\bm{\varepsilon}_1,\bm{\varepsilon}_{0},\bm{\varepsilon}_{-1},\ldots)-\mathbf{J}(\bm{\varepsilon}_k,\ldots, \bm{\varepsilon}_1,\bm{\varepsilon}_{0}^*,\bm{\varepsilon}_{-1},\ldots)\right\|_r,
\end{eqnarray}
where $\bm{\varepsilon}_0^*$ is an independent copy of $\{\bm{\varepsilon}_j\}$. Being able to measure the time series dependence such as \eqref{DefDM} is the starting point for time series analyses. 

We now introduce some basic assumptions.
\begin{assumption} \label{Ass1}
\item 
\begin{enumerate}
\item $\{\bm{\varepsilon}_t\}$ is a sequence of i.i.d. random vectors with $E(\bm{\varepsilon}_1)=\mathbf{0}$, $E (\bm{\varepsilon}_1\bm{\varepsilon}_1^\top )=\mathbf{I}_m$, and  $ \|\bm{\varepsilon}_1 \|_r<\infty$ for some $r\ge 2$.

\item For $\forall\mathbf{z},\mathbf{z}^\prime \in (\mathbb{R}^m)^\infty$ and $\forall\bm\vartheta\in \mathbb{R}^d$, there exist nonnegative sequences $\{\alpha_j(\bm{\bm\vartheta})\}_{j=1}^{\infty}$ and $\{\beta_j(\bm{\bm\vartheta})\}_{j=1}^{\infty}$ such that 
\begin{eqnarray*}
&&|\bm{\mu}(\mathbf{z};\bm{\bm\vartheta})-\bm{\mu}(\mathbf{z}^\prime;\bm{\bm\vartheta})| \leq \sum_{j=1}^{\infty}\alpha_j(\bm{\bm\vartheta}) |\mathbf{z}_j-\mathbf{z}_j^\prime|,\\
&&|\mathbf{H}(\mathbf{z};\bm{\bm\vartheta})-\mathbf{H}(\mathbf{z}^\prime;\bm{\bm\vartheta})| \leq \sum_{j=1}^{\infty}\beta_j(\bm{\bm\vartheta})|\mathbf{z}_j-\mathbf{z}_j^\prime|,
\end{eqnarray*}
where $\mathbf{z}_j$ and $\mathbf{z}_j^\prime$ are the $j^{th}$ columns of $\mathbf{z}$ and $\mathbf{z}^\prime$ respectively.

\item For $\forall \tau\in [0,1]$, $\bm{\theta}(\tau)$ lies in the interior of $\bm{\Theta}_r$, where 
\begin{eqnarray*}
\bm{\Theta}_r:=\left\{\bm{\vartheta}\in \bm{\Theta} \mid \sum_{j=1}^{\infty}\alpha_j(\bm{\vartheta}) + \left\|\bm{\varepsilon}_1\right\|_r \sum_{j=1}^{\infty}\beta_j(\bm{\vartheta}) < 1 \right\}
\end{eqnarray*}
and $\bm{\Theta}$ is a compact set of $\mathbb{R}^d$. 
\end{enumerate}
\end{assumption}

Assumption \ref{Ass1}.1 is standard when studying dynamic time series model (\citealp{lutkepohl2005new}). In Assumption \ref{Ass1}.2, $\bm\vartheta$ is a generic $d\times 1$ vector, and has the same length as $\bm \theta(\cdot)$. This assumption imposes Lipschitz-type conditions on $\bm{\mu}(\cdot)$ and $\mathbf{H}(\cdot)$, which are rather minor, and can be easily fulfilled by a variety of models such as those mentioned in Section \ref{Sec1}. See Propositions \ref{proposition4.1}-\ref{proposition4.2} below for details. Assumption \ref{Ass1}.3 does not only guarantee a stationary approximation for each $\mathbf{x}_t$, but also ensures the approximated process has some proper moments. Similar conditions have also been adopted in \cite{bardet2009asymptotic}. 

With these conditions in hand, we present the following proposition which facilitates the development in what follows.

\begin{proposition}\label{proposition2.1}
Let Assumption \ref{Ass1} hold. For any $\tau \in [0,1]$, there exists a stationary process
\begin{eqnarray*}
	\widetilde{\mathbf{x}}_t(\tau) = \bm{\mu}\left(\widetilde{\mathbf{x}}_{t-1}(\tau),\widetilde{\mathbf{x}}_{t-2}(\tau),\ldots;\bm{\theta}(\tau)\right) + \mathbf{H}\left(\widetilde{\mathbf{x}}_{t-1}(\tau),\widetilde{\mathbf{x}}_{t-2}(\tau),\ldots;\bm{\theta}(\tau)\right)\bm{\varepsilon}_t
\end{eqnarray*}
such that
	\begin{enumerate}
		\item $\sup_{\tau\in[0,1]}\left\|\widetilde{\mathbf{x}}_t(\tau)\right\|_r < \infty$,
		
		\item $\delta_{r}^{\widetilde{\mathbf{x}}(\tau)}(k) \leq O(1)\inf_{1\leq p \leq k}\{\rho(\tau)^{k/p} + \sum_{j=p+1}^{\infty}\left[\alpha_j(\bm{\theta}(\tau)) + \beta_j(\bm{\theta}(\tau))\right]\} \to 0$ as $k\to \infty$, 
	\end{enumerate}
	where $\rho(\tau) := \sum_{j=1}^{\infty}\alpha_j(\bm{\theta}(\tau)) + \left\|\bm{\varepsilon}_1\right\|_r \sum_{j=1}^{\infty}\beta_j(\bm{\theta}(\tau))$.
\end{proposition}	

It is worth mentioning that for a univariate $p$-Markov process
\begin{eqnarray*}
	\widetilde{x}_{p,t}(\tau) = \mu\left(\widetilde{x}_{t-1}(\tau),\ldots,\widetilde{x}_{t-p}(\tau);\bm{\theta}(\tau)\right) + H\left(\widetilde{x}_{t-1}(\tau),\ldots,\widetilde{x}_{t-p}(\tau);\bm{\theta}(\tau)\right)\varepsilon_t,
\end{eqnarray*}
\cite{karmakar2021simultaneous} show that there exists $0 < \rho < 1$ such that $\sup_{\tau \in[0,1]}\delta_{r}^{\widetilde{x}_p(\tau)}(k) =O(\rho^k)$ based on the development of \cite{wu2004limit}. From a methodological viewpoint, we give a set of new proofs which allow us to measure the dependence of multivariate causal processes with infinity memory. The term $\sum_{j=p+1}^{\infty}\left[\alpha_j(\bm{\theta}(\tau)) + \beta_j(\bm{\theta}(\tau))\right]$ in the second result of Proposition \ref{proposition2.1} arises due to the infinity memory structure of $\widetilde{\mathbf{x}}_t(\tau)$. Thus, the dependence $\delta_{r}^{\widetilde{\mathbf{x}}(\tau)}(k)$ relies on the choice of $p$ and the decay rates of the coefficients $\alpha_j(\bm{\theta}(\tau))$ and $\beta_j(\bm{\theta}(\tau))$. 

To ensure $\widetilde{\mathbf{x}}_t(\tau) $ can approximate $\mathbf{x}_t$ reasonably well, we impose more structure below.

\begin{assumption}\label{Ass2}
	\item
	\begin{enumerate}
		\item There exists a nonnegative sequence $\{\chi_j\}$ with $\sum_{j=1}^\infty\chi_j < \infty$ such that for $\forall\mathbf{z}\in (\mathbb{R}^m)^\infty$ and $\forall\bm{\vartheta},\bm{\vartheta}^\prime\in\bm{\Theta}_r$
\begin{eqnarray*}
&&|\bm{\mu}(\mathbf{z};\bm{\vartheta})-\bm{\mu}(\mathbf{z};\bm{\vartheta}^\prime)|+ |\mathbf{H}(\mathbf{z};\bm{\vartheta})-\mathbf{H}(\mathbf{z};\bm{\vartheta}^\prime)| \leq |\bm{\vartheta}-\bm{\vartheta}^\prime| \sum_{j=1}^{\infty}\chi_j|\mathbf{z}_j| .
\end{eqnarray*}
		
		\item Let $\sup_{\tau \in [0,1]} \alpha_j(\bm{\theta}(\tau)) = O(j^{-(2+s)})$ and $\sup_{\tau \in [0,1]} \beta_j(\bm{\theta}(\tau)) =  O(j^{-(2+s)})$ for some $s > 0$.
	\end{enumerate}
\end{assumption}

Assumption \ref{Ass2}.1 imposes another Lipschitz-type condition with respect to the parameter space. Assumption \ref{Ass2}.2 further restricts the decay rates of $\alpha_j(\bm{\theta}(\tau))$ and $\beta_j(\bm{\theta}(\tau))$.

Using Assumptions \ref{Ass1}--\ref{Ass2}, we can measure the distance between $\widetilde{\mathbf{x}}_t(\tau) $ and $\mathbf{x}_t$ as follows.

\begin{proposition}\label{proposition2.2}
	Suppose Assumptions \ref{Ass1}--\ref{Ass2} hold. Then
	\begin{enumerate}
		\item  $\left\|\widetilde{\mathbf{x}}_1(\tau)-\widetilde{\mathbf{x}}_1(\tau^\prime)\right\|_r = O(|\tau - \tau^\prime|)$ for $\forall\tau,\tau^\prime \in [0,1]$,
		\item  $\max_{t\ge 1}\left\|\mathbf{x}_t-\widetilde{\mathbf{x}}_t(\tau_t)\right\|_r = O(T^{-1})$.
	\end{enumerate}
\end{proposition}

We can consider Proposition \ref{proposition2.2} as the stochastic version of the H\"older continuity. Having established the stationary approximation in Proposition \ref{proposition2.2}, we move on to  investigate the estimation theory in the next subsection.

\subsection{Estimation}\label{Sec2.2}

We point out a few facts to facilitate the setup of the likelihood function. First, let $\mathbf{z}_{t} =(\mathbf{x}_t,\mathbf{x}_{t-1},\ldots)$ include all the information of $\mathbf{x}_t$ up to the time period $t$. However, in practice, our observation on $\mathbf{x}_t$ only starting from $t=1$, so we have to work with the truncated version of $\mathbf{z}_{t}$ for each $t\ge 1$:
\begin{eqnarray}
\mathbf{z}_{t}^c =(\mathbf{x}_t,\ldots, \mathbf{x}_1, \mathbf{0},\ldots).
\end{eqnarray}
Second, we note that when $\tau_t$ is sufficiently close to $\tau$, 
\begin{eqnarray}
\bm{\theta}(\tau_t) \approx \bm{\theta}(\tau ) +h\bm{\theta}^{(1)}(\tau) \cdot \frac{\tau_t-\tau}{h}.
\end{eqnarray}
Therefore, we are able to parametrize $\bm{\theta}(\cdot)$, and consider the maximum-likelihood estimation for each given $\tau$. Finally, since $\bm{\varepsilon}_t$ may not be normally distributed, we consider the local linear quasi-maximum-likelihood estimation (QMLE) method. 

Thus, our likelihood function is specified as follows:
\begin{equation}\label{Eq2.3}
	\mathcal{L}_\tau(\bm{\eta}_1,\bm{\eta}_2)=\frac{1}{T}\sum_{t=1}^{T}\mathcal{l}(\mathbf{x}_t,\mathbf{z}_{t-1}^c;\bm{\eta}_1+\bm{\eta}_2\cdot(\tau_t-\tau)/h)K_h(\tau_t-\tau),
\end{equation}
where 
\begin{eqnarray*}
	\mathcal{l}(\mathbf{x}_t,\mathbf{z}_{t-1}^c;\bm\vartheta) &=&- \frac{1}{2}(\mathbf{x}_t - \bm{\mu}(\mathbf{z}_{t-1}^c;\bm\vartheta))^\top\left(\mathbf{H}(\mathbf{z}_{t-1}^c;\bm\vartheta)\mathbf{H}(\mathbf{z}_{t-1}^c;\bm\vartheta)^\top\right)^{-1}(\mathbf{x}_t - \bm{\mu}(\mathbf{z}_{t-1}^c;\bm\vartheta))\\ 
	&&-\frac{1}{2}\log\det\left(\mathbf{H}(\mathbf{z}_{t-1}^c;\bm\vartheta)\mathbf{H}(\mathbf{z}_{t-1}^c;\bm\vartheta)^\top\right).
\end{eqnarray*}
Accordingly, for $\forall \tau$, $(\bm{\theta}(\tau),h\bm{\theta}^{(1)}(\tau))$ is estimated by
\begin{equation}\label{Eq2.4}
(\widehat{\bm{\theta}} (\tau),\widehat{\bm{\theta}}^\star(\tau)) = \argmax_{(\bm{\eta}_1,\bm{\eta}_2)\in\mathbf{E}_T(r)}\mathcal{L}_\tau (\bm{\eta}_1,\bm{\eta}_2),
\end{equation}
where $\mathbf{E}_T(r) = \bm{\Theta}_r\times(h\cdot\bm{\Theta}^{(1)})$ and $\bm{\Theta}^{(1)}$ is a compact set.

\medskip

We impose more structures in order to derive the asymptotic distribution.

\begin{assumption} \label{Ass3}
	\item
	\begin{enumerate}
		\item $\inf_{\bm{\vartheta}\in\bm{\Theta}_r, \mathbf{z}\in (\mathbb{R}^m)^\infty}\lambda_{\min}\left(\mathbf{H}(\mathbf{z};\bm{\vartheta})\mathbf{H}(\mathbf{z};\bm{\vartheta})^\top\right) \geq \underline{c}$ for some $\underline{c}>0$.
		
		\item For any $\bm{\vartheta}\in\bm{\Theta}_r$, $\bm{\mu}(\widetilde{\mathbf{z}}_t(\tau);\bm{\theta}(\tau)) =\bm{\mu}(\widetilde{\mathbf{z}}_t(\tau);\bm{\vartheta}) $ and $\mathbf{H}(\widetilde{\mathbf{z}}_t(\tau);\bm{\theta}(\tau)) =\mathbf{H}(\widetilde{\mathbf{z}}_t(\tau);\bm{\vartheta})$ a.s. imply $\bm{\vartheta} = \bm{\theta}(\tau)$ for some $t$, where $\widetilde{\mathbf{z}}_t(\tau)= \left(\widetilde{\mathbf{x}}_{t}(\tau),\widetilde{\mathbf{x}}_{t-1}(\tau),\ldots\right)$.
	\end{enumerate}
\end{assumption}

\begin{assumption} \label{Ass4}
	\item
	
	\begin{enumerate}
	\item $\bm{\mu}(\cdot;\bm\vartheta)$ and $\mathbf{H}(\cdot;\bm\vartheta)$ are twice continuously differentiable with respect to $\bm{\vartheta}$. 
	
	\item There exists a nonnegative sequence $\{\chi_j\}_{j=1}^{\infty}$ with $\chi_j = O(j^{-(2+s)})$ and some $s > 0$ such that for any $\mathbf{z},\mathbf{z}^\prime\in (\mathbb{R}^m)^\infty$ and any $\bm{\vartheta},\bm{\vartheta}^\prime\in\bm{\Theta}_r$:
\begin{eqnarray*}
&&|\gradient_{\bm{\vartheta}}^k\bm{\mu}(\mathbf{z};\bm{\vartheta})-\gradient_{\bm{\vartheta}}^k\bm{\mu}(\mathbf{z};\bm{\vartheta}^\prime)|+|\gradient_{\bm{\vartheta}}^k\mathbf{H}(\mathbf{z};\bm{\vartheta})-\gradient_{\bm{\vartheta}}^k\mathbf{H}(\mathbf{z};\bm{\vartheta}^\prime)| \leq |\bm{\vartheta}-\bm{\vartheta}^\prime| \sum_{j=1}^{\infty}\chi_j|\mathbf{z}_j|,\\
&&|\gradient_{\bm{\vartheta}}^k\bm{\mu}(\mathbf{z};\bm{\vartheta})-\gradient_{\bm{\vartheta}}^k\bm{\mu}(\mathbf{z}^\prime;\bm{\vartheta})|+|\gradient_{\bm{\vartheta}}^k\mathbf{H}(\mathbf{z};\bm{\vartheta})-\gradient_{\bm{\vartheta}}^k\mathbf{H}(\mathbf{z}^\prime;\bm{\vartheta})| \leq  \sum_{j=1}^{\infty}\chi_j|\mathbf{z}_j-\mathbf{z}_j^\prime|,
\end{eqnarray*}
	where $\gradient_{\bm{\vartheta}}= \left(\frac{\partial}{\partial \vartheta_1},\ldots,\frac{\partial}{\partial \vartheta_d}\right)^\top$, and $k=1,2$.
	\end{enumerate}

\end{assumption}

\begin{assumption}\label{Ass5}
	Let $K(\cdot)$ be a symmetric and positive kernel function defined on $[-1,1]$ with $\int_{-1}^{1}K(u)\mathrm{d}u = 1$. Moreover, $K(\cdot)$ is Lipschitz continuous on $[-1,1]$. As $(T,h) \to (\infty, 0)$, $Th\to \infty$.
\end{assumption}

Assumption \ref{Ass3}.1 ensures the positive definiteness of the covariance matrix of the likelihood function, and is widely adopted when studying the multivariate time series (e.g., page 2736 of \citealp{bardet2009asymptotic}). In fact, the validity of this assumption is easy to justify in view of \eqref{Eq4.32} and \eqref{Eq4.52} for Example 1 and Example 2 below. Assumption \ref{Ass3}.2 imposes an standard identification condition in the literature of M-estimation (e.g., Proposition 3.4 of \citealp{jeantheau1998strong}). It is noteworthy that the current form of Assumption \ref{Ass3} accommodates the flexibility of the model \eqref{Eq2.1}, which is in fact unnecessary if we have a detailed model in practice. See Section \ref{Sec2.4} for example. 

Assumption \ref{Ass4} imposes the Lipschitz-type conditions on the first and second order derivatives of $\bm{\mu}(\cdot)$ and $\mathbf{H}(\cdot)$ to ensure the smoothness of their functional components.

Assumption \ref{Ass5} is a set of regular conditions on the kernel function and the bandwidth. 

\medskip

With these conditions in hand, we summarize the first theorem of this paper below.

\begin{theorem}\label{Thm3.1}
Suppose Assumptions \ref{Ass1}--\ref{Ass5} hold with $r\geq 6$.  

(1). If $Th^7 \to 0$, then for any $\tau \in (0,1)$
\begin{eqnarray*}
\sqrt{Th}\left(\widehat{\bm{\theta}}(\tau) - \bm{\theta}(\tau) - \frac{1}{2}h^2\widetilde{c}_2\bm{\theta}^{(2)}(\tau)\right) \to_D N\left(\mathbf{0}, \widetilde{v}_0 \bm{\Sigma}_{\bm{\theta}}(\tau) \right),
\end{eqnarray*}
where $\bm{\Sigma}_{\bm{\theta}}(\tau) = \bm{\Sigma}^{-1}(\tau)\bm{\Omega}(\tau)\bm{\Sigma}^{-1}(\tau)$, $\bm{\Sigma}(\tau)=E\left(\gradient_{\bm{\vartheta}}^2\mathcal{l}(\widetilde{\mathbf{x}}_1(\tau),\widetilde{\mathbf{z}}_{0}(\tau);\bm{\theta}(\tau)) \right)$ and 
\begin{eqnarray*}
\bm{\Omega}(\tau)=E\left(\gradient_{\bm{\vartheta}}\mathcal{l}(\widetilde{\mathbf{x}}_1(\tau),\widetilde{\mathbf{z}}_0(\tau);\bm{\theta}(\tau))\cdot \gradient_{\bm{\vartheta}}\mathcal{l}(\widetilde{\mathbf{x}}_1(\tau),\widetilde{\mathbf{z}}_0(\tau);\bm{\theta}(\tau))^\top\right).
\end{eqnarray*}

(2). In addition,  if  $\bm{\varepsilon}_t$ is normally distributed, we have $\bm{\Omega}(\tau)=-\bm{\Sigma}(\tau)$ and thus $\bm{\Sigma}_{\bm{\theta}}(\tau) = \bm{\Omega}^{-1}(\tau)$.
\end{theorem}

After deriving the asymptotic distribution, we will establish both the point-wise inference and the simultaneous inference in the following.

\subsection{Inference}\label{Sec2.3}

In this section, we first discuss how to conduct point-wise inference, and then move on to derive the asymptotic results associated with the simultaneous inference. Specifically, for some preassigned significance level $\alpha \in (0,1)$, we shall construct a $100(1-\alpha)\%$ asymptotic simultaneous confidence band (SCB) $\{ \Upsilon(\tau), 0\leq \tau\leq 1 \}$ for $\bm{\theta}(\cdot)$ in the sense that 
$$
\lim_{T\to\infty} \Pr\left(\bm{\theta}(\tau) \in  \Upsilon(\tau), 0\leq \tau\leq 1\right)  =1-\alpha.
$$
Notably, the simultaneous inference nests the traditional constancy test as a special case. It does not only allow one to examine whether a time-varying model should be preferred to its parametric counterpart, but also allows one to test any particular functional form of interest. For example, if a horizontal line can be embedded in the SCB $\{ \Upsilon(\tau)\}$, then we accept the hypothesis that some elements of $\bm{\theta}(\tau)$ are constant. 

\medskip

\noindent \textbf{Point-wise Inference:} First, we construct a bias-corrected estimator in order to remove the asymptotic bias of Theorem \ref{Thm3.1}. Specifically, we let
\begin{equation}\label{Eq3.1}
	\widetilde{\bm{\theta}} (\tau) = 2 \widehat{\bm{\theta}}_{h/\sqrt{2}}(\tau)-\widehat{\bm{\theta}} (\tau),
\end{equation}
where $\widehat{\bm{\theta}}_{h/\sqrt{2}}(\tau)$ is defined in the same way as $\widehat{\bm{\theta}} (\tau)$ but using the bandwidth $h/\sqrt{2}$.

After tedious development (Lemma B.7 of Appendix B), we have uniformly over $\tau \in[h,1-h]$
\begin{eqnarray*}
	\widetilde{\bm{\theta}}(\tau) - \bm{\theta}(\tau) &=& -\bm{\Sigma}^{-1}(\tau)\frac{1}{Th}\sum_{t=1}^{T}\widetilde{K}((\tau_t-\tau)/h) \gradient_{\bm{\vartheta}}\mathcal{l}(\widetilde{\bm{x}}_t(\tau_t),\widetilde{\bm{z}}_{t-1}(\tau_t);\bm{\theta}(\tau_t))\\
	&& +O_P((Th)^{-1/2}h^{3/2}(\log T)^{1/2}) + o(h^3),
\end{eqnarray*}
where $\widetilde{K}(x)=2\sqrt{2}K(\sqrt{2}x)-K(x)$ that is essentially a fourth-order kernel. It then infers that under the conditions of Theorem \ref{Thm3.1},
\begin{eqnarray}
	\sqrt{Th}(\widetilde{\bm{\theta}} (\tau) - \bm{\theta}(\tau) ) \to_D N\left(\mathbf{0}, {v}_0 \bm{\Sigma}_{\bm{\theta}}(\tau) \right),
\end{eqnarray}
where $v_0 = \int_{-1}^{1}\widetilde{K}^2(u)\mathrm{d}u$.

It is noteworthy that the construction of \eqref{Eq3.1} is different from directly using the fourth-order kernel in the regression. In terms of bandwidth selection, the traditional methods (e.g., cross-validation)  still remain valid for \eqref{Eq3.1} (\citealp{richter2019cross}). However, if one directly employs the fourth-order kernel in the regression, it remains unclear how to select the optimal bandwidth in practice.
 
\medskip

Now we discuss how to estimate $\bm{\Sigma}_{\bm{\theta}}(\tau) $ which is constructed by $\bm{\Sigma}(\tau)$ and $\bm{\Omega}(\tau)$. Intuitively, we consider the following estimator 
\begin{equation}\label{Eq3.4}
\widehat{\bm{\Sigma}}_{\bm{\theta}}(\tau) = \widehat{\bm{\Sigma}}^{-1}(\tau) \widehat{\bm{\Omega}}(\tau) \widehat{\bm{\Sigma}}^{-1}(\tau),
\end{equation}
where
\begin{eqnarray*}
\widehat{\bm{\Sigma}}(\tau) &=&A_T(\tau)^{-1} \sum_{t=1}^{T}\gradient_{\bm{\vartheta}}^2\mathcal{l}(\mathbf{x}_t,\mathbf{z}_{t-1}^c;\widehat{\bm{\theta}}(\tau))K_h(\tau_t-\tau),\\
	\widehat{\bm{\Omega}}(\tau) &=& A_T(\tau)^{-1} \sum_{t=1}^{T}\gradient_{\bm{\vartheta}}\mathcal{l}(\mathbf{x}_t,\mathbf{z}_{t-1}^c;\widehat{\bm{\theta}}(\tau)) \cdot \gradient_{\bm{\vartheta}}\mathcal{l}(\mathbf{x}_t,\mathbf{z}_{t-1}^c;\widehat{\bm{\theta}}(\tau))^\top K_h(\tau_t-\tau),\\
	A_T(\tau)&=&\sum_{t=1}^{T}K_h(\tau_t-\tau).
\end{eqnarray*}
Note that we consider a local constant estimator in \eqref{Eq3.4} rather than a local linear one, that is to avoid an implementation issue for finite sample studies (i.e., nonpositive definite covariance may occur when the local linear approach is employed). Such a numerical problem has been well explained and investigated in the literature. See \cite{chen2015local} for example.
 
The following corollary summarizes the asymptotic property of \eqref{Eq3.4}.

\begin{corollary}\label{proposition3.1}
Under the conditions of Theorem \ref{Thm3.1}.1, suppose further that 
\begin{eqnarray*}
\sup_{\tau \in [0,1]} [\alpha_j(\bm{\theta}(\tau))+\beta_j(\bm{\theta}(\tau))] = O(j^{-(5/2+s)})
\end{eqnarray*}
for some $s > 0$. In addition, let $h (\log T)^2\to 0$ and $T^{1-6/r}h\to \infty$. Then
\begin{eqnarray*}
\sup_{\tau \in [0,1]}|\widehat{\bm{\Sigma}}_{\bm{\theta}}(\tau) -\bm{\Sigma}_{\bm{\theta}}(\tau)|=o_P(1).
\end{eqnarray*}
\end{corollary}

\medskip

\noindent \textbf{Simultaneous Inference:} We now consider the simultaneous inference.  To allow for flexibility, we first introduce a selection matrix $\mathbf{C}$ with full row rank, which selects the parameters of interest as follows:
\begin{eqnarray}
\bm{\theta}_{\mathbf{C}}(\tau):=\mathbf{C}\bm{\theta}(\tau).
\end{eqnarray}
Accordingly, the estimator and the corresponding asymptotic covariance matrix become
\begin{eqnarray}
\widehat{\bm{\theta}}_{\mathbf{C}}(\tau):=\mathbf{C}\widehat{\bm{\theta}} (\tau)  \quad \text{and}\quad \bm{\Sigma}_{\mathbf{C}}(\tau) = \mathbf{C}\bm{\Sigma}_{\bm{\theta}}(\tau)\mathbf{C}^\top.
\end{eqnarray}

\begin{theorem}\label{Thm3.2}
Under the conditions of Theorem \ref{Thm3.1}.1, suppose further that
\begin{eqnarray*}
\sup_{\tau \in [0,1]} [\alpha_j(\bm{\theta}(\tau)) + \beta_j(\bm{\theta}(\tau))]= O(j^{-(3+s)})
\end{eqnarray*}
for some $s > 0$. In addition, let $(\log T)^4/(T^{\nu}h)\to 0$ with $\nu = \frac{1}{2}-\frac{r-6}{4rs/3+2r-4}$ and $Th^7\log T \to 0$. Then
\begin{eqnarray*}
\lim_{T\to\infty}\mathrm{Pr}&&\left(\sqrt{\frac{Th}{\widetilde{v}_0}}\sup_{\tau\in[h,1-h]}\left|\bm{\Sigma}_{\mathbf{C}}^{-1/2}(\tau)\left\{\widehat{\bm{\theta}}_{\mathbf{C}}(\tau) - \bm{\theta}_{\mathbf{C}}(\tau) - \frac{1}{2}h^2\widetilde{c}_2\bm{\theta}_{\mathbf{C}}^{(2)}(\tau) \right\} \right| \right.\\
&& \left.- B(1/h)\leq\frac{u}{\sqrt{2\log(1/h)}} \right) = \exp(-2\exp(-u)),
\end{eqnarray*}
where 
\begin{eqnarray*}
B(1/h) &=& \sqrt{2 \log(1/h)} + \frac{\log (C_K) + (k/2-1/2)\log(\log(1/h))-\log(2)}{\sqrt{2 \log(1/h)}},\\
C_K &=& \frac{\{\int_{-1}^{1}|K^{(1)}(u)|^2\mathrm{d}u/\widetilde{v}_0\pi\}^{1/2}}{\Gamma(k/2)},
\end{eqnarray*}
and $\Gamma(\cdot)$ is the Gamma function.
\end{theorem}

In Theorem \ref{Thm3.2}, $\nu$ is slightly smaller than $1/2$ as we only require $r$ to be slightly larger than $6$. Hence, the usual optimal bandwidth $h_{opt} = O(T^{-1/5})$ satisfies the conditions $(\log T)^4/(T^{\nu}h)\to 0$ and $Th^7\log T \to 0$.

As shown in Theorem \ref{Thm3.2}, the convergence rate of the simultaneous confidence intervals for $\bm{\theta}_{\mathbf{C}}(\cdot)$ is of logarithmic rate and is therefore slow. In order to improve the rate, we consider a bootstrap method which shows a much better finite sample performance. We summarize the result in the following corollary.

\begin{corollary}\label{proposition3.4}
Under the conditions of Theorem \ref{Thm3.2}. Suppose that $h = O(T^{-\kappa})$ with $1/7<\kappa<\nu$. Then, on a richer probability space, there exists i.i.d. $k$-dimensional standard normal variables $\mathbf{v}_1,\ldots,\mathbf{v}_T$ such that
\begin{equation*} 
\sup_{\tau\in[0,1]}|\widehat{\bm{\theta}}_{\mathbf{C}}(\tau)-\bm{\theta}_{\mathbf{C}}(\tau) -\frac{1}{2}h^2b_h(\tau)\bm{\theta}_{\mathbf{C}}^{(2)}(\tau) -\bm{\Sigma}_{\mathbf{C}}^{1/2}(\tau)\mathbf{V}_h^*(\tau)| =O_P\left(\frac{T^{-\alpha}}{\sqrt{Th\log T}} \right),
\end{equation*}
where $\alpha = \min\{(\nu-\kappa)/2,(7\kappa-1)/2,\kappa/2\}$, $\widetilde{c}_{k,h}(\tau) = \int_{-\tau/h}^{(1-\tau)/h} u^k K(u) \mathrm{d}u $,  $\mathbf{V}_h^*(\tau) = T^{-1}\sum_{t=1}^{T}\mathbf{v}_t\omega_{t,h}(\tau)$,
$$
b_h(\tau) = \frac{\widetilde{c}_{2,h}^2(\tau) - \widetilde{c}_{1,h}(\tau)\widetilde{c}_{3,h}(\tau)}{ \widetilde{c}_{0,h}(\tau)\widetilde{c}_{2,h}(\tau)-\widetilde{c}_{1,h}^2(\tau)}\quad \text{and}\quad \omega_{t,h}(\tau) = K_h(\tau_t-\tau)\frac{ \widetilde{c}_{2,h}(\tau)- \frac{\tau_t-\tau}{h}\widetilde{c}_{1,h}(\tau)}{\widetilde{c}_{0,h}(\tau)\widetilde{c}_{2,h}(\tau)-\widetilde{c}_{1,h}^2(\tau)}.
$$
\end{corollary}

By Corollary \ref{proposition3.4}, we propose the following numerical procedure to construct the SCB of $\bm{\theta}_{\mathbf{C}}(\tau)$:

\begin{itemize}		
\item[Step 1] Use the sample $\{\mathbf{x}_t\}_{t=1}^T$ to estimate $\widehat{\bm{\theta}}_{\mathbf{C}}(\tau)$ by \eqref{Eq2.4}, and compute $\widetilde{\bm{\theta}}_{\mathbf{C}}(\tau)$ based on \eqref{Eq3.1}.
	
\item[Step 2] Generate i.i.d. $k$-dimensional standard normal variables $\{\mathbf{v}_t^*\}$ and calculate the quantity $\sup_{\tau \in [0,1]}|\mathbf{V}_{h}^*(\tau)|$, in which $\mathbf{V}_{h}^*(\tau)= T^{-1}\sum_{t=1}^{T}\mathbf{v}_t^*(2\omega_{t,h/\sqrt{2}}(\tau)-\omega_{t,h}(\tau))$.
	
\item[Step 3] Repeat Step 2 $R$ times to obtain the empirical $(1-\alpha)^{th}$ quantile $\widehat{q}_{1-\alpha}$ of $\sup_{\tau \in [0,1]}|\mathbf{V}_{h}^*(\tau)|$.
	
\item[Step 4] Calculate $\widehat{\bm{\Sigma}}_{\mathbf{C}}(\tau)$ using \eqref{Eq3.4}, and construct the SCB of $\bm{\theta}_{\mathbf{C}}(\tau)$ by $\widetilde{\bm{\theta}}_{\mathbf{C}}(\tau) + \widehat{\bm{\Sigma}}_{\mathbf{C}}^{1/2}(\tau) \widehat{q}_{1-\alpha} \mathbb{B}_k$, where $\mathbb{B}_k = \{\mathbf{u}\in \mathbb{R}^k:|\mathbf{u}|\leq 1\}$ is the unit ball, and $k$ is the rank of $\mathbf{C}$.
\end{itemize}

\subsection{Examples} \label{Sec2.4}

Below, we demonstrate the usefulness of the aforementioned results by considering Example 1 and Example 2 of Section \ref{Sec1}. 

\smallskip

\noindent \textbf{Example 1 (Cont.)} --- For $\forall\tau\in[0,1]$, simple algebra shows that the approximated stationary process is defined by
\begin{eqnarray}\label{Eq4.3}
\widetilde{\bm{x}}_t(\tau) = \bm{\mu}\left(\widetilde{\mathbf{x}}_{t-1}(\tau),\widetilde{\mathbf{x}}_{t-2}(\tau),\ldots;\bm{\theta}(\tau)\right) +\mathbf{H}\left(\widetilde{\mathbf{x}}_{t-1}(\tau),\widetilde{\mathbf{x}}_{t-2}(\tau),\ldots;\bm{\theta}(\tau)\right)\bm{\varepsilon}_t,
\end{eqnarray}
where $\bm{\theta}(\tau)$ has been defined in \eqref{Eq4.1.1}, and
\begin{eqnarray}\label{Eq4.32}
 \bm{\mu}\left(\widetilde{\mathbf{x}}_{t-1}(\tau),\widetilde{\mathbf{x}}_{t-2}(\tau),\ldots;\bm{\theta}(\tau)\right) &=&  \mathbf{B}_{\tau}^{-1}(1)\mathbf{a}(\tau)+\sum_{j=1}^{\infty}\bm{\Gamma}_j(\tau)\widetilde{\bm{x}}_{t-j}(\tau),\nonumber \\
 \mathbf{H}\left(\widetilde{\mathbf{x}}_{t-1}(\tau),\widetilde{\mathbf{x}}_{t-2}(\tau),\ldots;\bm{\theta}(\tau)\right) &=& \bm{\omega}(\tau).
\end{eqnarray}
Additionally, in \eqref{Eq4.32}, $\bm{\Gamma}_j(\tau)$ is yielded as follows:
\begin{eqnarray}
\mathbf{I}_m - \sum_{j=1}^{\infty}\bm{\Gamma}_j(\tau)L^j = \mathbf{B}_{\tau}^{-1}(L)\mathbf{A}_{\tau}(L),
\end{eqnarray}
where $\mathbf{A}_{\tau}(L):=\mathbf{I}_m - \mathbf{A}_1(\tau)L-\cdots-\mathbf{A}_p(\tau)L^p$
and $\mathbf{B}_{\tau}(L):=\mathbf{I}_m + \mathbf{B}_1(\tau)L+\cdots+\mathbf{B}_q(\tau)L^q$.

Then we are able to present the following proposition.

\begin{proposition}\label{proposition4.1}
Let $\|\bm{\varepsilon}_t\|_r < \infty$ for some $r>4$. Suppose that there is a compact set
\begin{eqnarray*}
\bm{\Theta} =\{\bm{\vartheta} = \mathrm{vec}(\mathbf{a} ,\mathbf{A}_1 ,\ldots,\mathbf{A}_p ,\mathbf{B}_1 ,\ldots,\mathbf{B}_q ,\bm{\Omega} ) \mid \bm{\vartheta} \in \mathbb{R}^d\}
\end{eqnarray*}
 such that (1). for $\forall \tau\in [0,1]$, $\bm{\theta}(\tau)$ lies in the interior of $\bm{\Theta}$, (2). $\mathrm{det}(\mathbf{A}(L)\mathbf{B}(L)) \neq 0$ for all $|L|\leq 1$, (3). $\bm{\Omega}>0$, where $\mathbf{A}(L):=\mathbf{I}_m - \mathbf{A}_1L-\cdots-\mathbf{A}_pL^p$ and $\mathbf{B}(L):=\mathbf{I}_m + \mathbf{B}_1L+\cdots+\mathbf{B}_qL^q$ are coprime and satisfy some necessary identification conditions. Then, the results of Theorems \ref{Thm3.1} and \ref{Thm3.2} hold for model \eqref{Eq4.1}.
\end{proposition}

We note that the detailed identification conditions required for VARMA processes (e.g., the final equations form or echelon form) can be found in \cite{lutkepohl2005new}. We no longer discuss them here in order not to derivative from our main goal.

\medskip

\noindent \textbf{Example 2 (Cont.)} ---  We further let
\begin{eqnarray}\label{GARCH_rho}
\bm{\Omega}(\tau) = \left[\begin{matrix}
1 & \rho_{1,2}(\tau) & \cdots &\rho_{1,m}(\tau) \\
\rho_{1,2}(\tau) & 1 & \ddots &\vdots \\
\vdots & \ddots & \ddots &\rho_{m-1,m}(\tau) \\
\rho_{1,m}(\tau) & \rho_{m-1,m}(\tau)  & \ddots &1 \\
\end{matrix} \right].
\end{eqnarray}

For $\forall\tau\in[0,1]$, the corresponding approximated stationary process is defined as
\begin{eqnarray}\label{Eq4.5}
\widetilde{\bm{x}}_t(\tau) = \mathbf{H}\left(\widetilde{\mathbf{x}}_{t-1}(\tau),\widetilde{\mathbf{x}}_{t-2}(\tau),\ldots;\bm{\theta}(\tau)\right)\bm{\varepsilon}_t,
\end{eqnarray}
where 
\begin{eqnarray}\label{Eq4.52}
&&\mathbf{H}\left(\widetilde{\mathbf{x}}_{t-1}(\tau),\widetilde{\mathbf{x}}_{t-2}(\tau),\ldots;\bm{\theta}(\tau)\right)\nonumber \\
& =& \mathrm{diag}^{1/2}\left( \mathbf{D}_{\tau}^{-1}(1)\mathbf{c}_0(\tau)+ \sum_{j=1}^{\infty} \mathbf{\Psi}_j(\tau) \left(\widetilde{\mathbf{x}}_{t-j}(\tau)\odot\widetilde{\mathbf{x}}_{t-j}(\tau)\right)\right).
\end{eqnarray}
Note that $ \mathbf{\Psi}_j(\tau)$ is generated as follows:
\begin{eqnarray}
\bm{\Psi}_{\tau}(L):=\mathbf{I}_m - \sum_{j=1}^{\infty}\bm{\Psi}_j(\tau)L^j=\mathbf{D}_{\tau}^{-1}(L)\mathbf{C}_{\tau}(L),
\end{eqnarray}
where $\mathbf{C}_{\tau}(L):=\mathbf{C}_1(\tau)L+\cdots+\mathbf{C}_p(\tau)L^p$ and $\mathbf{D}_{\tau}(L):=\mathbf{I}_m - \mathbf{D}_1(\tau)L-\cdots-\mathbf{D}_q(\tau)L^q$.
 
Consequently, we can present the following proposition.

\begin{proposition}\label{proposition4.2}
Suppose that there is a compact set
\begin{eqnarray*}
\bm{\Theta} =\{\bm{\vartheta} = \mathrm{vec}(\mathbf{c}_0,\mathbf{C}_1,\ldots,\mathbf{C}_p ,\mathbf{D}_1,\ldots,\mathbf{D}_q,\bm{\Omega}) \mid \bm{\vartheta} \in \mathbb{R}^d\}
\end{eqnarray*}
 such that (1). for $\forall \tau\in [0,1]$, $\bm{\theta}(\tau)$ lies in the interior of $\bm{\Theta}$, (2). $\|\bm{\Omega}^{1/2}\bm{\varepsilon}_t\|_r^2 \sum_{j=1}^\infty|\mathbf{\Psi}_j|< 1$ for some $r >  6$, (3). all the roots of $|\mathbf{I}_m - \sum_{j=1}^{p} \mathbf{C}_j - \sum_{j=1}^{q} \mathbf{D}_j|$ are outside the unit circle with $\mathbf{C}_j$'s and $\mathbf{D}_j$'s being squared matrices of nonnegative elements, (4). $\mathbf{c}_0$ is a vector of positive elements, (5). $\mathbf{C}(L)$ and $\mathbf{D}(L)$ are coprime and the formulation of the GARCH part is minimal, where $\mathbf{C}(L):=\mathbf{C}_1 L+\cdots+\mathbf{C}_p L^p$ and $\mathbf{D} (L):=\mathbf{I}_m - \mathbf{D}_1 L-\cdots-\mathbf{D}_q L^q$. Then the results Theorems \ref{Thm3.1} and \ref{Thm3.2} hold for model \eqref{Eq4.4}.
\end{proposition}

For the identification conditions of the GARCH process, we refer readers to Proposition 3.4 of \cite{jeantheau1998strong}, who proves that assuming the minimal representation is enough for ensuring Assumption \ref{Ass3} holds.

In the following section, we conduct numerical studies using both simulated and real data to evaluate the finite-sample performance of the proposed estimation and inferential methods.

\section{Numerical Studies}\label{Sec3}

In this section, we first present the details of the numerical implementations in Section \ref{Sec3.1}, and then conduct extensive simulations in Section \ref{Sec3.2}. Section \ref{Sec3.3} presents a real data example on the conditional correlations between the Chinese and U.S. stock markets.

\subsection{Numerical Implementation}\label{Sec3.1}
Throughout the numerical studies, the Epanechnikov kernel $K(u) = 0.75(1-u^2)I(|u|\leq1)$ is adopted. Following \cite{zhou2010simultaneous}, we use $\widetilde{h} = 2\widehat{h}$ for the biased corrected estimator, where $\widehat{h}$ is the bandwidth selected by the cross-validation method of \cite{richter2019cross}. 

Specifically, define the leave-one-out local linear QMLE
\begin{equation}\label{Eq5.1}
	(\widehat{\bm{\theta}}_{h,-t}(\tau),h\widehat{\bm{\theta}}_{h,-t}^{(1)}(\tau)) = \argmax_{(\bm{\eta}_1,\bm{\eta}_2)\in\mathbf{E}_T(r)}\mathcal{L}_{T,-t}^c(\tau,\bm{\eta}_1,\bm{\eta}_2),
\end{equation}
where 
$$
\mathcal{L}_{T,-t}^c(\tau,\bm{\eta}_1,\bm{\eta}_2)=\frac{1}{T}\sum_{s=1,\neq t}^{T}\mathcal{l}(\mathbf{x}_s,\mathbf{z}_{s-1}^c;\bm{\eta}_1+\bm{\eta}_2\cdot(\tau_s-\tau)/h)K_h(\tau_s-\tau).
$$
Then, the bandwidth is chosen by 
\begin{equation}\label{Eq5.2}
	\widehat{h} = \argmax_{h}T^{-1}\sum_{t=1}^{T}\mathcal{l}(\mathbf{x}_t,\mathbf{z}_{t-1}^c;\widehat{\bm{\theta}}_{h,-t}(\tau_t)).
\end{equation}
As shown in \cite{richter2019cross}, this cross validation method works well as long as $\gradient\mathcal{l}$ is uncorrelated, which implies that this desirable property should hold in our case.

Notably, when considering some specific models, the implementation may be further simplified. We provide more discussions along this line in Appendix B.4.

\subsection{Simulation Results}\label{Sec3.2}

In the simulation studies, we examine the empirical coverage probabilities of simultaneous confidence intervals for nominal levels $\alpha =90\%,\  95\%$. We consider the time-varying VARMA($2,1$) and multivariate GARCH($1,1$) model as follows:
\begin{enumerate}
\item $\text{DGP 1}: \mathbf{x}_t = a_1(\tau_t)\mathbf{x}_{t-1}+a_2(\tau_t)\mathbf{x}_{t-2} + \bm{\eta}_t + \mathbf{B}_1(\tau_t)\bm{\eta}_{t-1},\quad \bm{\eta}_t = \bm{\omega}(\tau_t) \bm{\varepsilon}_t$, where $\{\bm{\varepsilon}_t\}$ are i.i.d. draws from $N(\mathbf{0}_{2\times 1},\mathbf{I}_2)$, $a_1(\tau) = 0.6\exp(\tau-1)$, $a_2(\tau) = -0.3\exp(\tau-1)$,
\begin{eqnarray*}
	\mathbf{B}_1(\tau)&=&\left[\begin{matrix}
		0.5\exp{\tau-0.5} & -0.8(\tau-0.5)^2 \\
		-0.8(\tau-0.5)^2 & 0.5+0.3\sin(\pi \tau)
	\end{matrix} \right], \nonumber \\
	\bm{\omega}(\tau)&=&\left[\begin{matrix}
		1.5+0.2\exp{0.5-\tau}& 0 \\
		0.2\exp{0.5-\tau}   & 1.5+0.5(\tau-0.5)^2
	\end{matrix}\right].
\end{eqnarray*}
Here we use final equations form to ensure the uniqueness of the VARMA representation.

\item $\text{DGP 2}: \mathbf{x}_t = \mathrm{diag}(h_{1,t}^{1/2},\ldots,h_{m,t}^{1/2}) \bm{\eta}_t$, where $\bm{\eta}_t = \bm{\Omega}^{1/2}(\tau_t) \bm{\varepsilon}_t$, $\mathbf{h}_{t} = \mathbf{c}_0(\tau_t) + \mathbf{C}_1(\tau_t) \left(\mathbf{x}_{t-1}\odot\mathbf{x}_{t-1}\right) + \mathbf{D}_1(\tau_t)\mathbf{h}_{t-1}$, $\{\bm{\varepsilon}_t\}$ are i.i.d. draws from $N(\mathbf{0}_{2\times 1},\mathbf{I}_2)$, $\mathbf{c}_0(\tau) = [2\exp{0.5\tau-0.5}, 3+0.2\cos(\tau)]^\top$,
\begin{eqnarray*}
\mathbf{C}_1(\tau)&=&\left[\begin{matrix}
		0.4 + 0.05\cos(\tau) & 0.05(\tau-0.5)^2 \\
		0.05(\tau-0.5)^2 & 0.4+0.05\sin(\tau)
\end{matrix} \right], \nonumber \\
\mathbf{D}_1(\tau)&=&\left[\begin{matrix}
	0.4 - 0.1\cos(\tau) & 0 \\
	0 & 0.3 - 0.1\sin(\tau)
\end{matrix} \right], \nonumber \\
\bm{\Omega}(\tau)&=&\left[\begin{matrix}
		1 & 0.3\sin(\tau) \\
		0.3\sin(\tau)   & 1 
\end{matrix}\right].
\end{eqnarray*}
\end{enumerate}

Let the sample size be $T \in\{500,1000\}$ ($T \in\{1000,2000,4000\}$) for the VARMA model (the GARCH model). We conduct $1000$ replications for each choice of $T$. Several different bandwidths close to $\widetilde{h}$ are reported to check the sensitivity of bandwidth selection. 

We present the empirical coverage probabilities associated with the SCB
in Tables \ref{table_sim1}--\ref{table_sim2}. For the vector- or matrix-valued unknown coefficients, we take an average across the elements. A few facts emerge from the tables. First, the finite sample coverage probabilities are smaller than their nominal level when $T = 500$ ($T = 1000,2000$) for the VARMA model (the GARCH model), but are fairly close to their nominal level as $T = 1000$ ($T=4000$) for the VARMA model (the GARCH model). Second, the behaviour of the estimated simultaneous confidence intervals is not sensitive to the choices of bandwidths. Third, the GARCH model requires more data to reach a reasonable finite sample performance.

\begin{table}[h]
\caption{Empirical Coverage Probabilities of the SCB  for DGP 1}\label{table_sim1}
\begin{center}
\begin{tabular}{c c c cccc c cccc}
\hline
& & &\multicolumn{4}{c}{$90\%$}& &\multicolumn{4}{c}{$95\%$}\\
\cline{4-7} \cline{9-12}
&\text{$\widetilde{h}$}& &$\alpha_1(\cdot)$ &$\alpha_2(\cdot)$ &$\mathbf{B}_1(\cdot)$ & $\bm{\Omega}(\cdot)$ &   &$\alpha_1(\cdot)$ &$\alpha_2(\cdot)$ &$\mathbf{B}_1(\cdot)$ & $\bm{\Omega}(\cdot)$\\
\hline
\multirow{4}{*}{\shortstack{$T=500$}}
&$0.35$ & &  0.845 & 0.877  & 0.821&0.847 &  &  0.905 & 0.915  & 0.889 &0.905\\
&$0.4$ & &  0.865& 0.875  & 0.847 &0.876 &  &  0.912 & 0.930 & 0.897& 0.909\\
&$0.45$ & & 0.862  & 0.895  & 0.847 &0.878 &  &0.915   &0.930    & 0.898  &0.919\\
&$0.5$ & & 0.875   &0.895    & 0.847 &0.876 &  &0.905    &0.945    & 0.901  &0.920\\
\hline
\multirow{4}{*}{\shortstack{$T=1000$}}
&$0.3$ & & 0.895  & 0.925  & 0.887& 0.884&  & 0.960  & 0.960 & 0.947& 0.947\\
&$0.35$ & & 0.910 & 0.927  & 0.886&0.890&  & 0.940  & 0.967  & 0.940&0.930 \\
&$0.4$ & & 0.917  & 0.939  & 0.901&0.899 &  & 0.947  & 0.959 & 0.948 &0.939\\
&$0.45$ & &0.937  & 0.932   &0.908  &0.895 &  &0.957    &0.957   &0.947 &0.937 \\
\hline
\end{tabular}
\end{center}
\end{table}

\begin{table}[h]
	\caption{Empirical Coverage Probabilities of the SCB  for DGP 2}\label{table_sim2}
	\begin{center}
		\begin{tabular}{c c c cccc c cccc}
			\hline
			& & &\multicolumn{4}{c}{$90\%$}& &\multicolumn{4}{c}{$95\%$}\\
			\cline{4-7} \cline{9-12}
			&\text{$\widetilde{h}$}& &$\mathbf{c}_0(\cdot)$ &$\mathbf{C}_1(\cdot)$ &$\mathbf{D}_1(\cdot)$ & $\bm{\Omega}(\cdot)$ &   &$\mathbf{c}_0(\cdot)$ &$\mathbf{C}_1(\cdot)$ &$\mathbf{D}_1(\cdot)$ & $\bm{\Omega}(\cdot)$\\
			\hline
			\multirow{4}{*}{\shortstack{$T=1000$}}
			&$0.55$ & &  0.802& 0.810  & 0.784 &0.889 &  &  0.869 & 0.876 & 0.838 & 0.945\\
			&$0.60$ & & 0.824  & 0.820  & 0.791 &0.879 &  &0.882   &0.866    & 0.843  &0.945\\
			&$0.65$ & & 0.832   &0.820    & 0.796 &0.874 &  &0.889    &0.872    & 0.859  &0.945\\
			&$0.70$ & & 0.820  &0.823    & 0.792 &0.879 &  &0.892    &0.881    & 0.871  &0.950\\
			\hline
			\multirow{4}{*}{\shortstack{$T=2000$}}
			&$0.50$ & &  0.827 & 0.835  & 0.841 &0.889 &  &  0.897 & 0.881  & 0.901&0.950\\
			&$0.55$ & &  0.829 & 0.825  & 0.843 &0.884 &  &  0.892 & 0.881 & 0.903 & 0.940 \\
			&$0.60$ & & 0.849  & 0.833  & 0.871 &0.900 &  &0.900   &0.888    & 0.910  &0.950\\
			&$0.65$ & & 0.852   &0.835    & 0.873 &0.910 &  &0.907    &0.889    & 0.910  &0.950\\
			\hline
			\multirow{4}{*}{\shortstack{$T=4000$}}
			&$0.35$ & & 0.879 & 0.879  & 0.882 &0.869&  & 0.929  & 0.932  & 0.943 &0.920 \\
			&$0.4$ & & 0.899  & 0.879 & 0.882 &0.859 &  & 0.950  & 0.944 & 0.943 &0.919\\
			&$0.45$ & &0.904 & 0.899   &0.879  &0.838 &  &0.950    &0.947   &0.946 &0.950 \\
			&$0.50$ & &0.867 & 0.857   &0.884  &0.898 &  &0.929    &0.944   &0.946 &0.960 \\
			\hline
		\end{tabular}
	\end{center}
\end{table}

\subsection{A Real Data Example}\label{Sec3.3}

In this subsection, we investigate the time-varying conditional correlations between the Chinese and U.S. stock markets using the time-varying multivariate GARCH model. Recently, there is a growing literature to study the relationship of the two stock markets (e.g., \citealp{zhang2014has,pan2022modeling}), as the Chinese stock market has become the world's second largest stock market after 2009.  Understanding the interactions among different financial markets is important for investors and policymakers \cite[]{diebold2009measuring,bensaida2019good}. For example, high equity market interdependence implies poor diversification benefits from portfolios, but highlights the possibility of better hedging benefits.

Previous research documents a strong positive link between the degree of globalization and equity market interdependence \cite[]{baele2005volatility}. Along this line of research, one important question is that whether the interdependence between the Chinese and U.S. stock markets has increased over time due to globalization so that estimates from historical data are unreliable for modern policy analysis, asset pricing and risk management. The existing results present many discrepancies, which may be due to the fact that the relationship evolves with time. Apparently, the results also indicate that one should use time-varying GARCH model to accommodate potential nonstationarity inherited in these financial variables. In addition, as pointed out by \cite{caporin2013ten}, dynamic conditional correlation (DCC) GARCH model represents the dynamic conditional covariances of the standardized residuals, and hence does not yield dynamic
conditional correlations; DCC yields inconsistent two step estimators; DCC has no asymptotic properties. In what follows, we address these issues using the newly proposed approach. The estimation is conducted in exactly the same way as in Section \ref{Sec3.1}, so we no longer repeat the details.

We calculate the Chinese and U.S. stock returns based on weekly Shanghai Stock Exchange (SSE) Composite Index and S\&P 500 Index as they are the most comprehensive and diversified stock indices. The sample employed in this study spanning from January 2000 to February 2022 provides $1119$ observations\footnote{The data are collected from Yahoo Finance at \url{https://finance.yahoo.com/}.}. Figure \ref{Fg1} plots the two weekly returns as well as sample autocorrelation functions of squared data, which shows the typical ``volatility clustering'' phenomenon.

\begin{figure}[h]
	\centering
	{\includegraphics[width=16cm]{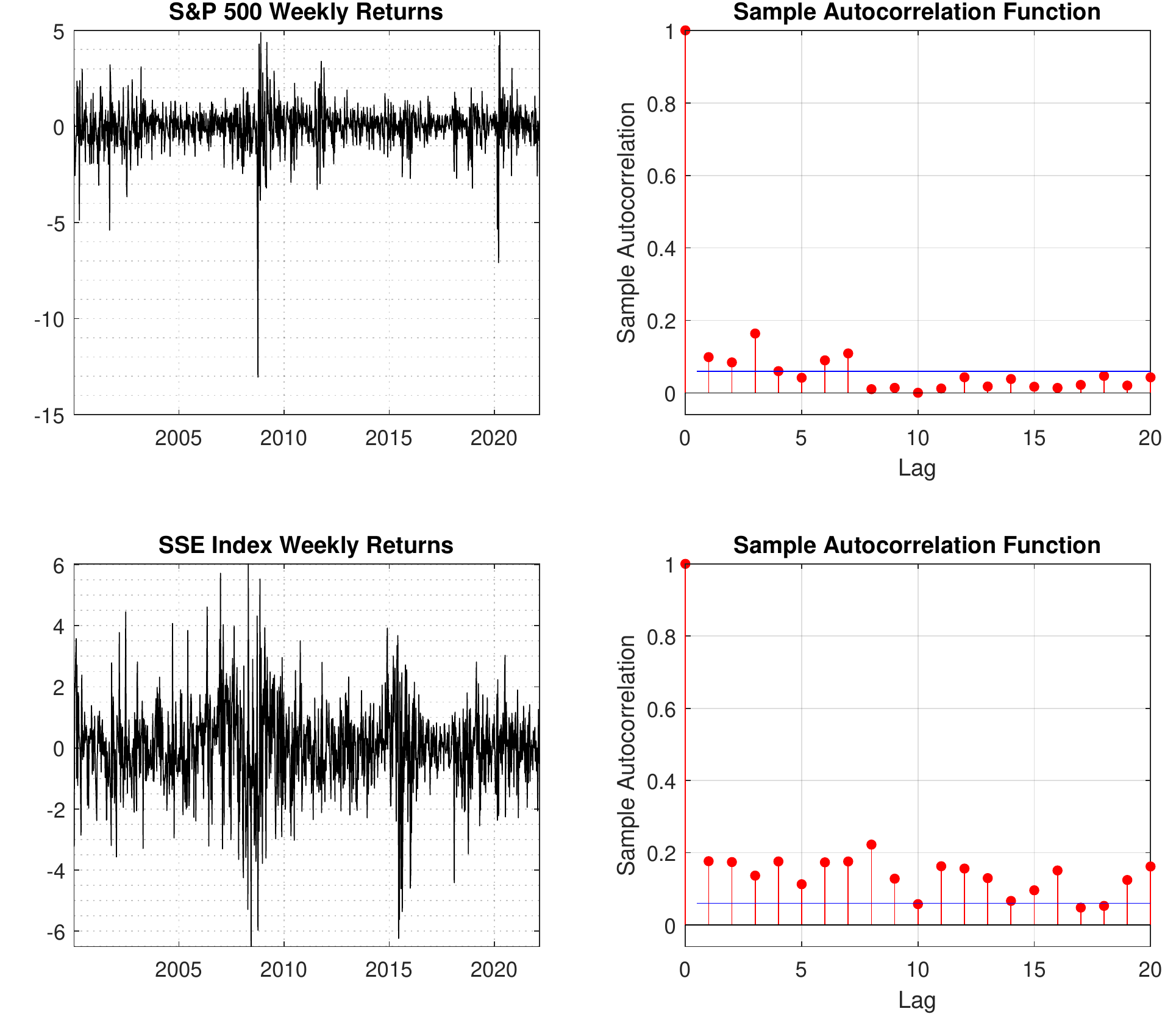}}
	\caption{\small S\&P 500 and SSE Index returns as well as sample autocorrelation functions of squared data}\label{Fg1}
\end{figure}

We next fit the data to a time-varying multivariate GARCH(1,1) model and are particularly interested in the estimates of time-varying conditional correlations, i.e.,
$$
E\left(x_{1,t}x_{2,t}\mid \mathcal{F}_{t-1}\right)/\sqrt{E\left(x_{1,t}^2\mid \mathcal{F}_{t-1}\right)E\left(x_{2,t}^2\mid \mathcal{F}_{t-1}\right)} = \rho_{1,2}(\tau_t),
$$
where $\rho_{1,2}(\cdot)$ is defined in \eqref{GARCH_rho}. Figure \ref{Fg2} plots the estimates (black solid line) of time-varying conditional correlations between the two stock markets as well as 95\% simultaneous confidence intervals (red dashed line) and 95\% pointwise confidence intervals (black dashed line). Based on the simultaneous confidence intervals, apparently, the conditional correlations vary with respect to time. Moreover, as clearly presented in Figure \ref{Fg2}, the interdependence between the two stock markets is increasing over time. By examining  the pointwise confidence intervals, we can conclude that the two stock markets are not significantly correlated before 2005, but the relationship has been greatly enhanced in recent years. These results have important implications for investment and risk management. For example, it implies that the Chinese and U.S. investors who use cross-country portfolio strategies to eliminate country specific risks may be benefit from hedging. However, all types of investors should be cautious since the relations between the Chinese and U.S. stock markets are time-varying.

\begin{figure}[h]
	\centering
	{\includegraphics[width=14cm]{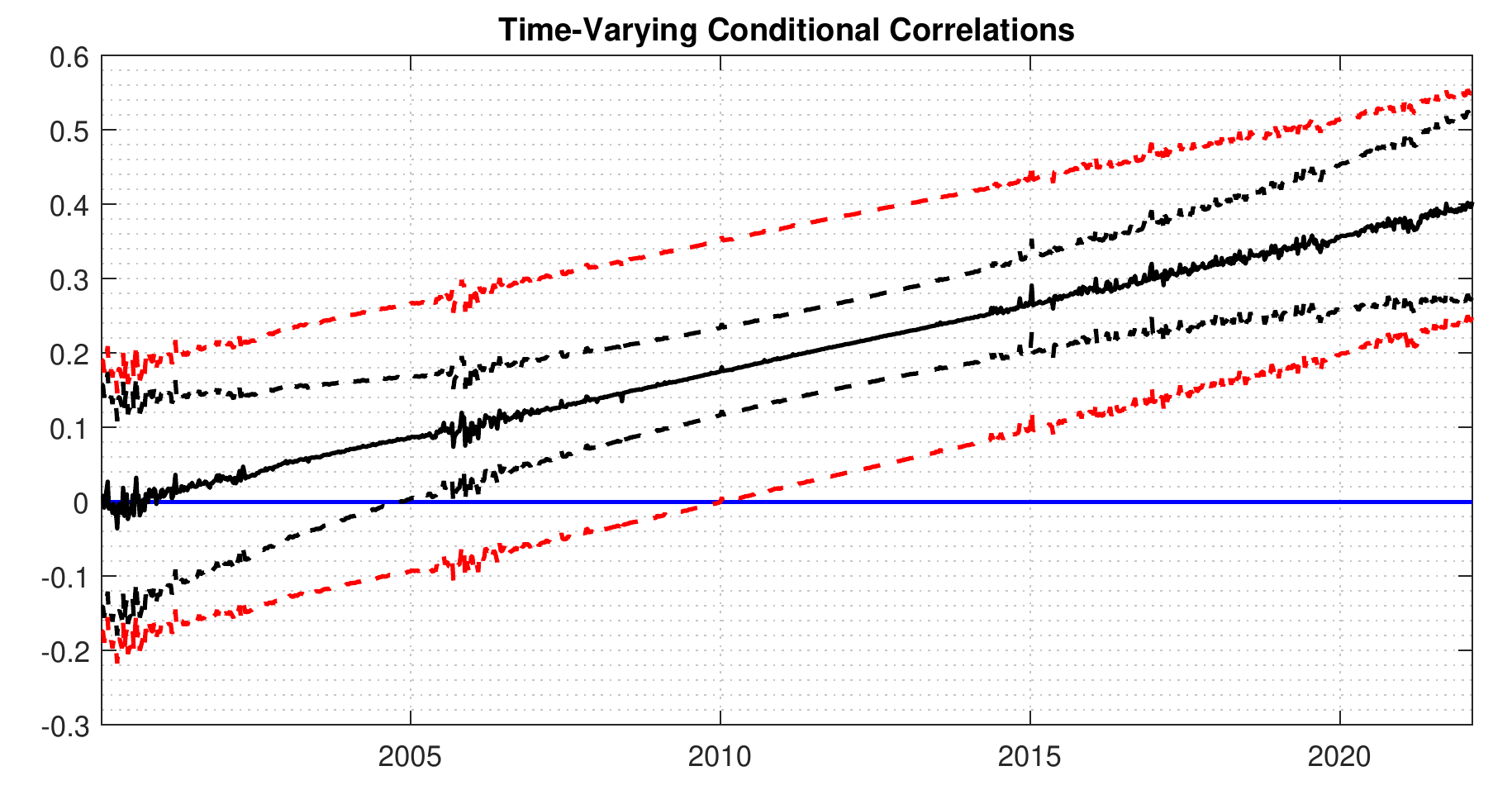}}
	\caption{\small Time-varying conditional correlations between the Chinese and U.S. stock markets}\label{Fg2}
\end{figure}

\section{Conclusions}\label{Sec4}

In this paper,  we consider a wide class of time-varying multivariate causal processes which nests many classic and new examples as special cases. We first  prove the existence of a weakly dependent stationary approximation for the model \eqref{Eq2.1} which is the foundation to establish the corresponding asymptotic properties. Afterwards, we consider the QMLE estimation approach, and provide both point-wise and simultaneous inferences on the coefficient functions. In addition, we demonstrate the theoretical findings through both simulated and real data examples. In particular, we show the empirical relevance of our study using an application to evaluate the conditional correlations between the stock markets of China and U.S. We find that the interdependence between the two stock markets is increasing over time.

There are several directions for possible extensions. The first one is to consider quantile regression methods for such locally stationary multivariate causal processes. The second one is to propose a more powerful $L_2$ test based on the weighted integrated squared errors for testing whether some coefficients are time-invariant. We wish to leave such issues for future study.

\section{Acknowledgements}

The authors of this paper would like to thank George Athanasopoulos, David Frazier and Gael Martin for their constructive comments on earlier versions of this paper. Thanks also go to seminar participants for their insightful suggestions. Gao and Peng would also like to acknowledge the Australian Research Council Discovery Projects Program for its financial support under Grant Numbers: DP170104421 \& DP210100476.
 
 {\footnotesize
 
\bibliography{Bibliography-MM-MC}

}

\bigskip

{\small

\begin{center}
	{\large \textbf{Online Supplementary Appendices to \\``Time-Varying Multivariate Causal Processes"}}
	
	\bigskip
	
	{\sc Jiti Gao$^\ast$ and Bin Peng$^{\ast}$ and Wei Biao Wu$^{\dag}$ and Yayi Yan$^{\ast}$}
	
	\medskip
	
	$^\ast$Monash University and $^{\dag}$University of Chicago
	
	\medskip
	
	\today
\end{center}

The file includes Appendix A and Appendix B. We first present some technical tools in Appendix \ref{App.A1}, which will be repeatedly used in the development. We then provide the proofs of main results in Appendix \ref{App.A3}.  We provide several preliminary lemmas in Appendix \ref{AppB.1} as well as some secondary lemmas in Appendix \ref{AppB.2}, and then present the proofs of preliminary lemmas in Appendix \ref{AppB.3}. Appendix \ref{AppB.4} discusses several computational issues of the local linear ML estimation. 

In what follows, $M$ and $O(1)$ always stand for some bounded constants, and may be different at each appearance. 

{\small
	
	\section*{Appendix A}\label{App.A}
	
	\renewcommand{\theequation}{A.\arabic{equation}}
	\renewcommand{\thesection}{A.\arabic{section}}
	\renewcommand{\thefigure}{A.\arabic{figure}}
	\renewcommand{\thetable}{A.\arabic{table}}
	\renewcommand{\thelemma}{A.\arabic{lemma}}
	\renewcommand{\theassumption}{A.\arabic{assumption}}
	\renewcommand{\thetheorem}{A.\arabic{theorem}}
	\renewcommand{\theproposition}{A.\arabic{proposition}}
	
	\setcounter{equation}{0}
	\setcounter{lemma}{0}
	\setcounter{section}{0}
	\setcounter{table}{0}
	\setcounter{figure}{0}
	\setcounter{assumption}{0}
	\setcounter{proposition}{0}
	\numberwithin{equation}{section}

	\section{Technical Tools}\label{App.A1}
	
	\textbf{Projection Operator:}
	Define the projection operator  
	\begin{eqnarray*}
		\mathcal{P}_t(\cdot) = E[\cdot \mid \mathcal{F}_{t}] - E[\cdot \mid \mathcal{F}_{t-1}],
	\end{eqnarray*}
	where $\mathcal{F}_{t} = \sigma(\bm{\varepsilon}_t, \bm{\varepsilon}_{t-1},\ldots)$.
	By the Jensen's inequality and the stationarity of $\widetilde{\mathbf{x}}_t(\tau)$, for $l\geq 0$, we have
	\begin{eqnarray*}
		\|\mathcal{P}_{t-l}(\widetilde{\mathbf{x}}_t(\tau))\|_{r}&=&\| E [\widetilde{\mathbf{x}}_t(\tau)\mid\mathcal{F}_{t-l}]-E[\widetilde{\mathbf{x}}_t(\tau)\mid\mathcal{F}_{t-l-1}]\|_{r}\nonumber\\
		&=&\|E[\widetilde{\mathbf{x}}_t(\tau) \mid \mathcal{F}_{t-l}]-E[ \widetilde{\mathbf{x}}_t^{(t-l,*)}(\tau) \mid \mathcal{F}_{t-l-1}]\|_{r}\nonumber\\
		&=&\|E[\widetilde{\mathbf{x}}_t(\tau)-\widetilde{\mathbf{x}}_t^{(t-l,*)}(\tau)\mid\mathcal{F}_{t-l}]\|_{r}\nonumber\\
		&\leq& \|\widetilde{\mathbf{x}}_t(\tau)-\widetilde{\mathbf{x}}_t^{(t-l,*)}(\tau)\|_{r} = \delta_{r}^{\mathbf{x}(\tau)}(l),
	\end{eqnarray*}
	where $\widetilde{\mathbf{x}}_t^{(t-l,*)}(\tau)$ is a coupled version of $\widetilde{\mathbf{x}}_t(\tau)$ with $\varepsilon_{t-l}$ replaced by $\varepsilon_{t-l}^*$. 
	
	\medskip
	
	\noindent \textbf{The Class $\mathcal{H}(C,\bm{\chi},M)$:}
	
	Recall that we have defined $\bm{\Theta}_r$ in Assumption1. Let $\bm{\chi} = \{\chi_j\}_{j=1}^{\infty}$ be a sequence of nonnegative real numbers with $|\bm{\chi}|_1:=\sum_{j=1}^{\infty}\chi_j < \infty$ and $M > 0$ be some finite constant. Let $|\mathbf{z}|_{\bm{\chi}} := \sum_{j=1}^{\infty}\chi_j|\mathbf{z}_j|$ for any $\mathbf{z} \in (\mathbb{R}^{m})^{\infty}$ and $C\geq 1$, where $\mathbf{z}_j$ is the $j^{th}$ column of $\mathbf{z}$. A function $g(\mathbf{z},\bm{\vartheta}):(\mathbb{R}^m)^{\infty}\times \bm{\Theta}_r\to \mathbb{R}$ is in class $\mathcal{H}(C,\bm{\chi},M)$ if  
	\begin{eqnarray*}
		&&\sup_{\bm{\vartheta} \in \bm{\Theta}_r}|g(\mathbf{0},\bm{\vartheta})|\leq M,\\
		&&\sup_{\mathbf{z}}\sup_{\bm{\vartheta}\neq\bm{\vartheta}^\prime}\frac{|g(\mathbf{z},\bm{\vartheta})-g(\mathbf{z},\bm{\vartheta}^\prime)|}{|\bm{\vartheta}-\bm{\vartheta}^\prime|(1+|\mathbf{z}|_{\bm{\chi}}^C )}\leq M,\\
		&&\sup_{\bm{\vartheta}}\sup_{\mathbf{z}\neq\mathbf{z}^\prime}\frac{|g(\mathbf{z},\bm{\vartheta})-g(\mathbf{z}^\prime,\bm{\vartheta})|}{|\mathbf{z}-\mathbf{z}^\prime|_{\bm{\chi}}(1+|\mathbf{z}|_{\bm{\chi}}^{C-1}+|\mathbf{z}^\prime|_{\bm{\chi}}^{C-1} )}\leq M.
	\end{eqnarray*}
	
	If $g$ is vector- or matrix-valued, $g \in \mathcal{H}(C,\bm{\chi},M)$ means that every component of $g$ is in $\mathcal{H}(C,\bm{\chi},M)$.
	
	\medskip
	
	\noindent \textbf{Analytical Gradient:}
	
	Let 
	\begin{eqnarray*}
		\mathcal{l}(\mathbf{x},\mathbf{z};\bm{\vartheta})=-\frac{1}{2}(\mathbf{x}- \bm{\mu}(\mathbf{z};\bm{\vartheta}))^\top\mathbf{M}^{-1}(\mathbf{z};\bm{\vartheta})(\mathbf{x} - \bm{\mu}(\mathbf{z};\bm{\vartheta}))-\frac{1}{2}\log\det\left(\mathbf{M}(\mathbf{z};\bm{\vartheta})\right),
	\end{eqnarray*}
	where $\mathbf{M}(\mathbf{z};\bm{\vartheta})=\mathbf{H}(\mathbf{z};\bm{\vartheta})\mathbf{H}(\mathbf{z};\bm{\vartheta})^\top$. Then the first partial derivative is as follows:
	\begin{eqnarray}\label{EqA.1}
		\frac{\partial\mathcal{l}(\mathbf{x},\mathbf{z};\bm{\vartheta})}{\partial \vartheta_i} &=& (\mathbf{x}- \bm{\mu}(\mathbf{z};\bm{\vartheta}))^\top\mathbf{M}^{-1}(\mathbf{z};\bm{\vartheta}) \frac{\partial\bm{\mu}(\mathbf{z};\bm{\vartheta})}{\partial\vartheta_i} \nonumber\\
		&&-\frac{1}{2}(\mathbf{x}- \bm{\mu}(\mathbf{z};\bm{\vartheta}))^\top\frac{\partial\mathbf{M}^{-1}(\mathbf{z};\bm{\vartheta})}{\partial\vartheta_i}(\mathbf{x} - \bm{\mu}(\mathbf{z};\bm{\vartheta}))\nonumber\\
		&&-\frac{1}{2}\mathrm{tr}\left(\mathbf{M}^{-1}(\mathbf{z};\bm{\vartheta})\frac{\partial \mathbf{M}(\mathbf{z};\bm{\vartheta})}{\partial \vartheta_i} \right),
	\end{eqnarray}
	where $\vartheta_i$ is the $i^{th}$ element of $\bm{\vartheta}$.
	
	By \eqref{EqA.1}, the second partial derivative of $\mathcal{l}(\mathbf{x},\mathbf{z};\bm{\vartheta})$ is given by
	\begin{eqnarray}\label{EqA.2}
		\frac{\partial^2\mathcal{l}(\mathbf{x},\mathbf{z};\bm{\vartheta})}{\partial \vartheta_i \partial \vartheta_j}&=& (\mathbf{x}- \bm{\mu}(\mathbf{z};\bm{\vartheta}))^\top\mathbf{M}^{-1}(\mathbf{z};\bm{\vartheta}) \frac{\partial^2\bm{\mu}(\mathbf{z};\bm{\vartheta})}{\partial\vartheta_i\partial\vartheta_j} \nonumber\\
		&&-\frac{1}{2}(\mathbf{x}- \bm{\mu}(\mathbf{z};\bm{\vartheta}))^\top\frac{\partial^2\mathbf{M}^{-1}(\mathbf{z};\bm{\vartheta})}{\partial\vartheta_i\partial\vartheta_j}(\mathbf{x} - \bm{\mu}(\mathbf{z};\bm{\vartheta})) \nonumber\\
		&&+(\mathbf{x}- \bm{\mu}(\mathbf{z};\bm{\vartheta}))^\top\left( \frac{\partial\mathbf{M}^{-1}(\mathbf{z};\bm{\vartheta})}{\partial\vartheta_i}\frac{\partial\bm{\mu}(\mathbf{z};\bm{\vartheta})}{\partial\vartheta_j}+\frac{\partial\mathbf{M}^{-1}(\mathbf{z};\bm{\vartheta})}{\partial\vartheta_j}\frac{\partial\bm{\mu}(\mathbf{z};\bm{\vartheta})}{\partial\vartheta_i}\right)\nonumber\\
		&&-\left(\frac{\partial\bm{\mu}(\mathbf{z};\bm{\vartheta})}{\partial\vartheta_j}\right)^\top\mathbf{M}^{-1}(\mathbf{z};\bm{\vartheta}) \frac{\partial\bm{\mu}(\mathbf{z};\bm{\vartheta})}{\partial\vartheta_i}\nonumber\\
		&&-\frac{1}{2}\mathrm{tr}\left(\mathbf{M}^{-1}(\mathbf{z};\bm{\vartheta})\frac{\partial^2 \mathbf{M}(\mathbf{z};\bm{\vartheta})}{\partial \vartheta_i\partial \vartheta_j} \right)-\frac{1}{2}\mathrm{tr}\left(\frac{\partial\mathbf{M}^{-1}(\mathbf{z};\bm{\vartheta})}{\partial \vartheta_j}\frac{\partial \mathbf{M}(\mathbf{z};\bm{\vartheta})}{\partial \vartheta_i} \right).
	\end{eqnarray}

	\section{Proofs of the Main Results}\label{App.A3}
	
	\begin{proof}[Proof of Proposition 2.1]
		\item
		\noindent (1). To prove the first result, we consider an approximated $p$-Markov process defined by
		\begin{eqnarray}\label{EqB.1}
			\widetilde{\mathbf{x}}_{p,t}(\tau) &=& \bm{\mu}\left(\widetilde{\mathbf{x}}_{p,t-1}(\tau),\ldots,\widetilde{\mathbf{x}}_{p,t-p}(\tau),\mathbf{0},\ldots;\bm{\theta}(\tau)\right)\nonumber \\
			&& + \mathbf{H}\left(\widetilde{\mathbf{x}}_{p,t-1}(\tau),\ldots,\widetilde{\mathbf{x}}_{p,t-p}(\tau),\mathbf{0},\ldots;\bm{\theta}(\tau)\right)\bm{\varepsilon}_t
		\end{eqnarray}
		for $p \geq 1$, and 
		\begin{eqnarray}
			\widetilde{\mathbf{x}}_{0,t}(\tau) = \bm{\mu}\left(\mathbf{0},\ldots;\bm{\theta}(\tau)\right) + \mathbf{H}\left(\mathbf{0},\ldots;\bm{\theta}(\tau)\right)\bm{\varepsilon}_t.
		\end{eqnarray}	
		Let $\mu_{p,r}(\tau)= \left\|\widetilde{\mathbf{x}}_{p,t}(\tau)\right\|_r$ and $\Delta_{p,r}(\tau) = \left\|\widetilde{\mathbf{x}}_{p+1,t}(\tau)-\widetilde{\mathbf{x}}_{p,t}(\tau)\right\|_r$. 
		
		By construction, we immediately obtain
		\begin{eqnarray*}
			\mu_{p,r}(\tau)&\leq& \|\widetilde{\mathbf{x}}_{p,t}(\tau)-\widetilde{\mathbf{x}}_{0,t}(\tau)\|_r + \mu_{0,r}(\tau) \\
			&\leq& \left(\sum_{j=1}^{p}\alpha_j(\bm{\theta}(\tau)) + \|\bm{\varepsilon}_1\|_r \sum_{j=1}^{p}\beta_j(\bm{\theta}(\tau))\right)\mu_{p,r}(\tau) + \mu_{0,r}(\tau),
		\end{eqnarray*}
		where the second inequality follows from Assumption 1.
		
		Recall that we have defined $\rho(\tau):=\sum_{j=1}^{\infty}\alpha_j(\bm{\theta}(\tau)) + \|\bm{\varepsilon}_t\|_r \sum_{j=1}^{\infty}\beta_j(\bm{\theta}(\tau))$ in the body of this proposition. As $0\leq \rho(\tau) < 1$ by Assumption 1, we have 
		\begin{eqnarray*}
			\sup_{p\geq 0}\mu_{p,r}(\tau) \leq (1-\rho(\tau))^{-1}\mu_{0,r}(\tau)<\infty.
		\end{eqnarray*}
		Similarly, we have
		\begin{eqnarray*}
			\Delta_{p,r}(\tau)&\leq& \left(\sum_{j=1}^{p}\alpha_j(\bm{\theta}(\tau)) + \|\bm{\varepsilon}_1\|_r \sum_{j=1}^{p}\beta_j(\bm{\theta}(\tau))\right)\Delta_{p,r}(\tau)\\
			&&+\left(\alpha_{p+1}(\bm{\theta}(\tau))+\|\bm{\varepsilon}_1\|_r\beta_{p+1}(\bm{\theta}(\tau))\right)\left\|\widetilde{\mathbf{x}}_{p+1,t-p-1}(\tau)\right\|_r.
		\end{eqnarray*}
		Hence,  
		\begin{eqnarray*}
			\Delta_{p,r}(\tau)\leq\left(\alpha_{p+1}(\bm{\theta}(\tau))+\|\bm{\varepsilon}_1\|_r\beta_{p+1}(\bm{\theta}(\tau))\right)\left(1-\rho(\tau)\right)^{-2}\mu_{0,r}(\tau) \to 0
		\end{eqnarray*}	
		as $p\to \infty$. 
		
		According to the above development, we are readily to conclude that $\widetilde{\mathbf{x}}_{p,t}(\tau) \to \widetilde{\mathbf{x}}_{t}(\tau)$ as $p \to \infty$ in the space of $\mathbb{L}^r =\{\mathbf{x}\mid \|\mathbf{x} \|_r<\infty \}$. As a limit of strictly stationary process in $\mathbb{L}^r$, $\widetilde{\mathbf{x}}_{t}(\tau)$ is a stationary process and $\sup_{\tau\in[0,1]}\left\|\widetilde{\mathbf{x}}_{t}(\tau)\right\|_r < \infty$.
		
		\medskip
		
		\noindent (2). Let $\{\bm{\varepsilon}_t^*\}$ be an independent copy of $\{\bm{\varepsilon}_t\}$. Similar to \eqref{EqB.1}, we define the process $\{\widetilde{\mathbf{x}}_{p,t}^*(\tau)\}$, in which the difference is that we use $\bm{\varepsilon}_t$ when $t\ne 0$, and use $\bm{\varepsilon}_t^*$ when $t= 0$.  In addition, define the process $\{\widetilde{\mathbf{x}}_t^*(\tau)\}$ as $\{ \widetilde{\mathbf{x}}_t(\tau)\}$,  in which again the difference is that we use $\bm{\varepsilon}_t$ when $t\ne 0$, and use $\bm{\varepsilon}_t^*$ when $t= 0$. Further define $u_t = \left\|\widetilde{\mathbf{x}}_{p,t}^*(\tau) - \widetilde{\mathbf{x}}_{p,t}(\tau)\right\|_r$.

		By construction, $u_t = 0$ for $t<0$, and $u_0 = \left\|\widetilde{\mathbf{x}}_{p,0}^*(\tau) - \widetilde{\mathbf{x}}_{p,0}(\tau)\right\|_r = O\left(\left\|\bm{\varepsilon}_0^* - \bm{\varepsilon}_0\right\|_r\right)=O(1)$. For $t > 0 $, Assumption 1 gives that
		\begin{eqnarray}\label{EqB.2}
			\left\|\widetilde{\mathbf{x}}_{p,t}(\tau) - \widetilde{\mathbf{x}}_{p,t}^*(\tau) \right\|_r\leq\sum_{j=1}^{p}\left(\alpha_j(\bm{\theta}(\tau)) + \|\bm{\varepsilon}_t\|_r \beta_j(\bm{\theta}(\tau))\right) \left\|\widetilde{\mathbf{x}}_{p,t-j}(\tau) - \widetilde{\mathbf{x}}_{p,t-j}^*(\tau)\right\|_r.
		\end{eqnarray}
		Since $ \sum_{j=1}^{p}\left(\alpha_j(\bm{\theta}(\tau)) + \|\bm{\varepsilon}_t\|_r \beta_j(\bm{\theta}(\tau))\right) \leq \rho(\tau)< 1$, by a recursion argument, we have $u_t\leq u_0$ for all $t$. 
		
		Now, let $v_t = \max_{k\geq t}u_k$. Using \eqref{EqB.2} and the fact that $v_t$ is a nonincreasing sequence, we have $v_t\leq \rho(\tau)v_{t-p}$ for all $t\geq 1$. Then recursively $v_t\leq \rho(\tau)^{-\lfloor-t/p\rfloor}v_{t+p\lfloor-t/p\rfloor}$. Since $v_{t+p\lfloor-t/p\rfloor}\leq u_0$ and $-\lfloor-t/p\rfloor\geq t/p$, we have $u_t\leq v_t\leq \rho(\tau)^{t/p}u_0$, i.e., $\left\|\widetilde{\mathbf{x}}_{p,t}(\tau) - \widetilde{\mathbf{x}}_{p,t}^*(\tau) \right\|_r = O(\rho(\tau)^{t/p})$.
		
		The proof of the first result gives
		\begin{eqnarray*}
			\left\|\widetilde{\mathbf{x}}_{t}(\tau) -\widetilde{\mathbf{x}}_{p,t}(\tau) \right\|_r \leq \sum_{j=p}^\infty \Delta_{p,r} \leq \frac{\mu_r(\tau)}{(1-\rho(\tau))^2}\sum_{j=p}^\infty\left(\alpha_{p+1}(\bm{\theta}(\tau))+\|\bm{\varepsilon}_t\|_r\beta_{p+1}(\bm{\theta}(\tau))\right).
		\end{eqnarray*}
		The same bound holds for the quantity $\|\widetilde{\mathbf{x}}_{t}^*(\tau) -\widetilde{\mathbf{x}}_{p,t}^*(\tau) \|_r$. Thus,
		\begin{eqnarray*}
			\|\widetilde{\mathbf{x}}_{t}(\tau) -\widetilde{\mathbf{x}}_{t}^*(\tau) \|_r &\leq& \|\widetilde{\mathbf{x}}_{t}(\tau) -\widetilde{\mathbf{x}}_{p,t}(\tau) \|_r  + \|\widetilde{\mathbf{x}}_{p,t}(\tau) -\widetilde{\mathbf{x}}_{p,t}^*(\tau) \|_r +\|\widetilde{\mathbf{x}}_{t}^*(\tau) -\widetilde{\mathbf{x}}_{p,t}^*(\tau) \|_r \\
			&=& O\left(\rho(\tau)^{t/p}+\sum_{j=p+1}^{\infty}\left(\alpha_j(\bm{\theta}(\tau)) + \|\bm{\varepsilon}_t\|_r \beta_j(\bm{\theta}(\tau))\right)\right),
		\end{eqnarray*}
		which completes the proof.
	\end{proof}

	\begin{proof}[Proof of Proposition 2.2]
		\item
		\noindent (1). Write
		\begin{eqnarray*}
			\|\widetilde{\mathbf{x}}_t(\tau)-\widetilde{\mathbf{x}}_t(\tau^\prime)\|_r &\leq& \|\bm{\mu}\left(\widetilde{\mathbf{x}}_{t-1}(\tau), \ldots;\bm{\theta}(\tau)\right)-\bm{\mu}\left(\widetilde{\mathbf{x}}_{t-1}(\tau^\prime),\ldots;\bm{\theta}(\tau^\prime)\right)\|_r\\
			&&+\|\bm{\varepsilon}_t\|_r\|\mathbf{H}\left(\widetilde{\mathbf{x}}_{t-1}(\tau), \ldots;\bm{\theta}(\tau)\right)-\mathbf{H}\left(\widetilde{\mathbf{x}}_{t-1}(\tau^\prime),\ldots;\bm{\theta}(\tau^\prime)\right)\|_r\\
			&\leq& \sum_{j=1}^\infty\left(\alpha_j(\tau)+\|\bm{\varepsilon}_t\|_r\beta_j(\tau)\right)\|\widetilde{\mathbf{x}}_{t-j}(\tau)-\widetilde{\mathbf{x}}_{t-j}(\tau^\prime)\|_r\\
			&& + M |\tau - \tau^\prime|\sum_{j=1}^\infty\chi_j\|\widetilde{\mathbf{x}}_{t-j}(\tau^\prime)\|_r,
		\end{eqnarray*}
		where the second inequality follows from Assumption 1 and Assumption 2. In view of the stationarity of $\widetilde{\mathbf{x}}_t(\tau)$, rearranging the terms in the above inequality yields that
		\begin{eqnarray*}
			\|\widetilde{\mathbf{x}}_t(\tau)-\widetilde{\mathbf{x}}_t(\tau^\prime)\|_r \leq M (1-\rho(\tau))^{-1}|\tau - \tau^\prime|\sum_{j=1}^\infty\chi_j\|\widetilde{\mathbf{x}}_{t-j}(\tau^\prime)\|_r = O(|\tau - \tau^\prime|).
		\end{eqnarray*}
		
		\medskip
		
		\noindent (2). Write
		\begin{eqnarray*}
			\|\mathbf{x}_t-\widetilde{\mathbf{x}}_t(\tau_t)\|_r&\leq&\|\bm{\mu}\left(\mathbf{x}_{t-1},\mathbf{x}_{t-2},\ldots;\bm{\theta}(\tau_t)\right)-\bm{\mu}\left(\widetilde{\mathbf{x}}_{t-1}(\tau_t),\widetilde{\mathbf{x}}_{t-2}(\tau_t),\ldots;\bm{\theta}(\tau_t)\right)\|_r\\
			&&+\|\bm{\varepsilon}_t\|_r\|\mathbf{H}\left(\mathbf{x}_{t-1},\mathbf{x}_{t-2},\ldots;\bm{\theta}(\tau_t)\right)-\mathbf{H}\left(\widetilde{\mathbf{x}}_{t-1}(\tau_t),\widetilde{\mathbf{x}}_{t-2}(\tau_t),\ldots;\bm{\theta}(\tau_t)\right)\|_r\\
			&\leq& \sum_{j=1}^{\infty}\left(\alpha_j(\tau_t)+\|\bm{\varepsilon}_t\|_r\beta_j(\tau_t)\right)\left\|\mathbf{x}_{t-j}-\widetilde{\mathbf{x}}_{t-j}(\tau_t)\right\|_r\\
			&\leq& \sum_{j=1}^{\infty}\left(\alpha_j(\tau_t)+\|\bm{\varepsilon}_t\|_r\beta_j(\tau_t)\right)\left\|\mathbf{x}_{t-j}-\widetilde{\mathbf{x}}_{t-j}(\tau_{t-j}\lor0)\right\|_r\\
			&& +\sum_{j=1}^{\infty}\left(\alpha_j(\tau_t)+\|\bm{\varepsilon}_t\|_r\beta_j(\tau_t)\right) \left\|\widetilde{\mathbf{x}}_{t-j}(\tau_{t-j}\lor0)-\widetilde{\mathbf{x}}_{t-j}(\tau_t)\right\|_r.
		\end{eqnarray*}
		As $\left\|\mathbf{x}_{t-j}-\widetilde{\mathbf{x}}_{t-j}(\tau_{t-j}\lor0)\right\|_r = 0$ for $j\geq t$, by the first result of this proposition, we have
		\begin{eqnarray*}
			\|\mathbf{x}_t-\widetilde{\mathbf{x}}_t(\tau_t)\|_r&\leq&\sum_{j=1}^{t-1}\left(\alpha_j(\tau_t)+\|\bm{\varepsilon}_t\|_r\beta_j(\tau_t)\right)\left\|\mathbf{x}_{t-j}-\widetilde{\mathbf{x}}_{t-j}(\tau_{t-j})\right\|_r\\
			&& + M\cdot T^{-1}\sum_{j=1}^{\infty}j\left(\alpha_j(\tau_t)+\|\bm{\varepsilon}_t\|_r\beta_j(\tau_t)\right).
		\end{eqnarray*} 
		In addition, as $\left\|\mathbf{x}_{1}-\widetilde{\mathbf{x}}_{1}(\tau_{1})\right\|_r = O(T^{-1})$ and $\sup_{t\geq 2}\sum_{j=1}^{t-1}\left(\alpha_j(\tau_t)+\|\bm{\varepsilon}_t\|_r\beta_j(\tau_t)\right) < 1$, we have
		\begin{eqnarray*}
			\|\mathbf{x}_t-\widetilde{\mathbf{x}}_t(\tau_t)\|_r&\leq&\sum_{j=1}^{t-1}\left(\alpha_j(\tau_t)+\|\bm{\varepsilon}_t\|_r\beta_j(\tau_t)\right)O(T^{-1})\\
			&& + M\cdot T^{-1}\sum_{j=1}^{\infty}j\left(\alpha_j(\tau_t)+\|\bm{\varepsilon}_t\|_r\beta_j(\tau_t)\right) = O(T^{-1}).
		\end{eqnarray*} 
		The proof is now complete.	
	\end{proof}

	\begin{proof}[Proof of Theorem 2.1]
		\item
		
		(1). First, we introduce a few notations to facilitate the development. Let $\widehat{\bm{\eta}} (\tau):= [\widehat{\bm{\theta}} (\tau)^\top, \widehat{\bm{\theta}}^\star(\tau)^\top]^\top$, $\bm{\eta}(\tau):= [\bm{\theta}(\tau)^\top,h\bm{\theta}^{(1)}(\tau)^\top]^\top$, and $\mathcal{L}_\tau(\bm{\eta}) := \mathcal{L}_\tau(\bm{\eta}_1,\bm{\eta}_2)$ for $\bm{\eta}=[\bm{\eta}_1^\top,\bm{\eta}_2^\top]^\top$. Recall that we have defined $\gradient_{\bm{\vartheta}}$,  and let $\gradient_{\bm{\eta}}$ be defined similarly with respect to the  elements of $\bm{\eta}$.
		
		By the Taylor expansion, we have
		\begin{eqnarray*}
			\widehat{\bm{\eta}}(\tau) - \bm{\eta}(\tau) &=& -(\gradient_{\bm{\eta}}^2 \mathcal{L}_\tau(\overline{\bm{\eta}}))^{-1}\gradient_{\bm{\eta}} \mathcal{L}_\tau(\bm{\eta}(\tau)),
		\end{eqnarray*}
		with $\overline{\bm{\eta}}$ between $\widehat{\bm{\eta}}(\tau)$ and $\bm{\eta}(\tau)$. By Lemma \ref{LemmaA.3}.4, we have 
		\begin{eqnarray*}
			|\gradient_{\bm{\eta}} \mathcal{L}_\tau( \bm{\eta}(\tau))-\gradient_{\bm{\eta}} \widetilde{\mathcal{L}}_\tau(\bm{\eta}(\tau))|=O_P((Th)^{-1}),
		\end{eqnarray*}
		where $
		\widetilde{\mathcal{L}}_\tau(\bm{\eta}(\tau))=T^{-1}\sum_{t=1}^{T}\mathcal{l}(\widetilde{\bm{x}}_t(\tau_t),\widetilde{\bm{z}}_{t-1}(\tau_t);\bm{\theta}(\tau)+\bm{\theta}^{(1)}(\tau)(\tau_t-\tau))K_h(\tau_t-\tau)$.
		
		\medskip
		
		Then we consider $\gradient_{\bm{\eta}} \widetilde{\mathcal{L}}_\tau(\bm{\eta}(\tau))$. Since each element of $\bm{\theta}(\tau)$ is in $C^3[0,1]$, we have $\bm{\theta}(\tau_t)=\bm{\theta}(\tau) + \bm{\theta}^{(1)}(\tau)(\tau_t-\tau)+\mathbf{r}(\tau_t)$, where $\mathbf{r}(\tau_t) = \frac{1}{2}\bm{\theta}^{(2)}(\tau)(\tau_t-\tau)^2+\frac{1}{6}\bm{\theta}^{(3)}(\overline{\tau})(\tau_t-\tau)^3$ with $\overline{\tau}$ between $\tau_t$ and $\tau$. Let $\widehat{\bm{K}}((\tau_t-\tau)/h)=[K((\tau_t-\tau)/h),(\tau_t-\tau)/hK((\tau_t-\tau)/h)]^\top$. By the Mean Value Theorem, we have
		\begin{eqnarray*}
			&&\gradient_{\bm{\eta}} \widetilde{\mathcal{L}}_\tau(\bm{\eta}(\tau))-\frac{1}{Th}\sum_{t=1}^{T}\widehat{\bm{K}}((\tau_t-\tau)/h)\otimes \gradient_{\bm{\vartheta}}\mathcal{l}(\widetilde{\bm{x}}_t(\tau_t),\widetilde{\bm{z}}_{t-1}(\tau_t);\bm{\theta}(\tau_t))\\
			&=&-\frac{1}{Th}\sum_{t=1}^{T}\widehat{\bm{K}}((\tau_t-\tau)/h)\otimes\left[\gradient_{\bm{\vartheta}}^2\mathcal{l}(\widetilde{\bm{x}}_t(\tau_t),\widetilde{\bm{z}}_{t-1}(\tau_t);\bm{\theta}(\tau_t)-u\mathbf{r}(\tau_t))\mathbf{r}(\tau_t)\right]
		\end{eqnarray*}
		with some $u\in[0,1]$.  Since $\gradient_{\bm{\vartheta}}^2\mathcal{l}$ is in class $\mathcal{H}(3,\bm{\chi},M)$ by Lemma \ref{LemmaA.2},  using Lemma \ref{LemmaB.1} and $|\tau_t-\tau|\leq h$ yields
		\begin{eqnarray*}
			\|\gradient_{\bm{\vartheta}}^2\mathcal{l}(\widetilde{\bm{x}}_t(\tau_t),\widetilde{\bm{z}}_{t-1}(\tau_t);\bm{\theta}(\tau_t)-u\mathbf{r}(\tau_t)) - \gradient_{\bm{\vartheta}}^2\mathcal{l}(\widetilde{\bm{x}}_t(\tau),\widetilde{\bm{z}}_{t-1}(\tau);\bm{\theta}(\tau))\|_1 = O(h).
		\end{eqnarray*}
		
		The above analyses plus Lemma \ref{LemmaA.5} reveal that
		\begin{eqnarray*}
			&&\gradient_{\bm{\eta}} \widetilde{\mathcal{L}}_\tau(\bm{\eta}(\tau))-\frac{1}{Th}\sum_{t=1}^{T}\widehat{\bm{K}}((\tau_t-\tau)/h)\otimes \gradient_{\bm{\vartheta}}\mathcal{l}(\widetilde{\bm{x}}_t(\tau_t),\widetilde{\bm{z}}_{t-1}(\tau_t);\bm{\theta}(\tau_t))\\
			&=&-\frac{1}{2}h^2\frac{1}{Th}\sum_{t=1}^{T}\widehat{\bm{K}}((\tau_t-\tau)/h)\otimes\left[\gradient_{\bm{\vartheta}}^2\mathcal{l}(\widetilde{\bm{x}}_t(\tau),\widetilde{\bm{z}}_{t-1}(\tau);\bm{\theta}(\tau))\cdot \bm{\theta}^{(2)}(\tau)\left(\frac{\tau_t-\tau}{h}\right)^2\right]+O_P(h^3)\\
			&=&\frac{1}{2}h^2\int_{-\tau/h}^{(1-\tau)/h}K(u)[u^2,u^3]^\top\mathrm{d}u\otimes \left(-\bm{\Sigma}(\tau) \bm{\theta}^{(2)}(\tau)\right)+ O_P(h^3).
		\end{eqnarray*}
		
		By Lemmas \ref{LemmaA.4} and \ref{LemmaA.5}, we have 
		\begin{eqnarray*}
			\gradient_{\bm{\eta}} \widetilde{\mathcal{L}}_\tau(\bm{\eta}(\tau)) \to_P0 \quad\text{and}\quad \sup_{\bm{\eta}}|\gradient_{\bm{\eta}}^2 \mathcal{L}_\tau( \bm{\eta})-E(\gradient_{\bm{\eta}}^2 \widetilde{\mathcal{L}}_\tau(\bm{\eta}))|\to_P 0.
		\end{eqnarray*}
		Hence, we have $\gradient_{\bm{\eta}}^2 \mathcal{L}_\tau(\overline{\bm{\eta}}) \to_P \bm{\Sigma}(\tau)$ and thus for any $\tau \in [h,1-h]$, as $Th^7\to 0$, we have
		\begin{eqnarray*}
			&&\sqrt{Th}\left(\widehat{\bm{\theta}} (\tau) - \bm{\theta}(\tau) -\frac{1}{2}h^2\widetilde{c}_2\bm{\theta}^{(2)}(\tau)\right)\\
			&=& -\bm{\Sigma}^{-1}(\tau)\frac{1}{\sqrt{Th}}\sum_{t=1}^{T}K((\tau_t-\tau)/h) \gradient_{\bm{\vartheta}}\mathcal{l}(\widetilde{\bm{x}}_t(\tau_t),\widetilde{\bm{z}}_{t-1}(\tau_t);\bm{\theta}(\tau_t))+o_P(1).
		\end{eqnarray*}
		
		In addition, by Lemma \ref{LemmaA.1}, we have $E\left(\gradient_{\bm{\vartheta}}\mathcal{l}(\widetilde{\mathbf{x}}_t(\tau),\widetilde{\mathbf{z}}_{t-1}(\tau);\bm{\theta}(\tau))\right) = 0$. To prove this theorem, by the Cramer-Wold device, it suffices to show that for any unit vector $\mathbf{d}$,
		\begin{eqnarray*}
			\frac{1}{\sqrt{Th}}\sum_{t=1}^{T}K((\tau_t-\tau)/h) \mathbf{d}^\top\gradient_{\bm{\vartheta}}\mathcal{l}(\widetilde{\bm{x}}_t(\tau_t),\widetilde{\bm{z}}_{t-1}(\tau_t);\bm{\theta}(\tau_t)) \to_D N\left(\mathbf{0}, \widetilde{v}_0\mathbf{d}^\top \bm{\Omega}(\tau)\mathbf{d}\right).
		\end{eqnarray*}
		Note that $\{\gradient_{\bm{\vartheta}}\mathcal{l}(\widetilde{\bm{x}}_t(\tau_t),\widetilde{\bm{z}}_{t-1}(\tau_t);\bm{\theta}(\tau_t))\}_t$ is a sequence of martingale differences, we prove the asymptotic normality by using the martingale central limit theorem \cite[]{hall2014martingale}. We first consider the convergence of conditional variance. Let $w_t(u) = \frac{1}{\sqrt{Th}}K((\tau_t-\tau)/h) \mathbf{d}^\top\gradient_{\bm{\vartheta}}\mathcal{l}(\widetilde{\bm{x}}_t(u),\widetilde{\bm{z}}_{t-1}(u);\bm{\theta}(u))$. By Lemma \ref{LemmaB.1}, we have
		\begin{eqnarray*}
			&&\sum_{t=1}^{T}\|w_t(\tau_t)^2- w_t(\tau)^2\|_1\\
			&\leq& \frac{1}{Th}\sum_{t=1}^{T}K((\tau_t-\tau)/h)^2\|\gradient_{\bm{\vartheta}}\mathcal{l}(\widetilde{\bm{x}}_t(\tau_t),\widetilde{\bm{z}}_{t-1}(\tau_t);\bm{\theta}(\tau_t))- \gradient_{\bm{\vartheta}}\mathcal{l}(\widetilde{\bm{x}}_t(\tau),\widetilde{\bm{z}}_{t-1}(\tau);\bm{\theta}(\tau))\|_2\\
			&&\times2\sup_{u}\|\gradient_{\bm{\vartheta}}\mathcal{l}(\widetilde{\bm{x}}_t(u),\widetilde{\bm{z}}_{t-1}(u);\bm{\theta}(u))\|_2\\
			&=&O(h)=o(1).
		\end{eqnarray*}
		In addition, by Proposition 2.1, $\{E\left[(\mathbf{d}^\top\gradient_{\bm{\vartheta}}\mathcal{l}(\widetilde{\bm{x}}_t(u),\widetilde{\bm{z}}_{t-1}(u);\bm{\theta}(u)))^2 \mid \mathcal{F}_{t-1} \right]\}_{t=1}^T$ is a sequence of stationary variables and thus we have
		\begin{eqnarray*}
			&&\sum_{t=1}^{T}E\left(w_t(\tau)^2\mid\mathcal{F}_{t-1}\right)\\
			&=&\frac{1}{Th}\sum_{t=1}^{T}K((\tau_t-\tau)/h)^2E\left[(\mathbf{d}^\top\gradient_{\bm{\vartheta}}\mathcal{l}(\widetilde{\bm{x}}_t(u),\widetilde{\bm{z}}_{t-1}(u);\bm{\theta}(u)))^2 \mid \mathcal{F}_{t-1} \right]\\
			&\to_P&\widetilde{v}_0\mathbf{d}^\top \bm{\Omega}(\tau)\mathbf{d}.
		\end{eqnarray*}
		We next verify the Lindeberg condition. The sum $\sum_{t=1}^{T}E\left(w_t^2(\tau_t) I(|w_t(\tau_t)|>v) \right)$ is bounded by
		\begin{eqnarray*}
			M E\left(\sup_{\tau}|\gradient_{\bm{\vartheta}}\mathcal{l}(\widetilde{\bm{x}}_t(\tau),\widetilde{\bm{z}}_{t-1}(\tau);\bm{\theta}(\tau))|^2 I(\sup_{\tau}|\gradient_{\bm{\vartheta}}\mathcal{l}(\widetilde{\bm{x}}_t(\tau),\widetilde{\bm{z}}_{t-1}(\tau);\bm{\theta}(\tau))|>\sqrt{Th}v) \right),
		\end{eqnarray*}
		which converges to zero since $\|\sup_{\tau}|\gradient_{\bm{\vartheta}}\mathcal{l}(\widetilde{\bm{x}}_t(\tau),\widetilde{\bm{z}}_{t-1}(\tau);\bm{\theta}(\tau))| \|_2<\infty$ by Lemma \ref{LemmaB.1}.3. The asymptotic normality is then obtained.
		
		The proof of the first result is now complete.
		
		\medskip
		
		
		(2).	For notation simplicity, we abbreviate $\mathcal{l}(\mathbf{x},\mathbf{z};\bm{\vartheta}),\bm{\mu}(\mathbf{x},\mathbf{z};\bm{\vartheta}),\mathbf{M}(\mathbf{z};\bm{\vartheta})$ to $\mathcal{l},\bm{\mu},\mathbf{M}$ in what follows. Note that
		\begin{eqnarray*}
			\mathrm{d}\mathcal{l} &=& (\mathbf{x} - \bm{\mu})^\top\mathbf{M}^{-1} \mathrm{d}\bm{\mu} - \frac{1}{2} (\mathbf{x} - \bm{\mu})^\top\mathrm{d}\mathbf{M}^{-1}(\mathbf{x} - \bm{\mu})-\frac{1}{2}\mathrm{tr}\{\mathbf{M}^{-1} \mathrm{d}\mathbf{M}\}\\
			&=&(\mathbf{x} - \bm{\mu})^\top\mathbf{M}^{-1}\frac{\partial \bm{\mu}}{\partial \bm{\vartheta}^\top} \mathrm{d}\bm{\vartheta} + \frac{1}{2} ((\mathbf{x} - \bm{\mu})^\top\mathbf{M}^{-1}\otimes(\mathbf{x} - \bm{\mu})^\top\mathbf{M}^{-1}) \frac{\partial \mathrm{vec}(\mathbf{M})}{\partial \bm{\vartheta}^\top} \mathrm{d}\bm{\vartheta}\\
			&& - \frac{1}{2} \mathrm{vec}(\mathbf{M}^{-1})^\top\frac{\partial \mathrm{vec}(\mathbf{M})}{\partial \bm{\vartheta}^\top} \mathrm{d}\bm{\vartheta}.
		\end{eqnarray*} 
		Hence, we have
		\begin{eqnarray}\label{EqB.3}
			\frac{\partial\mathcal{l}}{\partial \bm{\vartheta}} \frac{\partial\mathcal{l}}{\partial \bm{\vartheta}^\top} &=&  \frac{\partial \bm{\mu}^\top}{\partial \bm{\vartheta}} \mathbf{M}^{-1}(\mathbf{x} - \bm{\mu})(\mathbf{x} - \bm{\mu})^\top\mathbf{M}^{-1}\frac{\partial \bm{\mu}}{\partial \bm{\vartheta}^\top}+ \frac{1}{4} \frac{\partial \mathrm{vec}(\mathbf{M})^\top}{\partial \bm{\vartheta}} \mathrm{vec}(\mathbf{M}^{-1})\mathrm{vec}(\mathbf{M}^{-1})^\top \frac{\partial \mathrm{vec}(\mathbf{M})}{\partial \bm{\vartheta}^\top} \nonumber\\
			&& + \frac{1}{4} \frac{\partial \mathrm{vec}(\mathbf{M})^\top}{\partial \bm{\vartheta}} [ \mathbf{M}^{-1}(\mathbf{x} - \bm{\mu}) (\mathbf{x} - \bm{\mu})^\top\mathbf{M}^{-1} \otimes\mathbf{M}^{-1}(\mathbf{x} - \bm{\mu}) (\mathbf{x} - \bm{\mu})^\top\mathbf{M}^{-1} ] \frac{\partial \mathrm{vec}(\mathbf{M})}{\partial \bm{\vartheta}^\top} \nonumber\\
			&& + \frac{1}{2} \frac{\partial \bm{\mu}^\top}{\partial \bm{\vartheta}}[ \mathbf{M}^{-1}(\mathbf{x} - \bm{\mu}) (\mathbf{x} - \bm{\mu})^\top\mathbf{M}^{-1} \otimes(\mathbf{x} - \bm{\mu})^\top\mathbf{M}^{-1} ] \frac{\partial \mathrm{vec}(\mathbf{M})}{\partial \bm{\vartheta}^\top} \nonumber\\
			&& - \frac{1}{2}\frac{\partial \bm{\mu}^\top}{\partial \bm{\vartheta}} \mathbf{M}^{-1}(\mathbf{x} - \bm{\mu})\mathrm{vec}(\mathbf{M}^{-1})^\top\frac{\partial \mathrm{vec}(\mathbf{M})}{\partial \bm{\vartheta}^\top} \nonumber\\
			&& + \frac{1}{2} \frac{\partial \bm{\mu}^\top}{\partial \bm{\vartheta}}[(\mathbf{x} - \bm{\mu})^\top\mathbf{M}^{-1} \otimes \mathbf{M}^{-1}(\mathbf{x} - \bm{\mu}) (\mathbf{x} - \bm{\mu})^\top\mathbf{M}^{-1}] \frac{\partial \mathrm{vec}(\mathbf{M})}{\partial \bm{\vartheta}^\top} \nonumber\\
			&& - \frac{1}{4} \frac{\partial \mathrm{vec}(\mathbf{M})^\top}{\partial \bm{\vartheta}} [ \mathbf{M}^{-1}(\mathbf{x} - \bm{\mu})\otimes\mathbf{M}^{-1}(\mathbf{x} - \bm{\mu})]\mathrm{vec}(\mathbf{M}^{-1})^\top \frac{\partial \mathrm{vec}(\mathbf{M})}{\partial \bm{\vartheta}^\top} \nonumber\\
			&& - \frac{1}{2} \frac{\partial \mathrm{vec}(\mathbf{M})^\top}{\partial \bm{\vartheta}}\mathrm{vec}(\mathbf{M}^{-1}) (\mathbf{x} - \bm{\mu})^\top\mathbf{M}^{-1}\frac{\partial \bm{\mu}}{\partial \bm{\vartheta}^\top} \nonumber\\
			&& - \frac{1}{4} \frac{\partial \mathrm{vec}(\mathbf{M})^\top}{\partial \bm{\vartheta}}\mathrm{vec}(\mathbf{M}^{-1}) [(\mathbf{x} - \bm{\mu})^\top\mathbf{M}^{-1} \otimes (\mathbf{x} - \bm{\mu})^\top\mathbf{M}^{-1}] \frac{\partial \mathrm{vec}(\mathbf{M})}{\partial \bm{\vartheta}^\top}.
		\end{eqnarray}
		In addition, if $\bm{\varepsilon}_t$ is normal distributed, we have $E(\bm{\varepsilon_t}\bm{\varepsilon_t}^\top \otimes \bm{\varepsilon_t}\bm{\varepsilon_t}^\top) = 2\mathbf{N}_m + \mathrm{vec}(\mathbf{I}_m)\mathrm{vec}(\mathbf{I}_m)^\top$ and $E(\mathbf{c}\bm{\varepsilon_t}^\top \otimes \bm{\varepsilon_t}\bm{\varepsilon_t}^\top) = \mathbf{0}$, where $\mathbf{c}$ is independent of $\bm{\varepsilon_t}$, $2\mathbf{N}_m = \mathbf{I}_{m^2} + \mathbf{K}_{mm}$ and $\mathbf{K}_{mm}$ is a commutation matrix. By \eqref{EqB.3}, if $\bm{\varepsilon_t}$ is normal distributed, we have
		\begin{eqnarray*}
			&&\bm{\Omega}(\tau)=E\left(\gradient_{\bm{\vartheta}}\mathcal{l}(\widetilde{\mathbf{x}}_t(\tau),\widetilde{\mathbf{z}}_{t-1}(\tau);\bm{\vartheta}(\tau))\cdot \gradient_{\bm{\vartheta}}\mathcal{l}(\widetilde{\mathbf{x}}_t(\tau),\widetilde{\mathbf{z}}_{t-1}(\tau);\bm{\vartheta}(\tau))^\top\right) \\
			&=& E\left(\frac{\partial \bm{\mu}(\widetilde{\mathbf{x}}_t(\tau),\widetilde{\mathbf{z}}_{t-1}(\tau);\bm{\vartheta}(\tau))^\top}{\partial \bm{\vartheta}} \mathbf{M}^{-1}(\widetilde{\mathbf{x}}_t(\tau),\widetilde{\mathbf{z}}_{t-1}(\tau);\bm{\vartheta}(\tau))\frac{\partial \bm{\mu}(\widetilde{\mathbf{x}}_t(\tau),\widetilde{\mathbf{z}}_{t-1}(\tau);\bm{\vartheta}(\tau))}{\partial \bm{\vartheta}^\top}\right) \\
			&&+ \frac{1}{2} E\left(\frac{\partial \mathrm{vec}(\mathbf{M}(\widetilde{\mathbf{x}}_t(\tau),\widetilde{\mathbf{z}}_{t-1}(\tau);\bm{\vartheta}(\tau)))^\top}{\partial \bm{\vartheta}} [\mathbf{M}^{-1}(\widetilde{\mathbf{x}}_t(\tau),\widetilde{\mathbf{z}}_{t-1}(\tau);\bm{\vartheta}(\tau))\otimes \mathbf{M}^{-1}(\widetilde{\mathbf{x}}_t(\tau),\widetilde{\mathbf{z}}_{t-1}(\tau);\bm{\vartheta}(\tau))] \right.\\
			&&\left. \times \frac{\partial \mathrm{vec}(\mathbf{M}(\widetilde{\mathbf{x}}_t(\tau),\widetilde{\mathbf{z}}_{t-1}(\tau);\bm{\vartheta}(\tau)))}{\partial \bm{\vartheta}^\top}\right).
		\end{eqnarray*}
		
		Next, consider the Hessian matrix.
		\begin{eqnarray}\label{EqB.4}
			\mathrm{d}^2\mathcal{l} &=&- \mathrm{d}\bm{\vartheta}^\top \frac{\partial \bm{\mu}^\top}{\partial \bm{\vartheta}}\mathbf{M}^{-1}\frac{\partial \bm{\mu}}{\partial \bm{\vartheta}^\top} \mathrm{d}\bm{\vartheta} - \mathrm{d}\bm{\vartheta}^\top\frac{\partial \mathrm{vec}(\mathbf{M})^\top}{\partial \bm{\vartheta}} (\mathbf{M}^{-1}(\mathbf{x} - \bm{\mu}) \otimes \mathbf{M}^{-1}) \frac{\partial \bm{\mu}}{\partial \bm{\vartheta}^\top} \mathrm{d}\bm{\vartheta} \nonumber \\
			&&+\mathrm{d}\bm{\vartheta}^\top\frac{\partial \mathrm{vec}(\frac{\partial \bm{\mu}^\top}{\partial \bm{\vartheta}})^\top}{\partial \bm{\vartheta}} (\mathbf{M}^{-1}(\mathbf{x} - \bm{\mu}) \otimes \mathbf{I}_{m}) \frac{\partial \bm{\mu}}{\partial \bm{\vartheta}^\top} \mathrm{d}\bm{\vartheta} \nonumber\\
			&& + \frac{1}{2} \mathrm{d}\bm{\vartheta}^\top \frac{\partial \mathrm{vec}(\mathbf{M})^\top}{\partial \bm{\vartheta}} [ \mathbf{M}^{-1}\otimes\mathbf{M}^{-1}] \frac{\partial \mathrm{vec}(\mathbf{M})}{\partial \bm{\vartheta}^\top} \mathrm{d}\bm{\vartheta} \nonumber\\
			&&- \frac{1}{2}  \mathrm{d}\bm{\vartheta}^\top\frac{\partial \mathrm{vec}(\frac{\partial \mathrm{vec}(\mathbf{M})^\top}{\partial \bm{\vartheta}})^\top}{\partial \bm{\vartheta}} (\mathrm{vec}(\mathbf{M}^{-1})^\top \otimes \mathbf{I}_{d})  \mathrm{d}\bm{\vartheta} \nonumber\\
			&& -\frac{1}{2}\mathrm{d}\bm{\vartheta}^\top \frac{\partial \mathrm{vec}(\mathbf{M})^\top}{\partial \bm{\vartheta}} [ \mathbf{M}^{-1}(\mathbf{x} - \bm{\mu}) (\mathbf{x} - \bm{\mu})^\top\mathbf{M}^{-1}\otimes\mathbf{M}^{-1}] \frac{\partial \mathrm{vec}(\mathbf{M})}{\partial \bm{\vartheta}^\top} \mathrm{d}\bm{\vartheta} \nonumber\\
			&& - \frac{1}{2}\mathrm{d}\bm{\vartheta}^\top\frac{\partial \bm{\mu}^\top}{\partial \bm{\vartheta}}  (\mathbf{M}^{-1}(\mathbf{x} - \bm{\mu}) \otimes \mathbf{M}^{-1}) \frac{\partial \mathrm{vec}(\mathbf{M})}{\partial \bm{\vartheta}^\top} \mathrm{d}\bm{\vartheta} \nonumber\\
			&&- \frac{1}{2}\mathrm{d}\bm{\vartheta}^\top\frac{\partial \bm{\mu}^\top}{\partial \bm{\vartheta}}  (\mathbf{M}^{-1}\otimes \mathbf{M}^{-1}(\mathbf{x} - \bm{\mu}) ) \frac{\partial \mathrm{vec}(\mathbf{M})}{\partial \bm{\vartheta}^\top} \mathrm{d}\bm{\vartheta} \nonumber\\
			&& -\frac{1}{2}\mathrm{d}\bm{\vartheta}^\top \frac{\partial \mathrm{vec}(\mathbf{M})^\top}{\partial \bm{\vartheta}} [ \mathbf{M}^{-1}\otimes\mathbf{M}^{-1}(\mathbf{x} - \bm{\mu}) (\mathbf{x} - \bm{\mu})^\top\mathbf{M}^{-1}] \frac{\partial \mathrm{vec}(\mathbf{M})}{\partial \bm{\vartheta}^\top} \mathrm{d}\bm{\vartheta} \nonumber\\
			&& + \frac{1}{2} \mathrm{d}\bm{\vartheta}^\top\frac{\partial \mathrm{vec}(\frac{\partial \mathrm{vec}(\mathbf{M})^\top}{\partial \bm{\vartheta}})^\top}{\partial \bm{\vartheta}} (\mathrm{vec}(\mathbf{M}^{-1}(\mathbf{x} - \bm{\mu}) (\mathbf{x} - \bm{\mu})^\top\mathbf{M}^{-1}) \otimes \mathbf{I}_{d})  \mathrm{d}\bm{\vartheta}.
		\end{eqnarray}
		By \eqref{EqB.4}, if $\bm{\varepsilon_t}$ is normal distributed, we have
		\begin{eqnarray*}
			&&\bm{\Sigma}(\tau)=E\left(\gradient_{\bm{\vartheta}}^2\mathcal{l}(\widetilde{\mathbf{x}}_t(\tau),\widetilde{\mathbf{z}}_{t-1}(\tau);\bm{\vartheta}(\tau)) \right)\\
			&=& - E\left(\frac{\partial \bm{\mu}(\widetilde{\mathbf{x}}_t(\tau),\widetilde{\mathbf{z}}_{t-1}(\tau);\bm{\vartheta}(\tau))^\top}{\partial \bm{\vartheta}} \mathbf{M}^{-1}(\widetilde{\mathbf{x}}_t(\tau),\widetilde{\mathbf{z}}_{t-1}(\tau);\bm{\vartheta}(\tau))\frac{\partial \bm{\mu}(\widetilde{\mathbf{x}}_t(\tau),\widetilde{\mathbf{z}}_{t-1}(\tau);\bm{\vartheta}(\tau))}{\partial \bm{\vartheta}^\top}\right) \\
			&& - \frac{1}{2} E\left(\frac{\partial \mathrm{vec}(\mathbf{M}(\widetilde{\mathbf{x}}_t(\tau),\widetilde{\mathbf{z}}_{t-1}(\tau);\bm{\vartheta}(\tau)))^\top}{\partial \bm{\vartheta}} [\mathbf{M}^{-1}(\widetilde{\mathbf{x}}_t(\tau),\widetilde{\mathbf{z}}_{t-1}(\tau);\bm{\vartheta}(\tau))\otimes \mathbf{M}^{-1}(\widetilde{\mathbf{x}}_t(\tau),\widetilde{\mathbf{z}}_{t-1}(\tau);\bm{\vartheta}(\tau))] \right.\\
			&&\left. \times \frac{\partial \mathrm{vec}(\mathbf{M}(\widetilde{\mathbf{x}}_t(\tau),\widetilde{\mathbf{z}}_{t-1}(\tau);\bm{\vartheta}(\tau)))}{\partial \bm{\vartheta}^\top}\right).
		\end{eqnarray*}
		
		Then we have $\bm{\Omega}(\tau) = - \bm{\Sigma}(\tau)$ if $\bm{\varepsilon_t}$ is normal distributed. The proof is now complete.
	\end{proof}
	
	\begin{proof}[Proof of Corollary 2.1]
		\item
		By Lemma \ref{LemmaA.5} (2) and the proof of Theorem 2.1, we have 
		$$
		\sup_{\tau \in [0,1]}|\widehat{\bm{\eta}}(\tau) - \bm{\eta}(\tau)| = O_P((Th)^{-1/2}h^{-1/2}(\log T)^{1/2}).
		$$
		In addition, applying Lemma \ref{LemmaA.3} (4), Lemma \ref{LemmaA.5} (2) and Lemma \ref{LemmaA.3} (2) to $g = \gradient_{\bm{\vartheta}}^2\mathcal{l}$, we have
		$$
		\sup_{\tau \in [0,1]}|\widehat{\bm{\Sigma}}(\tau)-\bm{\Sigma}(\tau)|=O_P((Th)^{-1/2}h^{-1/2}(\log T)^{1/2}+h)=o_P(1)
		$$
		as $h (\log T)^2 \to 0$ and $\gradient_{\bm{\vartheta}}^2\mathcal{l} \in \mathcal{H}(3,\bm{\chi},M)$.
		
		For $\widehat{\bm{\Omega}}(\tau)$, as $\gradient_{\bm{\theta}}\mathcal{l} (\gradient_{\bm{\theta}}\mathcal{l})^\top \in \mathcal{H}(6,\bm{\chi},M)$, here we use a different argument to prove the result, which leads to weaker moment conditions.
		
		By Lemma \ref{LemmaA.3}.4 and Lemma \ref{LemmaB.1}.2, we have
		\begin{eqnarray*}
			&&\widehat{\bm{\Omega}}(\tau) -(Th)^{-1}\sum_{t=1}^{T}\gradient_{\bm{\theta}}\mathcal{l}(\widetilde{\mathbf{x}}_t(\tau_t),\widetilde{\mathbf{z}}_{t-1}(\tau_t);\bm{\theta}(\tau)) \gradient_{\bm{\theta}}\mathcal{l}(\widetilde{\mathbf{x}}_t(\tau_t),\widetilde{\mathbf{z}}_{t-1}(\tau_t);\bm{\theta}(\tau))^\top \widehat{K}\left(\frac{\tau_t-\tau}{h}\right)\\
			&=&O_P((Th)^{-1/2}h^{-1/2}(\log T)^{1/2}+(Th)^{-1})=o_P(1),
		\end{eqnarray*}
		where $\widehat{K}\left(\frac{\tau_t-\tau}{h}\right) = K\left(\frac{\tau_t-\tau}{h}\right)/\left(T^{-1}\sum_{t=1}^{T}K\left(\frac{\tau_t-\tau}{h}\right)\right)$.
		
		Define $g(\widetilde{\mathbf{y}}_t(\tau),\bm{\theta}(\tau)):=\gradient_{\bm{\vartheta}}\mathcal{l}(\widetilde{\mathbf{x}}_t(\tau),\widetilde{\mathbf{z}}_{t-1}(\tau);\bm{\theta}(\tau)) \gradient_{\bm{\vartheta}}\mathcal{l}(\widetilde{\mathbf{x}}_t(\tau),\widetilde{\mathbf{z}}_{t-1}(\tau);\bm{\theta}(\tau))^\top$. By Lemma \ref{LemmaA.4}.1, we have $\sup_{\tau \in [0,1]}\delta_{q/2}^{\sup_{\bm{\vartheta}}|g(\widetilde{\mathbf{y}}_t(\tau),\bm{\vartheta})|}(j) = O(j^{-(3/2+s)})$.
		
		Define $\mathbf{S}_T(\tau) = \sum_{t=1}^{T}\left[g(\widetilde{\mathbf{y}}_t(\tau_t),\bm{\theta}(\tau))-E(g(\widetilde{\mathbf{y}}_t(\tau_t),\bm{\theta}(\tau)))\right]\widehat{K}\left(\frac{\tau_t-\tau}{h}\right)$ and 
		$$
		\mathbf{S}_{k,T}=\sum_{t=1}^{k}\left[g(\widetilde{\mathbf{y}}_t(\tau_t),\bm{\theta}(\tau))-E(g(\widetilde{\mathbf{y}}_t(\tau_t),\bm{\theta}(\tau)))\right].
		$$
		By partial summation, we have
		$$
		\mathbf{S}_T(\tau) = \sum_{t=1}^{T-1}\left[\widehat{K}\left(\frac{\tau_t-\tau}{h}\right)-\widehat{K}\left(\frac{\tau_{t+1}-\tau}{h}\right) \right]\mathbf{S}_{t,T} + \widehat{K}\left(\frac{1-\tau}{h}\right)\mathbf{S}_{T,T}.
		$$
		Hence, we have $\sup_{\tau\in[0,1]}|\mathbf{S}_T(\tau)|\leq M \max_{t}|\mathbf{S}_{t,T}|$. Note that $\{\mathcal{P}_{t-j}g(\widetilde{\mathbf{y}}_t(\tau_t),\bm{\vartheta})\}_t$ forms a sequence of martingale differences. By the Doob's $\mathbb{L}^q$ maximal inequality, Burkholder's inequality and the elementary inequality $(\sum_i|a_i|)^{q}\leq \sum_i|a_i|^q$ for $0<q\leq1$, we obtain that
		\begin{eqnarray*}
			\|\max_{t}|\mathbf{S}_{t,T}|\|_{q/2}&\leq&\sum_{l=0}^{\infty}\|\max_{t=1,...,T}|\sum_{s=1}^{t} \mathcal{P}_{s-l}g(\widetilde{\mathbf{y}}_s(\tau_s),\bm{\theta}(\tau))|\|_{q/2}\\
			&\leq& \sum_{l=0}^{\infty}\frac{q/2}{q/2-1}\|\sum_{s=1}^{T} \mathcal{P}_{s-l}g(\widetilde{\mathbf{y}}_s(\tau_s),\bm{\theta}(\tau))\|_{q/2}\\
			&\leq& \sum_{l=0}^{\infty}\frac{q/2}{(q/2-1)^2}\left[E\left(\sum_{s=1}^{T}\left( \mathcal{P}_{s-l}g(\widetilde{\mathbf{y}}_s(\tau_s),\bm{\theta}(\tau))\right)^2\right)^{q/4}\right]\\
			&\leq& \frac{q/2}{(q/2-1)^2}\sum_{l=0}^{\infty}\left(\sum_{s=1}^{T}\|\mathcal{P}_{s-l}g(\widetilde{\mathbf{y}}_s(\tau_s),\bm{\theta}(\tau))\|_{q/2}^{q/2} \right)^{2/q}\\
			&\leq&  \frac{q/2}{(q/2-1)^2} T^{2/q}\sum_{l=0}^{\infty}\sup_{\tau \in [0,1]}\delta_{q/2}^{\sup_{\bm{\vartheta}}|g(\widetilde{\mathbf{y}}_t(\tau),\bm{\vartheta})|}(l)
		\end{eqnarray*}
		which shows that $\sup_{\tau\in[0,1]}|\frac{1}{Th}\mathbf{S}_T(\tau)|=O_P(T^{2/q-1}h^{-1})=o_P(1)$. The result then follows directly by Lemma \ref{LemmaA.3}.2.
	\end{proof}

	\begin{proof}[Proof of Theorem 2.2]
		\item
		
		We prove this theorem by applying Lemma \ref{LemmaB.2} to the weak Bahadur representation of $\widehat{\bm{\theta}}(\tau)$ given in Lemma \ref{LemmaA.8}.
		
		By Lemma \ref{LemmaA.8}, we have
		\begin{eqnarray}\label{EqB.5}
			\sup_{\tau \in[h,1-h]}&&\left|\mathbf{C}(\widehat{\bm{\theta}}(\tau)-\bm{\theta}(\tau))-\frac{1}{2}h^2\widetilde{c}_2\mathbf{C}\bm{\theta}^{(2)}(\tau)\right. \nonumber\\ &&\left.-\frac{1}{T}\sum_{t=1}^{T}(-\mathbf{C}\bm{\Sigma}^{-1}(\tau))\gradient_{\bm{\vartheta}}\mathcal{l}(\widetilde{\mathbf{x}}_t(\tau_t),\widetilde{\mathbf{z}}_{t-1}(\tau_t);\bm{\theta}(\tau_t))K_h(\tau_t-\tau)\right|\nonumber\\
			&=& O_P(\gamma_T+\beta_Th^2+h^3+(Th)^{-1})=o_P((Th\log T)^{-1/2})
		\end{eqnarray}
		as $Th^7\log T \to 0$ and $Th^2/(\log T)^4 \to \infty$. In addition, by Lemmas \ref{LemmaB.1}, \ref{LemmaA.4}.2 and \ref{LemmaB.2}, we have
		\begin{eqnarray}\label{EqB.6}
			\lim_{T\to\infty}\mathrm{Pr}&&\left(\sqrt{\frac{Th}{\widetilde{v}_0}}\sup_{\tau\in[h,1-h]}\left|\bm{\Sigma}_{\mathbf{C}}^{-1/2}(\tau)\frac{1}{T}\sum_{t=1}^T(-\mathbf{C}\bm{\Sigma}^{-1}(\tau))\gradient_{\bm{\vartheta}}\mathcal{l}(\widetilde{\mathbf{x}}_t(\tau_t),\widetilde{\mathbf{z}}_{t-1}(\tau_t);\bm{\theta}(\tau_t))K_h(\tau_t-\tau) \right| \right. \nonumber\\
			&& \left.- B(1/h)\leq\frac{u}{\sqrt{2\log(1/h)}} \right) = \exp(-2\exp(-u)).
		\end{eqnarray}
		
		By \eqref{EqB.5} and \eqref{EqB.6}, the proof is complete.
	\end{proof}

	\begin{proof}[Proof of Corollary 2.2]
		\item
		By the proof of Lemma \ref{LemmaA.8}, we have
		\begin{eqnarray*}
			&&\sup_{\tau \in [0,1]}\left| \widehat{\bm{\eta}}(\tau)-\bm{\eta}(\tau) - \frac{1}{2}h^2\left[\begin{matrix}
				\widetilde{c}_{0,h}(\tau) & \widetilde{c}_{1,h}(\tau)\\
				\widetilde{c}_{1,h}(\tau) & \widetilde{c}_{2,h}(\tau)
			\end{matrix}\right]^{-1}\left[\begin{matrix}
				\widetilde{c}_{2,h}(\tau) \\
				\widetilde{c}_{3,h}(\tau)
			\end{matrix} \right] \otimes \bm{\theta}^{(2)}(\tau) \right.\\
			&&\left.-\frac{1}{T}\sum_{t=1}^{T}K_h(\tau_t-\tau)\left[\begin{matrix}
				\widetilde{c}_{0,h}(\tau) & \widetilde{c}_{1,h}(\tau)\\
				\widetilde{c}_{1,h}(\tau) & \widetilde{c}_{2,h}(\tau)
			\end{matrix}\right]^{-1}\left[\begin{matrix}
				1 \\
				\frac{\tau_t-\tau}{h}
			\end{matrix} \right]\otimes(-\bm{\Sigma}^{-1}(\tau)) \gradient_{\bm{\vartheta}}\mathcal{l}(\widetilde{\bm{x}}_t(\tau_t),\widetilde{\bm{z}}_{t-1}(\tau_t);\bm{\theta}(\tau_t))  \right|\\
			&=& O_P((Th)^{-1/2}h^{3/2}(\log T)^{1/2}) + O(h^3). 
		\end{eqnarray*}
		Hence, we have
		\begin{eqnarray*}
			&&\sup_{\tau \in [0,1]}\left| \widehat{\bm{\theta}}(\tau)-\bm{\theta}(\tau) - \frac{1}{2}h^2b_h(\tau) \bm{\theta}^{(2)}(\tau) -\frac{1}{T}\sum_{t=1}^{T} (-\bm{\Sigma}^{-1}(\tau)) \gradient_{\bm{\vartheta}}\mathcal{l}(\widetilde{\bm{x}}_t(\tau_t),\widetilde{\bm{z}}_{t-1}(\tau_t);\bm{\theta}(\tau_t)) \omega_{t,h}(\tau) \right|\\
			&=& O_P((Th)^{-1/2}h^{3/2}(\log T)^{1/2}) + O(h^3). 
		\end{eqnarray*}
		
		By Lemma \ref{LemmaB.3}, there exists i.i.d. $k$-dimensional standard normal variables $\mathbf{v}_1,...,\mathbf{v}_T$ such that
		\begin{eqnarray*}
			&&\sup_{\tau \in [0,1]}\left|\frac{1}{Th}\sum_{t=1}^{T}\omega_{t,h}(\tau) (\gradient_{\bm{\theta}}\mathcal{l}(\widetilde{\bm{x}}_t(\tau_t),\widetilde{\bm{z}}_{t-1}(\tau_t);\bm{\theta}(\tau_t)) - \bm{\Omega}^{1/2}(\tau_t)\mathbf{v}_t) \right|\\
			&=&O_P\left(\frac{T^{\frac{q(s+3)-4}{2q(2s+3)-4}}(\log T)^{\frac{2(s+1)(q+1)}{q(2s+3)-2}}}{Th}\right) = O_P\left(\frac{(\log T)^{2}(hT^{\frac{qs+2}{q(2s+3)-2}})^{-1/2}}{(Th)^{1/2}(\log T)^{1/2}} \right)\\
			&=&O_P\left(\frac{(\log T)^{2}(hT^{\nu})^{-1/2}}{(Th \log T)^{1/2}} \right)
		\end{eqnarray*}
		with $\nu = \frac{qs+2}{q(2s+3)-2}$. Since $\bm{\Omega}(\tau)$ is Lipschitz continuous and $\{\mathbf{v}_t\}_{t=1}^T$ is a sequence of i.i.d. normal variables, we have
		\begin{eqnarray*}
			&&\sup_{\tau \in [0,1]}\left|\frac{1}{Th}\sum_{t=1}^{T}\omega_{t,h}(\tau) ( \bm{\Omega}^{1/2}(\tau) - \bm{\Omega}^{1/2}(\tau_t))\mathbf{v}_t \right|\\
			&=&O_P\left(\frac{h (\log T)^{1/2}}{(Th)^{1/2}}\right) = O_P\left(\frac{h \log T}{(Th\log T)^{1/2}}\right).
		\end{eqnarray*}
		
		Combining the above analyses, we then complete the proof.
		
	\end{proof}
	
	\begin{proof}[Proof of Proposition 2.3]
		\item
		Note that in this case $\mathcal{l},\gradient\mathcal{l}, \gradient^2\mathcal{l}$ is in class $\mathcal{H}(2,\bm{\chi},M)$ as $\mathbf{H}(\mathbf{z};\bm{\theta}) \equiv \mathbf{H}(\mathbf{0};\bm{\theta})$ by Lemma \ref{LemmaA.2}. Hence, we only need that the innovation process has $4 + s$ moments for some $s > 0$ compared to $6+s$ moments needed in Theorem 2.2.
		
		Consider Assumptions 1--2 first. For notation simplicity, we ignore the time-varying intercept, and rewrite model (2.11) as
		$$
		\widetilde{\mathbf{y}}_t(\tau) = \bm{\Gamma}(\tau) \widetilde{\mathbf{y}}_{t-1}(\tau) + \widetilde{\mathbf{u}}_t(\tau),
		$$
		where $\widetilde{\mathbf{y}}_t(\tau) = [\widetilde{\bm{\eta}}_t^\top(\tau),...,\widetilde{\bm{\eta}}_{t-q+1}^\top(\tau),\widetilde{\mathbf{x}}_t^\top(\tau),...,\widetilde{\mathbf{x}}_{t-p+1}^\top(\tau)]^\top$, $\widetilde{\mathbf{u}}_t(\tau) = [\widetilde{\mathbf{x}}_t^\top(\tau),\mathbf{0}_{m(q-1)\times1}^\top,\widetilde{\mathbf{x}}_t^\top(\tau),\mathbf{0}_{m(p-1)\times1}^\top]^\top$ and
		$$
		\bm{\Psi}(\tau)=\left[
		\begin{array}{cccc  cccc}
			-\mathbf{B}_1(\tau)& \cdots  &  -\mathbf{B}_{q-1}(\tau)  &  -\mathbf{B}_{q}(\tau) & -\mathbf{A}_1(\tau)& \cdots  &  -\mathbf{A}_{p-1}(\tau)  &  -\mathbf{A}_{p}(\tau) \\
			\mathbf{I}_m & \cdots  &  \mathbf{0}_m       &  \mathbf{0}_m    & \mathbf{0}_{m}     & \cdots  &  \mathbf{0}_{m}    &  \mathbf{0}_{m}      \\
			\vdots   & \ddots  &  \vdots  &\vdots &\vdots & \ddots  &\vdots &  \vdots \\
			\mathbf{0}_m        & \cdots  &  \mathbf{I}_m&  \mathbf{0}_m    & \mathbf{0}_{m}     & \cdots  &  \mathbf{0}_{m}    &  \mathbf{0}_{ m}      \\
			\multicolumn{4}{c}{\multirow{4}{*}{$\mathbf{0}$}}& \mathbf{0}_m     & \cdots  &  \mathbf{0}_m   &  \mathbf{0}_m      \\
			\multicolumn{4}{c}{}& \mathbf{I}_{m}     & \cdots  &  \mathbf{0}_m    &  \mathbf{0}_m      \\
			\multicolumn{4}{c}{}& \vdots         & \ddots  &  \vdots      &  \vdots      \\
			\multicolumn{4}{c}{}& \mathbf{0}_m       & \cdots  &  \mathbf{I}_{m}  &  \mathbf{0}_m      \\
		\end{array}
		\right].
		$$
		Let $\mathbf{J} = [\mathbf{I}_m,\mathbf{0}_{m\times(m(p+q-1))}]$ and $\mathbf{H} = [\mathbf{I}_{m},\mathbf{0}_{m(q-1)\times1}^\top,\mathbf{I}_{m},\mathbf{0}_{m(p-1)\times1}^\top]^\top$, we have
		\begin{eqnarray*}
			\widetilde{\bm{\eta}}_t(\tau) &=& \widetilde{\mathbf{x}}_t(\tau) + \sum_{j=1}^{\infty}(\mathbf{J} \bm{\Psi}^j(\tau) \mathbf{H})\widetilde{\mathbf{x}}_{t-j}(\tau)\\
			\widetilde{\mathbf{x}}_t(\tau) &=&  \sum_{j=1}^{\infty}(-\mathbf{J} \bm{\Psi}^j(\tau) \mathbf{H})\widetilde{\mathbf{x}}_{t-j}(\tau) + \widetilde{\bm{\eta}}_t(\tau)
		\end{eqnarray*}
		and thus $\bm{\Gamma}_j(\tau) = - \mathbf{J} \bm{\Psi}^j(\tau) \mathbf{H}$. Then Assumption 1 is automatically met if $\sum_{j=1}^{\infty}|\bm{\Gamma}_j(\tau)| < 1$. By using the property of block matrix determinants and $\mathrm{det}(\mathbf{B}_\tau(L)) \neq 0$ for all $|L|\leq1$ (this implies the maximum eigenvalue of left upper $mq\times mq$ matrix in $\bm{\Gamma}(\tau)$ is less than $1$), it can be shown that the maximum eigenvalue of $\bm{\Gamma}(\tau)$, denoted by $\rho$, is less than 1 uniformly over $\tau\in[0,1]$. Hence, we have $\alpha_j(\bm{\theta}(\tau)) = |\bm{\Gamma}_j(\tau)| = O(\rho^j)$ and $\beta_j(\bm{\theta}(\tau)) = 0$. In addition, $|\bm{\Psi}^j-\bm{\Psi}^{\prime j}| = |\sum_{i=1}^{j-1}\bm{\Psi}^i(\bm{\Psi} - \bm{\Psi}')\bm{\Psi}^{\prime j-1-i} | = |\bm{\Psi} - \bm{\Psi}'|O(\rho^{j-1})$. Then Assumption 2 is met.
		
		However, by using techniques which are more specific to the VARMA models, the condition $\sum_{j=1}^{\infty}|\bm{\Gamma}_j(\tau)| < 1$ can be weakened to $\mathrm{det}(\mathbf{A}_\tau(L)) \neq 0$ for all $|L|\leq1$. Similar to the above analysis, we have $\widetilde{\mathbf{x}}_t(\tau) =  \sum_{j=0}^{\infty}\bm{\Phi}_{j}(\tau)\widetilde{\bm{\eta}}_{t-j}(\tau)$ with $|\bm{\Phi}_{j}(\tau)| = O(\rho^{j})$ as $\mathrm{det}(\mathbf{A}_\tau(L)) \neq 0$ for all $|L|\leq1$, which implies that $\| \widetilde{\mathbf{x}}_t(\tau)\|_r < \infty$ and $\delta_{r}^{\widetilde{\mathbf{x}}(\tau)}(k) = O(\rho^k)$.
		
		For the identification conditions stated in Assumption 3, it is well known that the final form or echelon form is enough to ensure the uniqueness of the VARMA representation.
		
		For verifying Assumption 4, one need the derivatives of $\bm{\Gamma}_j$. Define $\bm{\alpha}= - \mathrm{vec}(\mathbf{B}_1,...,\mathbf{B}_{q}, \mathbf{A}_1,..., \mathbf{A}_{p})$. Note that $\mathrm{d}\mathrm{vec}(\bm{\Psi}) = (\mathbf{I}_{m(p+q)}\otimes \mathbf{J}^\top)\mathrm{d}\bm{\alpha}$ and $\mathrm{d}\mathrm{vec}(\bm{\Psi}^j) = (\bm{\Psi}^\top\otimes\mathbf{I}_{m(p+q)})\mathrm{d}\mathrm{vec}(\bm{\Psi}^{j-1})+(\mathbf{I}_{m(p+q)}\otimes \bm{\Psi}^{j-1}\mathbf{J}^\top)\mathrm{d}\bm{\alpha}$, it is easy to show that
		\begin{equation*}
			\frac{\partial \mathrm{vec}(\bm{\Gamma}_j)}{\partial \bm{\alpha}^\top} = - \sum_{i=0}^{j-1} \mathbf{H}^\top (\bm{\Psi}^{\top})^{j-1-i} \otimes \mathbf{J} \bm{\Psi}^j(\tau) \mathbf{J}^\top.
		\end{equation*}
		Hence, we have $|\frac{\partial \mathrm{vec}(\bm{\Gamma}_j)}{\partial \bm{\alpha}^\top}|=O(\rho^{j-1})$ and $|\frac{\partial \mathrm{vec}(\bm{\Gamma}_j)}{\partial \bm{\alpha}^\top} - \frac{\partial \mathrm{vec}(\bm{\Gamma}_j^\prime)}{\partial \bm{\alpha}^{\prime\top}}|= |\bm{\Psi} - \bm{\Psi}'|O(\rho^{j-2})$. Similarly, we can verify the conditions imposed on second order derivatives.
		
		The proof is now complete.
		
	\end{proof}
	
	\begin{proof}[Proof of Proposition 2.4]
		\item
		Since $\mathbf{H}(\mathbf{z};\bm{\theta})$ is a positive and diagonal matrix, we have
		\begin{eqnarray*}
			|\mathbf{H}(\mathbf{z};\bm{\theta}) - \mathbf{H}(\mathbf{z}^\prime;\bm{\theta})|  & = & |(\mathbf{H}^2(\mathbf{z};\bm{\theta}) - \mathbf{H}^2(\mathbf{z}^\prime;\bm{\theta}))\cdot (\mathbf{H}(\mathbf{z};\bm{\theta}) + \mathbf{H}(\mathbf{z}^\prime;\bm{\theta}))^{-1}| \\
			&\leq & \sum_{j=1}^{\infty}|\bm{\Psi}_j(\tau)|^{1/2}\cdot |\mathbf{x}_j - \mathbf{x}_j^\prime|.
		\end{eqnarray*}
		Then Assumption 1 is automatically met if $\|\widetilde{\bm{\eta}}_t(\tau)\|_r\sum_{j=1}^{\infty}|\bm{\Psi}_j(\tau)|^{1/2}<1$. In addition, as $|\bm{\Psi}_j(\tau)|$ converges to zero with exponential rate and ${\partial h_{i,t}^{1/2}}/{\partial  \theta_i} = \frac{1}{2}h_{i,t}^{-1/2}{\partial h_{i,t}}/{\partial  \theta_i}$, similar to the proof of Proposition 2.3 we can easily verify Assumptions 2 and 4. For the identification conditions of the GARCH process, we refer readers to Proposition 3.4 of \cite{jeantheau1998strong}, who proves that assuming the minimal representation is enough for ensuring Assumption 3 holds.
		
		However, by using techniques which are more specific to the GARCH models, the condition $\|\widetilde{\bm{\eta}}_t(\tau)\|_r\sum_{j=1}^{\infty}|\bm{\Psi}_j(\tau)|^{1/2}<1$ can be weaken to $\|\widetilde{\bm{\eta}}_t(\tau)\|_{r}^2\sum_{j=1}^{\infty}|\bm{\Psi}_j(\tau)|<1$. Define $\widetilde{\mathbf{y}}_t(\tau) = \widetilde{\mathbf{x}}_t(\tau) \odot \widetilde{\mathbf{x}}_t(\tau)$ and $\widetilde{\mathbf{V}}_t(\tau) = \mathrm{diag}\left(\widetilde{\bm{\eta}}_t(\tau)\odot \widetilde{\bm{\eta}}_t(\tau)\right)$. We first prove the existence of $\|\widetilde{\mathbf{y}}_t(\tau)\|_{r/2}$ (which implies the existence of $\|\widetilde{\mathbf{x}}_t(\tau)\|_{r}$) as well as its weak dependence property by means of a chaotic expansion. Since $\widetilde{\mathbf{y}}_t(\tau) = \widetilde{\mathbf{V}}_t(\tau)\bm{\alpha}(\tau) + \sum_{j=1}^{\infty}\widetilde{\mathbf{V}}_t(\tau)\bm{\Psi}_j(\tau)\widetilde{\mathbf{y}}_{t-j}(\tau)$, by substitute $\widetilde{\mathbf{y}}_{t-j}(\tau)$ recursively, we have
		\begin{equation*}
			\widetilde{\mathbf{y}}_t(\tau) = \widetilde{\mathbf{V}}_t(\tau)\left\{\bm{\alpha}(\tau) + \sum_{k=1}^{\infty}\sum_{j_1,..,j_k=1}^{\infty}\mathbf{\Psi}_{j_1}(\tau)\widetilde{\mathbf{V}}_{t-j_1}(\tau)\cdots\mathbf{\Psi}_{j_k}(\tau)\widetilde{\mathbf{V}}_{t-j_1-\cdots-j_k}(\tau)\bm{\alpha}(\tau)\right\}.
		\end{equation*}
		To prove the boundedness of $\left\|\widetilde{\mathbf{y}}_t(\tau)\right\|_{r/2}$, since $\{\widetilde{\mathbf{V}}_t(\tau)\}$ are independent random variables, it suffices to show that
		\begin{equation*}
			\sum_{k=1}^{\infty}\sum_{j_1,..,j_k=1}^{\infty}\left\|\bm{\Psi}_{j_1}(\tau)\widetilde{\mathbf{V}}_{t-j_1}(\tau)\cdots\bm{\Psi}_{j_k}(\tau)\widetilde{\mathbf{V}}_{t-j_1-\cdots-j_k}(\tau)\bm{\alpha}(\tau)\right\|_{r/2} < \infty.
		\end{equation*}
		By using $\sup_{\tau \in [0,1]}|\bm{\alpha}(\tau)| < \infty$ and $\|\widetilde{\mathbf{V}}_t(\tau)\|_{r/2}\sum_{j=1}^{\infty}|\bm{\Psi}_j(\tau)|<1$, we have
		\begin{eqnarray*}
			&&\sum_{k=1}^{\infty}\sum_{j_1,..,j_k=1}^{\infty}\left\|\bm{\Psi}_{j_1}(\tau)\widetilde{\mathbf{V}}_{t-j_1}(\tau)\cdots\bm{\Psi}_{j_k}(\tau)\widetilde{\mathbf{V}}_{t-j_1-\cdots-j_k}(\tau)\bm{\alpha}(\tau)\right\|_{r/2} \\
			&\leq& \sup_{\tau \in [0,1]}|\bm{\alpha}(\tau)| \sum_{k=1}^{\infty}\sum_{j_1,..,j_k=1}^{\infty}|\bm{\Psi}_{j_1}(\tau)|\cdots |\bm{\Psi}_{j_k}(\tau)| \|\widetilde{\mathbf{V}}_t(\tau)\|_{r/2}^k\\
			&\leq& \sup_{\tau \in [0,1]}|\bm{\alpha}(\tau)| \sum_{k=1}^\infty\left(\|\widetilde{\mathbf{V}}_t(\tau)\|_{r/2}\sum_{j=1}^{\infty}|\bm{\Psi}_j(\tau)|\right)^k < \infty.
		\end{eqnarray*}
		Hence, we have $\|\widetilde{\bm{x}}_t(\tau)\|_r < \infty$. Next, we show that $\delta_{r}^{\widetilde{\mathbf{x}}(\tau)}(k) = O(\rho^k)$ for some $0<\rho <1$. Write
		\begin{equation*}
			\widetilde{\mathbf{y}}_t(\tau) = \mathrm{diag}\left(\bm{\alpha}(\tau) + \sum_{j=1}^{\infty}\bm{\Psi}_j(\tau)\widetilde{\mathbf{y}}_{t-j}(\tau)\right)(\widetilde{\bm{\eta}}_t(\tau) \odot \widetilde{\bm{\eta}}_t(\tau) ).
		\end{equation*}
		By using the same arguments as in the proof of Proposition 2.1, we have $\delta_{r}^{\widetilde{\mathbf{y}}(\tau)}(k) = O(\rho^k)$ since $\|\widetilde{\bm{\eta}}_t(\tau)\|_{r}^2\sum_{j=1}^{\infty}|\bm{\Psi}_j(\tau)|<1$ and $|\bm{\Psi}_j(\tau)|=O(\rho^j)$. Since $|a-b|\leq |a^2-b^2|^{1/2}$ for $a\geq 0, b\geq 0$ and $\mathbf{H}(\cdot)$ is a positive diagonal matrix, for $t\geq 1$, we have 
		\begin{eqnarray*}
			\|\widetilde{\bm{x}}_t(\tau) - \widetilde{\bm{x}}_t^*(\tau)\|_r &\leq& \sum_{j=1}^{t}|\bm{\Psi}_j(\tau)|^{1/2}\cdot\|\widetilde{\mathbf{y}}_{t-j}(\tau)-\widetilde{\mathbf{y}}_{t-j}^*(\tau)\|_{r/2}^{1/2} \cdot\|\widetilde{\bm{\eta}}_t(\tau)\|_r\\
			&=& \sum_{j=1}^{t} O(\rho^{j/2})O(\rho^{(t-j)/2}) = O(\rho^{\prime t})
		\end{eqnarray*}
		for some $0<\rho^{\prime}<1$.
		
		The proof is now completed.
		
	\end{proof}
	
	%
	
	\newpage
	
	\section*{Appendix B}\label{App.B}
	
	\renewcommand{\theequation}{B.\arabic{equation}}
	\renewcommand{\thesubsection}{B.\arabic{subsection}}
	\renewcommand{\thefigure}{B.\arabic{figure}}
	\renewcommand{\thetable}{B.\arabic{table}}
	\renewcommand{\thelemma}{B.\arabic{lemma}}
	
	\setcounter{equation}{0}
	\setcounter{lemma}{0}
	\setcounter{section}{0}
	\setcounter{table}{0}
	\setcounter{figure}{0}

	\subsection{Preliminary Lemmas}\label{AppB.1}
	
	First, we define a few notations for better presentation. First, let $\bm{\eta} =(\bm{\eta}_1^\top, \bm{\eta}_2^\top)^\top$, where $\bm{\eta}_1$ and $\bm{\eta}_2$ are the same generic vectors as in (2.4). Let $\widehat{K}(\cdot)$ be a a kernel function being Lipschitz continuous and bounded on $[-1,1]$. 
	
	For $\tau \in[0,1]$ and $\bm{\eta} \in \mathbf{E}_T(r) = \bm{\Theta}_r\times(h\cdot\bm{\Theta}^{(1)})$, define
	\begin{equation}\label{EqA.3}
		G_\tau(\bm{\eta}):=\frac{1}{Th}\sum_{t=1}^{T}\widehat{K}\left(\frac{\tau_t-\tau}{h}\right)\left[g(\mathbf{y}_t,\bm{\eta}_1+\bm{\eta}_2\cdot(\tau_t-\tau)/h)-E(g(\mathbf{y}_t,\bm{\eta}_1+\bm{\eta}_2\cdot(\tau_t-\tau)/h))\right],
	\end{equation}
	where $g(\cdot)\in\mathcal{H}(C,\bm{\chi},M)$ and $\mathbf{y}_t=(\mathbf{x}_t,\mathbf{z}_{t-1})$. Let $G_\tau^c(\bm{\eta})$, $\widetilde{G}_\tau(\bm{\eta})$ denote the same quantity but with $\mathbf{y}_t$ replaced by $\mathbf{y}_t^c=(\mathbf{x}_t,\mathbf{z}_{t-1}^c)$ or $\widetilde{\mathbf{y}}_t(\tau_t)=(\widetilde{\mathbf{x}}_t(\tau_t),\widetilde{\mathbf{z}}_{t-1}(\tau_t))$.
	
	In addition, let 
	\begin{eqnarray}
		\widetilde{B}_\tau(\bm{\eta})&:=&\frac{1}{Th}\sum_{t=1}^{T}\widehat{K}\left(\frac{\tau_t-\tau}{h}\right)g(\widetilde{\mathbf{y}}_t(\tau_t),\bm{\eta}_1+\bm{\eta}_2\cdot(\tau_t-\tau)/h),\nonumber \\
		\mathbf{m}_t^{(2)}(u,\tau)&:=&\widehat{K}\left(\frac{\tau-u}{h}\right)g(\widetilde{\mathbf{y}}_t(u),\bm{\theta}(\tau)-v\mathbf{d}(u, \tau))\cdot \mathbf{d}(u, \tau),
	\end{eqnarray}
	where $\mathbf{d}(u, \tau):=\bm{\theta}(u)-\bm{\theta}(\tau)-(u-\tau)\bm{\theta}^{(1)}(\tau)$ and some $v\in[0,1]$.
	
	\medskip
	
	\begin{lemma}\label{LemmaA.1}
		Suppose Assumptions 1 and 3 hold. Then, $ E\left(\mathcal{l}(\widetilde{\mathbf{x}}_1(\tau),\widetilde{\mathbf{z}}_{0}(\tau);\bm{\vartheta})\right)$ is uniquely maximized at $\bm{\theta}(\tau)$.
	\end{lemma}
	
	\begin{lemma}\label{LemmaA.2}
		Suppose Assumptions 3--4 hold. Then, $\mathcal{l},\gradient\mathcal{l}, \gradient^2\mathcal{l}\in\mathcal{H}(3,\bm{\chi},M)$ for some $M > 0 $ and $\bm{\chi}=\{\chi_{j}\}_{j=1,2,\ldots}$ with $\chi_j =O(j^{-(2+s)})$ and $s > 0$. In addition, if $\mathbf{H}(\mathbf{z};\bm{\vartheta}) \equiv \mathbf{H}(\mathbf{0};\bm{\vartheta})$, $\mathcal{l},\gradient\mathcal{l}, \gradient^2\mathcal{l}\in \mathcal{H}(2,\bm{\chi},M)$.
	\end{lemma}

	\begin{lemma} \label{LemmaA.3}
		Suppose Assumptions 1--2 hold with $r\geq C$. Then 
		\begin{enumerate}
			\item [1.] $\sup_{\tau \in [0,1]}\left\|\sup_{\bm{\eta}\neq\bm{\eta}^\prime}\frac{|\widetilde{G}_\tau(\bm{\eta})-\widetilde{G}_\tau(\bm{\eta}^\prime)|}{|\bm{\eta}-\bm{\eta}^\prime|}\right\|_1 \leq M$ and $\left\|\sup_{\tau\neq\tau^\prime,\bm{\eta}\neq\bm{\eta}^\prime}\frac{|\widetilde{G}_\tau(\bm{\eta})-\widetilde{G}_{\tau^\prime}(\bm{\eta}^\prime)|}{|\tau-\tau^\prime|+|\bm{\eta}-\bm{\eta}^\prime|}\right\|_1 \leq Mh^{-2}$;
			
			\item [2.] $\sup_{\tau \in [0,1],\bm{\eta} \in \mathbf{E}_T(r)}|E(\widetilde{B}_\tau(\bm{\eta}))-\int_{-\tau/h}^{(1-\tau)/h}\widehat{K}(u)E(g(\widetilde{\mathbf{y}}_0(\tau),\bm{\eta}_1+\bm{\eta}_2u))\mathrm{d}u |=O((Th)^{-1}+h)$;
			
			\item [3.] $\left\|\sup_{\tau\neq\tau^\prime}\frac{|\bm{\Pi}(\tau)-\bm{\Pi}(\tau^\prime)|}{|\tau-\tau^\prime|}\right\|_1 \leq M h^{-2}$ with $\bm{\Pi}(\tau) :=(Th)^{-1}\sum_{t=1}^{T}\left[\mathbf{m}_t^{(2)}(\tau,\tau_t) -E(\mathbf{m}_t^{(2)}(\tau,\tau_t))\right]$.
		\end{enumerate}
		In addition, suppose $\chi_j = O(j^{-(2+s)})$ for some $s>0$, then
		\begin{enumerate}
			\item [4.] $ \|\sup_{\tau \in [0,1], \bm{\eta} \in \mathbf{E}_T(r)}|\widetilde{G}_\tau(\bm{\eta})-G_\tau^c(\bm{\eta})|  \|_1 = O((Th)^{-1})$.
		\end{enumerate}
	\end{lemma}

	\begin{lemma} \label{LemmaA.4}
		Let $g(\cdot)\in\mathcal{H}(C,\bm{\chi},M)$, where $\chi_j = O(j^{-(a+s)})$ for some $s>0$ and $a \geq 1$. Suppose Assumptions 1--2 hold with $q= r/C\geq 1$, and
		
		\begin{eqnarray*}
			\sup_{\tau \in [0,1]} [\alpha_j(\bm{\theta}(\tau)) +\beta_j(\bm{\theta}(\tau)) ]= O(j^{-(a+1+s)}).
		\end{eqnarray*}
		Then we obtain
		
		\begin{enumerate}
			\item $\sup_{\tau \in [0,1]}\delta_{q}^{\sup_{\bm{\vartheta}}|g(\widetilde{\mathbf{y}}_t(\tau),\bm{\vartheta})|}(j) = O(j^{-(a+s)})$;
			
			\item $\sup_{\tau \in [0,1]}\sup_{u,\bm{\eta}}\delta_{q}^{m(\tau,\bm{\eta},u)}(j) = O(j^{-(a+s)})$ and $\sup_{\tau \in [0,1]}\delta_{q}^{\sup_{u,\bm{\eta}}|m(\tau,\bm{\eta},u)|}(j) = O(j^{-(a+s)})$, where $m_t(\tau,\bm{\eta},u):=\widehat{K}\left(\frac{\tau-u}{h}\right)g(\widetilde{\mathbf{y}}_t(\tau),\bm{\eta}_1+\bm{\eta}_2\cdot(\tau-u)/h)$;
			
			\item $\sup_{\tau,u \in [0,1]}\delta_{q}^{\mathbf{m}_t^{(2)}(u,\tau)}(j) = O(h^2j^{-(a+s)})$ and $\sup_{u \in [0,1]}\delta_{q}^{\sup_{\tau}|\mathbf{m}_t^{(2)}(u, \tau)|}(j) = O(h^2j^{-(a+s)})$.
		\end{enumerate}
	\end{lemma}

	\begin{lemma} \label{LemmaA.5}
		Under the conditions of Lemma \ref{LemmaA.4} with $q=r/C > 1$, then 
		
		\begin{enumerate}
			\item $\|\widetilde{G}_\tau(\bm{\eta})\|_q = O\left((Th)^{-(q'-1)/q'}\right)$ with $q' = \min(2,q)$,
			
			\item $\sup_{\bm{\eta} \in \mathbf{E}_T(r)}|\widetilde{G}_\tau(\bm{\eta})| =o_P(1)$;
		\end{enumerate}
		Suppose further $q = r/C >2$ and $a \geq 3/2$. Then 
		\begin{enumerate}
			\item[3.] $\sup_{\tau\in[0,1]} \sup_{\bm{\eta} \in \mathbf{E}_T(r)} |\widetilde{G}_\tau(\bm{\eta})| = O_P( (\log T)^{1/2}(Th)^{-1/2}h^{-1/2})$.
		\end{enumerate}
	\end{lemma}

	\begin{lemma} \label{LemmaA.7}
		Suppose Assumptions 1--5 hold with $r > 6$, and
		
		\begin{eqnarray*}
			\sup_{\tau \in [0,1]} [\alpha_j(\bm{\theta}(\tau)) +\beta_j(\bm{\theta}(\tau)) ]= O(j^{-(a+1+s)})
		\end{eqnarray*}
		for $a\geq3/2$ and some $s > 0$. Then
		\begin{eqnarray*}
			\sup_{\tau\in[0,1]}&&\Big|\gradient \widetilde{\mathcal{L}}_\tau(\bm{\theta}(\tau),h\bm{\theta}^{(1)}(\tau)) -E[\gradient \widetilde{\mathcal{L}}_\tau(\bm{\theta}(\tau),h\bm{\theta}^{(1)}(\tau))]\\
			&&-\frac{1}{Th}\sum_{t=1}^{T}\widehat{\bm{K}}((\tau_t-\tau)/h)\otimes \gradient_{\bm{\vartheta}}\mathcal{l}(\widetilde{\bm{x}}_t(\tau_t),\widetilde{\bm{z}}_{t-1}(\tau_t);\bm{\theta}(\tau_t))\Big|=O_P(h^2\beta_T),
		\end{eqnarray*}
		where 
		\begin{eqnarray*}
			\beta_T &=&(\log T)^{1/2}(Th)^{-1/2}h^{-1/2},\\
			\widehat{\bm{K}}((\tau_t-\tau)/h)&=&K((\tau_t-\tau)/h) [1,(\tau_t-\tau)/h ]^\top,\\
			\widetilde{\mathcal{L}}_\tau(\bm{\theta}(\tau),h\bm{\theta}^{(1)}(\tau))&:=&T^{-1}\sum_{t=1}^{T}\mathcal{l}(\widetilde{\bm{x}}_t(\tau_t),\widetilde{\bm{z}}_{t-1}(\tau_t);\bm{\theta}(\tau)+\bm{\theta}^{(1)}(\tau)(\tau_t-\tau))K_h(\tau_t-\tau).
		\end{eqnarray*}
	\end{lemma}

	\begin{lemma} \label{LemmaA.8}
		Under the conditions of Theorem 2.2,
		\begin{eqnarray*}
			(1). && \sup_{\tau\in[h,1-h]}\left|-\bm{\Sigma}(\tau)(\widehat{\bm{\theta}} (\tau)-\bm{\theta}(\tau))-\gradient_{\bm{\vartheta}}\mathcal{L}_\tau(\bm{\theta}(\tau),h\bm{\theta}^{(1)}(\tau)) \right| = O_P(\gamma_T),\\
			(2). && \sup_{\tau\in[h,1-h]} \left|\gradient_{\bm{\vartheta}}\mathcal{L}_\tau(\bm{\theta}(\tau),h\bm{\theta}^{(1)}(\tau)) + \frac{1}{2}h^2\widetilde{c}_2\bm{\Sigma}(\tau)\bm{\theta}^{(2)}(\tau)\right.\\ &&\left.-\frac{1}{T}\sum_{t=1}^{T}\gradient_{\bm{\vartheta}}\mathcal{l}(\widetilde{\mathbf{x}}_t(\tau_t),\widetilde{\mathbf{z}}_{t-1}(\tau_t);\bm{\theta}(\tau_t))K_h(\tau_t-\tau) \right| = O_P(\beta_Th^2+h^3+(Th)^{-1}),
		\end{eqnarray*}
		where $\beta_T = (Th)^{-1/2}h^{-1/2}(\log T)^{1/2}$ and $\gamma_T = (\beta_T + h)((Th)^{-1/2}\log T+h^2)$.
	\end{lemma}
	
	\subsection{Secondary Lemmas}\label{AppB.2}
	
	Before proceeding further, we introduce some extra notations. Assume that there exists some measurable function $\widetilde{\mathbf{H}}(\cdot,\cdot)$ such that for $\forall\tau \in [0,1]$, $\widetilde{\mathbf{h}}_t(\tau) = \widetilde{\mathbf{H}}(\tau,\mathcal{F}_t) \in \mathbb{R}^{d}$ is well defined, where $\mathcal{F}_t=\sigma(\bm{\varepsilon}_t,\bm{\varepsilon}_{t-1},\ldots)$. Let 
	\begin{eqnarray*}
		\widetilde{\mathbf{D}}_{\widetilde{\mathbf{h}}}(\tau) := (Th)^{-1}\sum_{t=1}^{T}\widetilde{\mathbf{h}}_t(\tau_t)\widehat{K}((\tau_t-\tau)/h) \quad \text{and}\quad \bm{\Sigma}_{\widetilde{\mathbf{h}}}(\tau) = \sum_{j=-\infty}^{\infty}E[\widetilde{\mathbf{h}}_0(\tau)\widetilde{\mathbf{h}}_j^\top(\tau) ].
	\end{eqnarray*}
	Assume that $\bm{\Sigma}_{\widetilde{\mathbf{h}}}(\tau)$ is Lipschitz continuous and its smallest eigenvalue is bounded away from $0$ uniformly over $\tau\in[0,1]$. In what follows, we let $\widetilde{h}_{0,i}(\tau)$ stand for the $i^{th}$ component of $\widetilde{\mathbf{h}}_t(\tau) $.
	
	\begin{lemma}\label{LemmaB.1}
		Let $q > 0$. Let $g \in \mathcal{H}(C,\bm{\chi},M)$. Let $\mathbf{y}=(\mathbf{y}_0,\mathbf{y}_1,\mathbf{y}_2,\ldots) $ and $\mathbf{y}^\prime=(\mathbf{y}_0^\prime,\mathbf{y}_1^\prime,\mathbf{y}_2^\prime,\ldots) $ be two sequences of random variables. Assume that $\max_{j\ge 0}\|\mathbf{y}_j\|_{qC} \leq M $ and $\max_{j\ge 0}\|\mathbf{y}_j^\prime\|_{qC} \leq M$. Then, we have
		\begin{enumerate}
			\item $\left\|\sup_{\bm{\vartheta}\in\bm{\Theta}_r} |g(\mathbf{y},\bm{\vartheta}) - g(\mathbf{y}^\prime,\bm{\vartheta})| \right\|_q\leq M \sum_{j=0}^{\infty}\chi_j\|\mathbf{y}_j-\mathbf{y}_j^\prime\|_{qC}$;
			
			\item $\left\|\sup_{\bm{\vartheta}\neq\bm{\vartheta}^\prime} \frac{|g(\mathbf{y},\bm{\vartheta}) - g(\mathbf{y},\bm{\vartheta}^\prime)|}{|\bm{\vartheta}-\bm{\vartheta}^\prime|} \right\|_q\leq M$;
			
			\item $\left\|\sup_{\bm{\vartheta}\in\bm{\Theta}_r} |g(\mathbf{y},\bm{\vartheta})| \right\|_q\leq M$.
		\end{enumerate}
	\end{lemma}
	
	Lemma \ref{LemmaB.1} is Lemma D.4 of \cite{karmakar2021simultaneous}.
	
	\medskip
	
	\begin{lemma}\label{LemmaB.2}
		Assume that for  $ i =1,\ldots,d$
		
		\begin{enumerate}
			\item $\sup_{\tau\in[0,1]}\|\widetilde{h}_{0,i}(\tau)\|_{q}<\infty$ with some $2\leq q\leq 4$, 
			
			\item $\sup_{\tau\neq\tau^\prime}\|\widetilde{h}_{0,i}(\tau)-\widetilde{h}_{0,i}(\tau^\prime)\|_{2}/|\tau-\tau^\prime|<\infty$,
			
			\item $\sup_{\tau\in[0,1]}\delta_q^{\widetilde{h}_{0,i}(\tau)}(j)=O(j^{-(2+s)})$ for some $s \geq 0$.
		\end{enumerate}
		In addition, assume that $h(\log T)^{3/2}\to 0$ and $\frac{(\log T)^4}{T^{(sq+2)/(2sq+3q-2)}h}\to 0$. Then
		
		\begin{eqnarray*}
			\lim_{T\to\infty}\Pr\left(\sqrt{\frac{Th}{\widetilde{v}_0}}\sup_{\tau\in[h,1-h]}\left|\bm{\Sigma}_{\widetilde{\mathbf{h}}}^{-1/2}(\tau)\widetilde{\mathbf{D}}_{\widetilde{\mathbf{h}}}(\tau) \right| - B(m^*)\leq\frac{u}{\sqrt{2\log(m^*)}}  \right) = \exp(-2\exp(-u)),
		\end{eqnarray*}
		where 
		
		\begin{eqnarray*}
			B(m^*) &=& \sqrt{2 \log(m^*)} + \frac{\log (C_K) + (k/2-1/2)\log(\log(m^*))-\log(2)}{\sqrt{2 \log(m^*)}},\\
			C_K &=& \frac{\{\int_{-1}^{1}|K^{(1)}(u)|^2\mathrm{d}u/\widetilde{v}_0\pi\}^{1/2}}{\Gamma(k/2)},\quad m^* = 1/h,
		\end{eqnarray*}
		and $\Gamma(\cdot)$ is the Gamma function.
	\end{lemma}
	
	Lemma \ref{LemmaB.2} is Lemma B.3 of \cite{karmakar2021simultaneous}.
	
	\medskip
	
	\begin{lemma}\label{LemmaB.3}
		Assume that for $ i =1,\ldots,d$
		
		\begin{enumerate}
			\item $\sup_{\tau\in[0,1]}\|\widetilde{h}_{0,i}(\tau)\|_{q}<\infty$ with some $2\leq q\leq 4$, 
			
			\item $\sup_{\tau\neq\tau^\prime}\|\widetilde{h}_{0,i}(\tau)-\widetilde{h}_{0,i}(\tau^\prime)\|_{2}/|\tau-\tau^\prime|<\infty$,
			
			\item $\sup_{\tau\in[0,1]}\delta_q^{\widetilde{h}_{0,i}(\tau)}(j)=O(j^{-(2+s)})$ for some $s \geq 0$.
		\end{enumerate}
		Let $\mathbf{S}_{\widetilde{\mathbf{h}}}(t) = \sum_{s=1}^{t}\widetilde{\mathbf{h}}_s(\tau_s)$. Then on a richer probability space, there exists i.i.d. k-dimensional standard normal variables $\mathbf{v}_1,\mathbf{v}_2,\ldots$ and a process $\mathbf{S}_{\widetilde{\mathbf{h}}}^{0}(t) = \sum_{s=1}^{t}\bm{\Sigma}_{\widetilde{\mathbf{h}}}^{1/2}(\tau_s)\mathbf{v}_s$ such that 
		\begin{eqnarray*}
			(\mathbf{S}_{\widetilde{\mathbf{h}}}(t))_{t=1}^T =_D (\mathbf{S}_{\widetilde{\mathbf{h}}}^{0}(t))_{t=1}^T\quad \text{and}\quad \max_{t\ge 1} |\mathbf{S}_{\widetilde{\mathbf{h}}}(t) - \mathbf{S}_{\widetilde{\mathbf{h}}}^{0}(t)| = O_P(\pi_T),
		\end{eqnarray*}
		where $\pi_T = T^{\frac{q(s+3)-4}{2q(2s+3)-4}}(\log T)^{\frac{2(s+1)(q+1)}{q(2s+3)-2}}.$
	\end{lemma}
	
	Lemma \ref{LemmaB.3} is from Theorem 1 and Corollary 2 of \cite{wu2011gaussian}.

	\subsection{Proofs of Preliminary Lemmas}\label{AppB.3}
	
	\begin{proof}[Proof of Lemma \ref{LemmaA.1}]
		\item
		Let
		\begin{eqnarray*}
			\mathbf{M}_t(\bm{\vartheta}, \bm{\theta}(\tau))&:=& (\mathbf{H}(\widetilde{\mathbf{z}}_{t-1}(\tau);\bm{\vartheta})\mathbf{H}(\widetilde{\mathbf{z}}_{t-1}(\tau);\bm{\vartheta})^\top)^{-1/2}\mathbf{H}(\widetilde{\mathbf{z}}_{t-1}(\tau);\bm{\theta}(\tau))\mathbf{H}(\widetilde{\mathbf{z}}_{t-1}(\tau);\bm{\theta}(\tau))^\top\\
			&& \times(\mathbf{H}(\widetilde{\mathbf{z}}_{t-1}(\tau);\bm{\vartheta})\mathbf{H}(\widetilde{\mathbf{z}}_{t-1}(\tau);\bm{\vartheta})^\top)^{-1/2}.
		\end{eqnarray*}
		By Assumption 1.1 and the construction of $\widetilde{\mathbf{x}}_t(\tau)$, we write
		\begin{eqnarray*}
			&&	E\left(\mathcal{l}(\widetilde{\mathbf{x}}_t(\tau),\widetilde{\mathbf{z}}_{t-1}(\tau);\bm{\vartheta})\right) \\
			&=&-\frac{1}{2}E\log\det\left\{\mathbf{H}(\widetilde{\mathbf{z}}_{t-1}(\tau);\bm{\vartheta})\mathbf{H}(\widetilde{\mathbf{z}}_{t-1}(\tau);\bm{\vartheta})^\top\right\}\\
			&&-\frac{1}{2}E\mathrm{tr}\left\{(\mathbf{H}(\widetilde{\mathbf{z}}_{t-1}(\tau);\bm{\vartheta})\mathbf{H}(\widetilde{\mathbf{z}}_{t-1}(\tau);\bm{\vartheta})^\top)^{-1}[\widetilde{\mathbf{x}}_t(\tau)-\bm{\mu}(\widetilde{\mathbf{z}}_{t-1}(\tau);\bm{\vartheta})][\widetilde{\mathbf{x}}_t(\tau)-\bm{\mu}(\widetilde{\mathbf{z}}_{t-1}(\tau);\bm{\vartheta})]^\top\right\}\\
			&=&-\frac{1}{2}E\log\det\left\{\mathbf{H}(\widetilde{\mathbf{z}}_{t-1}(\tau);\bm{\vartheta})\mathbf{H}(\widetilde{\mathbf{z}}_{t-1}(\tau);\bm{\vartheta})^\top\right\}-\frac{1}{2}E\mathrm{tr}\left\{\mathbf{M}_t(\bm{\vartheta}, \bm{\theta}(\tau))\right\}\\
			&& - \frac{1}{2}E\left([\bm{\mu}(\widetilde{\mathbf{z}}_{t-1}(\tau);\bm{\theta}(\tau))-\bm{\mu}(\widetilde{\mathbf{z}}_{t-1}(\tau);\bm{\vartheta})]^\top(\mathbf{H}(\widetilde{\mathbf{z}}_{t-1}(\tau);\bm{\vartheta})\mathbf{H}(\widetilde{\mathbf{z}}_{t-1}(\tau);\bm{\vartheta})^\top)^{-1}\right.\\
			&&\left.\times[\bm{\mu}(\widetilde{\mathbf{z}}_{t-1}(\tau);\bm{\theta}(\tau))-\bm{\mu}(\widetilde{\mathbf{z}}_{t-1}(\tau);\bm{\vartheta})]\right)\\
			&=&-\frac{1}{2}\left[-E\log\det\left(\mathbf{M}_t(\bm{\vartheta}, \bm{\theta}(\tau))\right) + E\mathrm{tr}\left\{\mathbf{M}_t(\bm{\vartheta}, \bm{\theta}(\tau))\right\}\right] \\
			&&-  \frac{1}{2}E\log\det\left(\widetilde{\mathbf{H}}_t(\tau,\bm{\theta}(\tau))\widetilde{\mathbf{H}}_t(\tau,\bm{\theta}(\tau))^\top\right)\\
			&& - \frac{1}{2}E\left([\bm{\mu}(\widetilde{\mathbf{z}}_{t-1}(\tau);\bm{\theta}(\tau))-\bm{\mu}(\widetilde{\mathbf{z}}_{t-1}(\tau);\bm{\vartheta})]^\top(\mathbf{H}(\widetilde{\mathbf{z}}_{t-1}(\tau);\bm{\vartheta})\mathbf{H}(\widetilde{\mathbf{z}}_{t-1}(\tau);\bm{\vartheta})^\top)^{-1}\right.\\
			&&\left.\times[\bm{\mu}(\widetilde{\mathbf{z}}_{t-1}(\tau);\bm{\theta}(\tau))-\bm{\mu}(\widetilde{\mathbf{z}}_{t-1}(\tau);\bm{\vartheta})]\right).
		\end{eqnarray*}
		
		For any positive definite matrix $\mathbf{M}$ with eigenvalues $\lambda_1,\ldots,\lambda_m > 0$, we have
		\begin{eqnarray*}
			f(\mathbf{M}) := -\log\det\left( \mathbf{M}  \right) + \mathrm{tr}\left\{\mathbf{M}\right\} = \sum_{i=1}^{m}(\lambda_i-\log \lambda_i) \geq m,
		\end{eqnarray*}
		where the equality holds if $\lambda_1=\cdots=\lambda_m=1$ in which case $\mathbf{M}=\mathbf{I}_m$. Thus, $f(\mathbf{M})$ is uniquely minimized at $\mathbf{M} = \mathbf{I}_m$, which implies that $E[f(\mathbf{M}_t(\bm{\vartheta}, \bm{\theta}(\tau)))]$ is uniquely minimized at $\bm{\vartheta} = \bm{\theta}(\tau)$ by Assumption 3.2. In addition, since $\mathbf{H}(\widetilde{\mathbf{z}}_{t-1}(\tau);\bm{\vartheta})\mathbf{H}(\widetilde{\mathbf{z}}_{t-1}(\tau);\bm{\vartheta})^\top$ is a positive definite matrix, then
		\begin{eqnarray*}
			&&E\left([\bm{\mu}(\widetilde{\mathbf{z}}_{t-1}(\tau);\bm{\theta}(\tau))-\bm{\mu}(\widetilde{\mathbf{z}}_{t-1}(\tau);\bm{\vartheta})]^\top(\mathbf{H}(\widetilde{\mathbf{z}}_{t-1}(\tau);\bm{\vartheta})\mathbf{H}(\widetilde{\mathbf{z}}_{t-1}(\tau);\bm{\vartheta})^\top)^{-1}\right.\\
			&&\left.\times[\bm{\mu}(\widetilde{\mathbf{z}}_{t-1}(\tau);\bm{\theta}(\tau))-\bm{\mu}(\widetilde{\mathbf{z}}_{t-1}(\tau);\bm{\vartheta})]\right)\geq 0
		\end{eqnarray*}
		is uniquely minimized at $\bm{\vartheta} = \bm{\theta}(\tau)$ by Assumption 3.2. Hence, $E\left(\mathcal{l}(\widetilde{\mathbf{x}}_t(\tau),\widetilde{\mathbf{z}}_{t-1}(\tau);\bm{\vartheta})\right)$ is uniquely maximized at $\bm{\theta}(\tau)$.
	\end{proof}
	
	\medskip
	
	\begin{proof}[Proof of Lemma \ref{LemmaA.2}]
		\item
		We first consider $\mathcal{l}(\cdot)$. Write
		\begin{eqnarray*}
			&&\mathcal{l}(\mathbf{x},\mathbf{z};\bm{\vartheta})-\mathcal{l}(\mathbf{x}^\prime,\mathbf{z}^\prime;\bm{\vartheta})\\
			& =&-\frac{1}{2}\left[ (\mathbf{x}- \bm{\mu}(\mathbf{z};\bm{\vartheta}))^\top\mathbf{M}^{-1}(\mathbf{z};\bm{\vartheta})(\mathbf{x} - \bm{\mu}(\mathbf{z};\bm{\vartheta}))-(\mathbf{x}^\prime- \bm{\mu}(\mathbf{z}^\prime;\bm{\vartheta}))^\top\mathbf{M}^{-1}(\mathbf{z}^\prime;\bm{\vartheta})(\mathbf{x}^\prime - \bm{\mu}(\mathbf{z}^\prime;\bm{\vartheta}))\right]\\
			&&-\frac{1}{2}\left[\log\det\left(\mathbf{M}(\mathbf{z};\bm{\vartheta})\right)-\log\det\left(\mathbf{M}(\mathbf{z}^\prime;\bm{\vartheta})\right)\right]\\
			&:=& -\frac{1}{2}(I_{1}+I_{2}),
		\end{eqnarray*}
		where the definitions of $I_1$ and $I_2$ should be obvious, and $\mathbf{M}(\mathbf{z};\bm{\vartheta})=\mathbf{H}(\mathbf{z};\bm{\vartheta})\mathbf{H}(\mathbf{z};\bm{\vartheta})^\top$.
		
		\medskip
		
		For $I_2$,  we have
		\begin{eqnarray*}
			|\mathbf{M}(\mathbf{z};\bm{\vartheta})-\mathbf{M}(\mathbf{z}^\prime;\bm{\vartheta})|&\leq&|\mathbf{H}(\mathbf{z};\bm{\vartheta})-\mathbf{H}(\mathbf{z}^\prime;\bm{\vartheta})|\left(|\mathbf{H}(\mathbf{z};\bm{\vartheta})|+|\mathbf{H}(\mathbf{z}^\prime;\bm{\vartheta})|\right)\\
			&\leq&M |\mathbf{z}-\mathbf{z}^\prime|_{\bm{\chi}}(2+|\mathbf{z}|_{\bm{\chi}}+|\mathbf{z}^\prime|_{\bm{\chi}}),
		\end{eqnarray*}
		where the second inequality follows from the facts that 
		\begin{eqnarray*}
			&&|\mathbf{H}(\mathbf{z};\bm{\vartheta})-\mathbf{H}(\mathbf{z}^\prime;\bm{\vartheta})|=O( |\mathbf{z}-\mathbf{z}^\prime|_{\bm{\chi}}),\\
			&&|\mathbf{H}(\mathbf{z};\bm{\vartheta})|\leq |\mathbf{H}(\mathbf{z};\bm{\vartheta})-\mathbf{H}(\mathbf{0};\bm{\vartheta})|+|\mathbf{H}(\mathbf{0};\bm{\vartheta})|=O(1+|\mathbf{z}|_{\bm{\chi}}) 
		\end{eqnarray*}
		by using Assumption 1.2 twice.
		
		By Assumption 3, it is easy to know that $\det(\mathbf{M}(\mathbf{z} ;\bm{\vartheta}))\geq\underline{H}>0$, which in connection with the fact $\log(\cdot)$ is Lipschitz continuous on $[\underline{H},\infty)$ yields that
		\begin{eqnarray*}
			I_2\leq M |\det\left(\mathbf{M}(\mathbf{z};\bm{\vartheta})\right)-\det\left(\mathbf{M}(\mathbf{z}^\prime;\bm{\vartheta})\right)|.
		\end{eqnarray*}
		In addition, for an invertible matrix $\mathbf{A}$, $\det(\mathbf{A}+\mathbf{B}) = \det(\mathbf{A})+\mathrm{tr}(\mathbf{A}^{-1,\top}\mathbf{B})+o(|\mathbf{B}|)$, and for a positive definite matrix $\mathbf{A}$ and symmetric matrix $\mathbf{B}$, $|\mathrm{tr}(\mathbf{A}^{-1,\top}\mathbf{B})|\leq |\mathbf{B}|\mathrm{tr}(\mathbf{A}^{-1})$. Hence, we have
		\begin{eqnarray*}
			I_2 &\leq& M |\det\left(\mathbf{M}(\mathbf{z};\bm{\vartheta})\right)-\det\left(\mathbf{M}(\mathbf{z}^\prime;\bm{\vartheta})\right)|\\
			&\leq &M \mathrm{tr}(\mathbf{M}^{-1}(\mathbf{z}^\prime;\bm{\vartheta}))\cdot |\mathbf{M}(\mathbf{z};\bm{\vartheta})-\mathbf{M}(\mathbf{z}^\prime;\bm{\vartheta})|\\
			&=&O(|\mathbf{z}-\mathbf{z}^\prime|_{\bm{\chi}}(1+|\mathbf{z}|_{\bm{\chi}}+|\mathbf{z}^\prime|_{\bm{\chi}})).
		\end{eqnarray*}
		Note that if $\mathbf{H}(\mathbf{z};\bm{\vartheta}) \equiv \mathbf{H}(\mathbf{0};\bm{\vartheta})$, $I_2 = 0$.
		
		\medskip
		
		For $I_1$, since $|\mathbf{H}^{-1}(\mathbf{z};\bm{\vartheta})|$ is bounded by Assumption 3, we can obtain that
		\begin{eqnarray*}
			I_1&\leq& |\mathbf{H}^{-1}(\mathbf{z};\bm{\vartheta})(\mathbf{x}-\bm{\mu}(\mathbf{z};\bm{\vartheta}))-\mathbf{H}^{-1}(\mathbf{z}^\prime;\bm{\vartheta})(\mathbf{x}^\prime-\bm{\mu}(\mathbf{z}^\prime;\bm{\vartheta}))|\\ 
			&&\cdot (|\mathbf{H}^{-1}(\mathbf{z};\bm{\vartheta})(\mathbf{x}-\bm{\mu}(\mathbf{z};\bm{\vartheta}))| + |\mathbf{H}^{-1}(\mathbf{z}^\prime;\bm{\vartheta})(\mathbf{x}^\prime-\bm{\mu}(\mathbf{z}^\prime;\bm{\vartheta}))|)\\
			&=&O(|\mathbf{y}-\mathbf{y}^\prime|_{\bm{\chi}}\cdot(1+|\mathbf{y}|_{\bm{\chi}}^2+|\mathbf{y}^\prime|_{\bm{\chi}}^2 )),
		\end{eqnarray*}
		where $\mathbf{y} = (\mathbf{x}, \mathbf{z})$. Similarly, if $\mathbf{H}(\mathbf{z};\bm{\vartheta}) \equiv \mathbf{H}(\mathbf{0};\bm{\vartheta})$, $I_1 =O(|\mathbf{y}-\mathbf{y}^\prime|_{\bm{\chi}}\cdot (1+|\mathbf{y}|_{\bm{\chi}}+|\mathbf{y}^\prime|_{\bm{\chi}}))$.
		
		For $\mathcal{l}(\mathbf{x},\mathbf{z};\bm{\vartheta})-\mathcal{l}(\mathbf{x},\mathbf{z};\bm{\vartheta}^\prime)$, write
		\begin{eqnarray*}
			&&\mathcal{l}(\mathbf{x},\mathbf{z};\bm{\vartheta})-\mathcal{l}(\mathbf{x},\mathbf{z};\bm{\vartheta}^\prime)\\
			& =&-\frac{1}{2}\left[ (\mathbf{x}- \bm{\mu}(\mathbf{z};\bm{\vartheta}))^\top\mathbf{M}^{-1}(\mathbf{z};\bm{\vartheta})(\mathbf{x} - \bm{\mu}(\mathbf{z};\bm{\vartheta}))-(\mathbf{x}- \bm{\mu}(\mathbf{z};\bm{\vartheta}^\prime))^\top\mathbf{M}^{-1}(\mathbf{z};\bm{\vartheta}^\prime)(\mathbf{x}- \bm{\mu}(\mathbf{z};\bm{\vartheta}^\prime))\right]\\
			&&-\frac{1}{2}\left[\log\det\left(\mathbf{M}(\mathbf{z};\bm{\vartheta})\right)-\log\det\left(\mathbf{M}(\mathbf{z};\bm{\vartheta}^\prime)\right)\right]\\
			&:=& -\frac{1}{2}(I_{3}+I_{4}).
		\end{eqnarray*}
		
		Similar to the development for $I_1$ and $I_2$, we can obtain that
		\begin{eqnarray*}
			I_3=O(|\bm{\vartheta}-\bm{\vartheta}^\prime|(1+|\mathbf{y}|_{\bm{\chi}}^3) \quad \text{and}\quad
			I_4=O(|\bm{\vartheta}-\bm{\vartheta}^\prime|(1+|\mathbf{z}|_{\bm{\chi}}^2)),
		\end{eqnarray*}
		where we again let $\mathbf{y} = (\mathbf{x}, \mathbf{z})$. Also if $\mathbf{H}(\mathbf{z};\bm{\vartheta}) \equiv \mathbf{H}(\mathbf{0};\bm{\vartheta})$,
		\begin{eqnarray*}
			I_3 = O(|\bm{\vartheta}-\bm{\vartheta}^\prime|(1+|\mathbf{y}|_{\bm{\chi}}^2)\quad \text{and}\quad I_4 = O(|\bm{\vartheta}-\bm{\vartheta}^\prime|).
		\end{eqnarray*}

		Combing the above analysis, we have shown $\mathcal{l} \in \mathcal{H}(3,\bm{\chi},M)$. In addition, if $\mathbf{H}(\mathbf{z};\bm{\vartheta}) \equiv \mathbf{H}(\mathbf{0};\bm{\vartheta})$, $\mathcal{l}\in \mathcal{H}(2,\bm{\chi},M)$.
		
		\medskip
		
		Similar to the development for $\mathcal{l}$, we can show $\gradient \mathcal{l}$, $\gradient ^2 \mathcal{l} \in \mathcal{H}(3,\bm{\chi},M)$ and $\gradient\mathcal{l}, \gradient^2\mathcal{l}\in \mathcal{H}(2,\bm{\chi},M)$ if $\mathbf{H}(\mathbf{z};\bm{\vartheta}) \equiv \mathbf{H}(\mathbf{0};\bm{\vartheta})$.

		The proof is now complete. 
	\end{proof}
	
	\medskip
	
	\begin{proof}[Proof of Lemma \ref{LemmaA.3}]
		\item
		
		\noindent (1). By Proposition 2.1.1, we have $\sup_{\tau\in[0,1]} \|\widetilde{\mathbf{x}}_t(\tau) \|_C < \infty$. Since $g \in \mathcal{H}(C,\bm{\chi},M)$, we have
		\begin{eqnarray*}
			\sup_{\bm{\eta}\neq\bm{\eta}^\prime}\frac{|\widetilde{G}_\tau(\bm{\eta})-\widetilde{G}_\tau(\bm{\eta}^\prime)|}{|\bm{\eta}-\bm{\eta}^\prime|}\leq M (Th)^{-1}\sum_{t=1}^{T}\widehat{K}\left(\frac{\tau_t-\tau}{h}\right)\left[2 + |\widetilde{\mathbf{y}}_t(\tau_t)|_{\bm{\chi}}^C+ \||\widetilde{\mathbf{y}}_t(\tau_t)|_{\bm{\chi}}^C \|_1\right].
		\end{eqnarray*}
		Using $(Th)^{-1}\sum_{t=1}^{T}\widehat{K}\left(\frac{\tau_t-\tau}{h}\right)<\infty$, we have
		\begin{eqnarray*}
			\left\|\sup_{\bm{\eta}\neq\bm{\eta}^\prime}\frac{|\widetilde{G}_\tau(\bm{\eta})-\widetilde{G}_\tau(\bm{\eta}^\prime)|}{|\bm{\eta}-\bm{\eta}^\prime|}\right\|_1\leq M \max_t\||\widetilde{\mathbf{y}}_t(\tau_t)|_{\bm{\chi}}^C \|_1<\infty.
		\end{eqnarray*}
		
		In addition, by using the Lipschitz property of $\widehat{K}(\cdot)$, we have
		\begin{eqnarray*}
			&&|\widetilde{G}_\tau(\bm{\eta})-\widetilde{G}_{\tau^\prime}(\bm{\eta}^\prime)|\\
			&\leq&(Th)^{-1}\sum_{t=1}^{T}\left|\widehat{K}\left(\frac{\tau_t-\tau}{h}\right)-\widehat{K}\left(\frac{\tau_t-\tau^\prime}{h}\right)\right|\cdot\sup_{\bm{\vartheta}}(|g(\widetilde{\mathbf{y}}_t(\tau_t),\bm{\vartheta})|+\|g(\widetilde{\mathbf{y}}_t(\tau_t),\bm{\vartheta})\|_1)\\
			&&+(Th)^{-1}\sum_{t=1}^{T}\widehat{K}\left(\frac{\tau_t-\tau^\prime}{h}\right)\cdot|g(\widetilde{\mathbf{y}}_t(\tau_t),\bm{\eta}_1+\bm{\eta}_2(\tau_t-\tau)/h)-g(\widetilde{\mathbf{y}}_t(\tau_t),\bm{\eta}_1^\prime+\bm{\eta}_2^\prime(\tau_t-\tau^\prime)/h)|\\
			&\leq&M\left(h^{-2}|\tau-\tau^\prime|+h^{-1}|\bm{\eta}-\bm{\eta}^\prime             |+h^{-2}|\bm{\eta}_2|\cdot|\tau-\tau^\prime| \right)\cdot\frac{1}{T}\sum_{t=1}^T(2 + |\widetilde{\mathbf{y}}_t(\tau_t)|_{\bm{\chi}}^C+ \||\widetilde{\mathbf{y}}_t(\tau_t)|_{\bm{\chi}}^C \|_1).
		\end{eqnarray*}
		Combing the above analyses, the first result follows.
		
		\medskip
		
		\noindent (2). By Lemma \ref{LemmaB.1}.1 and Proposition 2.2, for $|\tau_t-\tau|\leq h$, we have
		\begin{eqnarray*}
			&&\|g(\widetilde{\mathbf{y}}_t(\tau_t),\bm{\eta}_1+\bm{\eta}_2\cdot(\tau_t-\tau)/h) - g(\widetilde{\mathbf{y}}_t(\tau),\bm{\eta}_1+\bm{\eta}_2\cdot(\tau_t-\tau)/h)\|_1\\
			&\leq& M \sum_{j=0}^{\infty}\chi_j\|\widetilde{\mathbf{x}}_{t-j}(\tau_t)-\widetilde{\mathbf{x}}_{t-j}(\tau)\|_{C} = O(h).
		\end{eqnarray*}
		Hence, we have
		\begin{eqnarray*}
			&&\left\|\widetilde{B}_\tau(\bm{\eta}) - \frac{1}{Th}\sum_{t=1}^{T}\widehat{K}\left(\frac{\tau_t-\tau}{h} \right)g(\widetilde{\mathbf{y}}_t(\tau),\bm{\eta}_1+\bm{\eta}_2\cdot(\tau_t-\tau)/h) \right\|_1\\
			&\leq& M \frac{1}{Th}\sum_{t=1}^{T}\widehat{K}\left(\frac{\tau_t-\tau}{h} \right) \sum_{j=0}^{\infty}\chi_j \|\widetilde{\mathbf{x}}_{t-j}(\tau_t)-\widetilde{\mathbf{x}}_{t-j}(\tau)\|_{C} = O(h)
		\end{eqnarray*}
		and
		\begin{eqnarray*}
			&&\frac{1}{Th}\sum_{t=1}^{T}\widehat{K}\left(\frac{\tau_t-\tau}{h} \right)E\left[g(\widetilde{\mathbf{y}}_t(\tau),\bm{\eta}_1+\bm{\eta}_2\cdot(\tau_t-\tau)/h)\right] \\
			&=& \int_{-\tau/h}^{(1-\tau)/h}\widehat{K}(u)E(g(\widetilde{\mathbf{y}}_0(\tau),\bm{\eta}_1+\bm{\eta}_2u))\mathrm{d}u + O((Th)^{-1})
		\end{eqnarray*}
		by the definition of Riemann integral and the stationarity of $\widetilde{\mathbf{y}}_t(\tau)$.
		
		\medskip
		
		\noindent (3). Write
		\begin{eqnarray*}
			&&|\mathbf{m}_t^{(2)}(u, \tau) - \mathbf{m}_t^{(2)}(u, \tau^\prime)|\\
			&\leq& |\widehat{K}\left(\frac{\tau-u}{h}\right)-\widehat{K}\left(\frac{\tau^\prime-u}{h}\right)|\cdot|g(\widetilde{\mathbf{y}}_t(u),\bm{\theta}(\tau)-v\mathbf{d}(u,\tau))|\cdot |\mathbf{d}(u,\tau)|\\
			&& + |\widehat{K}\left(\frac{\tau^\prime-u}{h}\right)|\cdot |g(\widetilde{\mathbf{y}}_t(u),\bm{\theta}(\tau)-v\mathbf{d}(u,\tau))-g(\widetilde{\mathbf{y}}_t(u),\bm{\theta}(\tau^\prime)-v\mathbf{d}(u,\tau^\prime))|\cdot |\mathbf{d}(u,\tau)|\\
			&& + |\widehat{K}\left(\frac{\tau^\prime-u}{h}\right)|\cdot |g(\widetilde{\mathbf{y}}_t(u),\bm{\theta}(\tau^\prime)-v\mathbf{d}(u,\tau^\prime))|\cdot |\mathbf{d}(u,\tau)-\mathbf{d}(u,\tau^\prime)|\\
			&:=&I_1+I_2+I_3.
		\end{eqnarray*}
		
		By the Lipschitz continuity of $\widehat{K}(\cdot)$ and $\|\sup_{\bm{\vartheta}}|g(\widetilde{\mathbf{y}}_t(u),\bm{\vartheta})|\|_1=O(1)$ (by Lemma \ref{LemmaB.1}.3), we have 
		\begin{eqnarray*}
			E(I_1) = O(h^{-1}|\tau-\tau^\prime|).
		\end{eqnarray*}
		Similarly, by Lemma \ref{LemmaB.1}.2 and $|\mathbf{d}(u,\tau)| = O(1)$, we have 
		\begin{eqnarray*}
			E(I_2) = O(|\tau-\tau^\prime|).
		\end{eqnarray*}
		By the Lipschitz continuity of $\mathbf{d}(u, \cdot)$, we have $E(I_3) = O(|\tau-\tau^\prime|)$. Hence,
		\begin{eqnarray*}
			\left|\frac{1}{Th}\sum_{t=1}^{T}[\mathbf{m}_t^{(2)}(\tau,\tau_t) - \mathbf{m}_t^{(2)}(\tau^\prime,\tau_t)]\right|\leq \frac{1}{Th}\sum_{t=1}^{T} |\mathbf{m}_t^{(2)}(\tau,\tau_t) - \mathbf{m}_t^{(2)}(\tau^\prime,\tau_t)| =O(h^{-2} |\tau-\tau^\prime|).
		\end{eqnarray*}
		The proof is now complete.
		
		\medskip
		
		\noindent (4). By Propositions 2.1.1 and 2.2.2, we have $\sup_{\tau\in[0,1]}\left\|\widetilde{\mathbf{x}}_t(\tau)\right\|_C < \infty$ and $\max_t\left\|\mathbf{x}_t-\widetilde{\mathbf{x}}_t(\tau_t)\right\|_C = O(T^{-1})$. Hence, we have $\max_t\left\|\mathbf{x}_t\right\|_C \leq M$. 
		
		By Lemma \ref{LemmaB.1} and the definitions of $\mathbf{y}_t$ and $\mathbf{y}_t^c$, we have
		\begin{eqnarray*}
			\|\sup_{\bm{\vartheta} \in \bm{\Theta}_r}|g(\mathbf{y}_t,\bm{\vartheta})-g(\mathbf{y}_t^c,\bm{\vartheta})| \|_1\leq M \sum_{j=t}^{\infty}\chi_j\|\mathbf{x}_{t-j}\|_C=O\left(\sum_{j=t}^{\infty}\chi_j\right).
		\end{eqnarray*}
		
		In addition, by Proposition 2.2.1 and $\chi_j=O(j^{-(2+s)})$ for some $s > 0$, we have
		\begin{eqnarray*}
			&&\|\sup_{\bm{\vartheta} \in \bm{\Theta}_r}|g(\mathbf{y}_t,\bm{\vartheta})-g(\widetilde{\mathbf{y}}_t(\tau_t),\bm{\vartheta})| \|_1\leq M \sum_{j=0}^{\infty}\chi_j\|\mathbf{x}_{t-j}-\widetilde{\mathbf{x}}_{t-j}(\tau_t)\|_C\\
			&\leq&M \sum_{j=0}^{\infty}\chi_j\|\mathbf{x}_{t-j}-\widetilde{\mathbf{x}}_{t-j}(\tau_{t-j})\|_C+M \sum_{j=0}^{\infty}\chi_j\|\widetilde{\mathbf{x}}_{t-j}(\tau_t)-\widetilde{\mathbf{x}}_{t-j}(\tau_{t-j})\|_C\\
			&=&O(\sum_{j=0}^{\infty}\chi_j/T)+O(\sum_{j=0}^{\infty}j\chi_j/T)=O(T^{-1}).
		\end{eqnarray*}
		Hence, we have
		\begin{eqnarray*}
			&&\|\sup_{\tau \in [0,1]}\sup_{\bm{\eta} \in \mathbf{E}_T(r)}|\widetilde{G}_\tau(\bm{\eta})-G_\tau^c(\bm{\eta})|\|_1\\
			&\leq & M (Th)^{-1} \sum_{t=1}^{T}\sup_{\bm{\vartheta} \in \bm{\Theta}_r}\|g(\widetilde{\mathbf{y}}_t,\bm{\vartheta})-g(\mathbf{y}_t^c,\bm{\vartheta})\|_1\\
			&\leq&M (Th)^{-1} \sum_{t=1}^{T}\sum_{j=t}^{\infty}\chi_j \leq M (Th)^{-1} \sum_{j=1}^{\infty}j\chi_j = O((Th)^{-1}).
		\end{eqnarray*}
		The proof of the fourth result is now complete.
	\end{proof}

	\begin{proof}[Proof of Lemma \ref{LemmaA.4}]
		\item
		\noindent (1). Let $\widetilde{\mathbf{y}}_t^*(\tau)$ be a coupled version of $\widetilde{\mathbf{y}}_t(\tau)$ with $\bm{\varepsilon}_0$ replaced by $\bm{\varepsilon}_0^*$. By Lemma \ref{LemmaB.1}, we have
		\begin{eqnarray*}
			\delta_{q}^{\sup_{\bm{\vartheta}}|g(\widetilde{\mathbf{y}}_t(\tau),\bm{\vartheta})|}(t) &=& \|\sup_{\bm{\vartheta}}|g(\widetilde{\mathbf{y}}_t(\tau),\bm{\vartheta})|-\sup_{\bm{\vartheta}}|g(\widetilde{\mathbf{y}}_t^*(\tau),\bm{\vartheta})| \|_q\\
			&\leq &\|\sup_{\bm{\vartheta}}|g(\widetilde{\mathbf{y}}_t(\tau),\bm{\vartheta})-g(\widetilde{\mathbf{y}}_t^*(\tau),\bm{\vartheta})| \|_q\\
			&\leq&M \sum_{j=0}^{\infty}\chi_j \|\widetilde{\mathbf{x}}_{t-j}(\tau)-\widetilde{\mathbf{x}}_{t-j}^*(\tau)\|_{qC}\\
			&=&M \sum_{j=0}^{t}\chi_j\delta_{r}^{\widetilde{\mathbf{x}}(\tau)}(t-j).
		\end{eqnarray*}
		By Proposition 2.1.2 and the conditions on $  \alpha_j(\bm{\theta}(\tau))  $ and $  \beta_j(\bm{\theta}(\tau)) $ in the body of this lemma, we have $\delta_{r}^{\widetilde{\mathbf{x}}(\tau)}(j)=O(j^{-(a+s)})$ for some $s>0$. Hence, we have
		\begin{eqnarray*}
			\sum_{j=0}^{t}\chi_j\delta_{r}^{\widetilde{\mathbf{x}}(\tau)}(t-j)&\leq & \sum_{j\geq t/2}\chi_j\delta_{r}^{\widetilde{\mathbf{x}}(\tau)}(t-j)+\sum_{0\leq j\leq t/2}\chi_j\delta_{r}^{\widetilde{\mathbf{x}}(\tau)}(t-j)\\
			&\leq& (t/2)^{-(a+s)}  \sum_{j\geq t/2}\delta_{r}^{\widetilde{\mathbf{x}}(\tau)}(t-j) + (t/2)^{-(a+s)}  \sum_{0\leq j\leq t/2}\chi_j\\
			&=& O(t^{-(a+s)}).
		\end{eqnarray*}
		The proof of the first result of this lemma is now complete.
		
		\noindent (2)--(3). Since
		\begin{eqnarray*}
			|\sup_{u,\bm{\eta}}|m(\tau,\bm{\eta},u)|-\sup_{u,\bm{\eta}}|m^*(\tau,\bm{\eta},u)||&\leq&\sup_{u,\bm{\eta}}|m(\tau,\bm{\eta},u)-m^*(\tau,\bm{\eta},u)| 
			\\
			&\leq&M\sup_{\bm{\vartheta}}|g(\widetilde{\mathbf{y}}_t(\tau),\bm{\vartheta})-g(\widetilde{\mathbf{y}}_t^*(\tau),\bm{\vartheta})|,
		\end{eqnarray*}
		the second result follows directly from the first result.
		
		Since $\mathbf{d}(u, \tau) = O(h^2)$ when $|\tau-u|\leq h$, for each element of $\mathbf{m}_t^{(2)}(u, \tau)$, we have
		\begin{eqnarray*}
			|\sup_{\tau}|m_{t,i}^{(2)}(u,\tau)|-\sup_{\tau}|m_{t,i}^{(2)}(\tau,u)^*|| &\leq&\sup_{\tau}|m_{t,i}^{(2)}(\tau,u) -m_{t,i}^{(2)}(\tau,u)^*|\\
			&\leq&Mh^2\sup_{\bm{\vartheta}}|g(\widetilde{\mathbf{y}}_t(\tau),\bm{\vartheta})-g(\widetilde{\mathbf{y}}_t^*(\tau),\bm{\vartheta})|,
		\end{eqnarray*}
		where $ m_{t,i}^{(2)}(\tau,u)$ is yielded by the coupled version.
		
		The proof is now complete.
	\end{proof}

	\begin{proof}[Proof of Lemma \ref{LemmaA.5}]
		\item
		\noindent (1). Note that
		\begin{eqnarray*}
			&&g(\widetilde{\mathbf{y}}_t(\tau_t),\bm{\eta}_1+\bm{\eta}_2\cdot(\tau_t-\tau)/h)-E(g(\widetilde{\mathbf{y}}_t(\tau_t),\bm{\eta}_1+\bm{\eta}_2\cdot(\tau_t-\tau)/h))\\
			&=&\sum_{l=0}^{\infty}\mathcal{P}_{t-l}g(\widetilde{\mathbf{y}}_t(\tau_t),\bm{\eta}_1+\bm{\eta}_2\cdot(\tau_t-\tau)/h),
		\end{eqnarray*}
		in which $\{\mathcal{P}_{t-l}(g(\widetilde{\mathbf{y}}_t(\tau_t),\bm{\eta}_1+\bm{\eta}_2\cdot(\tau_t-\tau)/h))\}_{t=1}^T$ is a sequence of martingale differences.
		
		If $1<q\leq2$, by the Burkholder's inequality, $|\sum_{i=1}^{d}a_i|^r\leq \sum_{i=1}^{d}|a_i|^r$ for $r \in (0,1]$ and Lemma \ref{LemmaA.4}.1, we have
		\begin{eqnarray*}
			\|\widetilde{G}_\tau(\bm{\eta})\|_q  &\leq&\sum_{l=0}^{\infty}\left\|\sum_{t=1}^{T}\frac{1}{Th}\widehat{K}\left(\frac{\tau_t-\tau}{h}\right)\mathcal{P}_{t-l}g(\widetilde{\mathbf{y}}_t(\tau_t),\bm{\eta}_1+\bm{\eta}_2\cdot(\tau_t-\tau)/h)\right\|_q \\
			&\leq&O(1)\sum_{l=0}^{\infty} \left\{E\left[\sum_{t=1}^{T}\left(\frac{1}{Th}\widehat{K}\left(\frac{\tau_t-\tau}{h}\right)\mathcal{P}_{t-l}g(\widetilde{\mathbf{y}}_t(\tau_t),\bm{\eta}_1+\bm{\eta}_2\cdot(\tau_t-\tau)/h)\right)^2\right]^{q/2} \right\}^{1/q} \\
			&\leq&O(1)\sum_{l=0}^{\infty}\left\{E\left[\sum_{t=1}^{T}\left(\frac{1}{Th}\widehat{K}\left(\frac{\tau_t-\tau}{h}\right)\mathcal{P}_{t-l}g(\widetilde{\mathbf{y}}_t(\tau_t),\bm{\eta}_1+\bm{\eta}_2\cdot(\tau_t-\tau)/h)\right)^q\right] \right\}^{1/q} \\
			&\leq&O(1)(Th)^{-(q-1)/q}\sum_{l=0}^{\infty} \sup_{\tau \in [0,1]}\delta_{q}^{\sup_{\bm{\vartheta}}|g(\widetilde{\mathbf{y}}_t(\tau_t),\bm{\vartheta})|}(l) \left(\frac{1}{Th}\sum_{t=1}^{T} \widehat{K}\left(\frac{\tau_t-\tau}{h}\right)^q\right)^{1/q}\\
			&=&  O((Th)^{-(q-1)/q}).
		\end{eqnarray*}
		
		Similarly, for $q\geq 2$, by the Burkholder's inequality and the Minkowski inequality, we have
		\begin{eqnarray*}
			\|\widetilde{G}_\tau(\bm{\eta})\|_q &=&\left\|\frac{1}{Th}\sum_{t=1}^{T}\widehat{K}\left(\frac{\tau_t-\tau}{h}\right)\sum_{l=0}^{\infty}\mathcal{P}_{t-l}g(\widetilde{\mathbf{y}}_t(\tau_t),\bm{\eta}_1+\bm{\eta}_2\cdot(\tau_t-\tau)/h)\right\|_q \\
			&\leq&\sum_{l=0}^{\infty}\left\|\sum_{t=1}^{T}\frac{1}{Th}\widehat{K}\left(\frac{\tau_t-\tau}{h}\right)\mathcal{P}_{t-l}g(\widetilde{\mathbf{y}}_t(\tau_t),\bm{\eta}_1+\bm{\eta}_2\cdot(\tau_t-\tau)/h)\right\|_q \\
			&\leq&O(1)\sum_{l=0}^{\infty} \left\{E\left[\sum_{t=1}^{T}\left(\frac{1}{Th}\widehat{K}\left(\frac{\tau_t-\tau}{h}\right)\mathcal{P}_{t-l}g(\widetilde{\mathbf{y}}_t(\tau_t),\bm{\eta}_1+\bm{\eta}_2\cdot(\tau_t-\tau)/h)\right)^2\right]^{q/2} \right\}^{1/q} \\
			&\leq&O(1)\sum_{l=0}^{\infty}\left\{\sum_{t=1}^{T}\left[E\left(\frac{1}{Th}\widehat{K}\left(\frac{\tau_t-\tau}{h}\right)\mathcal{P}_{t-l}g(\widetilde{\mathbf{y}}_t(\tau_t),\bm{\eta}_1+\bm{\eta}_2\cdot(\tau_t-\tau)/h)\right)^q\right]^{2/q} \right\}^{1/2} \\
			&=&O(1)(Th)^{-1/2} \sum_{l=0}^{\infty} \sup_{\tau \in [0,1]}\delta_{q}^{\sup_{\bm{\vartheta}}|g(\widetilde{\mathbf{y}}_t(\tau),\bm{\vartheta})|}(l) = O((Th)^{-1/2}).
		\end{eqnarray*}
		The proof of the first result is now complete.
		
		\medskip
		
		\noindent (2). For any fixed $v > 0$, let $\kappa > 0$ and $\mathbf{E}_T^{\kappa}(r)$ be a discretization of $\mathbf{E}_T(r)$ such that for each $\bm{\eta} \in \mathbf{E}_T(r)$ one can find $\bm{\eta}^\prime \in \mathbf{E}_T^{\kappa}(r)$ satisfying $|\bm{\eta}-\bm{\eta}^\prime|\leq \kappa$. Let $\#\mathbf{E}_T^{\kappa}(r)$ denote the numbers of sets in $\mathbf{E}_T^{\kappa}(r)$.  Write
		\begin{eqnarray*}
			\Pr(\sup_{\bm{\eta}\in\mathbf{E}_T(r)}|\widetilde{G}_\tau(\bm{\eta})|>v)&\leq & \#\mathbf{E}_T^{\kappa}(r) \sup_{\bm{\eta}\in\mathbf{E}_T(r)} \Pr(|\widetilde{G}_\tau(\bm{\eta})|>v/2)\\
			&& + \Pr(\sup_{|\bm{\eta}-\bm{\eta}^\prime|\leq\kappa}|\widetilde{G}_\tau(\bm{\eta})-\widetilde{G}_\tau(\bm{\eta}^\prime)|>v/2).
		\end{eqnarray*}
		
		By the Markov inequality, we have
		\begin{eqnarray*}
			\Pr(|\widetilde{G}_\tau(\bm{\eta})|>v/2) \leq \frac{\|\widetilde{G}\tau(\bm{\eta})\|_q^q}{(v/2)^q}.
		\end{eqnarray*}
		
		Note that $\{\mathcal{P}_{t-j}g(\widetilde{\mathbf{y}}_t(\tau_t),\bm{\eta}_1+\bm{\eta}_2\cdot(\tau_t-\tau)/h)\}_t$ forms a sequence of martingale differences. By the Burkholder's inequality and Lemma \ref{LemmaA.4}.1, we have
		\begin{eqnarray*}
			\|\widetilde{G}_\tau(\bm{\eta})\|_q
			&\leq&(Th)^{-1}\sum_{j=0}^\infty \|\sum_{t=1}^{T}\widehat{K}\left(\frac{\tau_t-\tau}{h}\right)\mathcal{P}_{t-j}g(\widetilde{\mathbf{y}}_t(\tau_t),\bm{\eta}_1+\bm{\eta}_2\cdot(\tau_t-\tau)/h) \|_q\\
			&\leq&(q-1)^{-1}(Th)^{-1}\sum_{j=0}^\infty\left( \|\sum_{t=1}^{T}\widehat{K}\left(\frac{\tau_t-\tau}{h}\right)^2\mathcal{P}_{t-j}^2g(\widetilde{\mathbf{y}}_t(\tau_t),\bm{\eta}_1+\bm{\eta}_2\cdot(\tau_t-\tau)/h) \|_{q/2}^{q/2}\right)^{1/q}\\
			&\leq&M(Th)^{-(q-1)/q}\sum_{j=0}^\infty\sup_{\tau \in [0,1]}\delta_{q}^{\sup_{\bm{\vartheta}}|g(\widetilde{\mathbf{y}}_t(\tau),\bm{\vartheta})|}(j) = O((Th)^{-(q-1)/q}),
		\end{eqnarray*}
		which in connection with the fact $\#\mathbf{E}_T^{\kappa}(r) $ is independent of $T$ yields that
		\begin{eqnarray*}
			\#\mathbf{E}_T^{\kappa}(r) \sup_{\bm{\eta}\in\mathbf{E}_T(r)} \Pr(|\widetilde{G}_\tau(\bm{\eta})|>v/2) =o(1).
		\end{eqnarray*}
		
		In addition, by Lemma \ref{LemmaA.3}.1, we have
		\begin{eqnarray*}
			\Pr(\sup_{|\bm{\eta}-\bm{\eta}^\prime|\leq\kappa}|\widetilde{G}_\tau(\bm{\eta})-\widetilde{G}_\tau(\bm{\eta}^\prime)|>v/2)\leq M \kappa \to 0
		\end{eqnarray*}
		by choosing $\kappa$ small enough. Hence, $\Pr(\sup_{\bm{\eta}\in\mathbf{E}_T(r)}|\widetilde{G}_\tau(\bm{\eta})|>v)\to 0$ as $T \to \infty$.
		
		\medskip
		
		\noindent (3). Let $\beta_T := (\log T)^{1/2}(Th)^{-1/2}h^{-1/2}$ for short. Let further $\mathbf{E}_{T,\kappa}(r)$ be a discretization of $\mathbf{E}_{T}(r)$ such that for each $\bm{\eta} \in \mathbf{E}_T(r)$ one can find $\bm{\eta}^\prime \in \mathbf{E}_{T,\kappa}(r)$ satisfying $|\bm{\eta}-\bm{\eta}^\prime|\leq \kappa_T^{-1}$. Define $\mathcal{T}_{T,\kappa} = \{t/\kappa_T:t=1,2,...,\kappa_T\}$ as a discretization of $[0,1]$. For some constant $M > 0$, we have
		\begin{eqnarray*}
			&&\Pr(\sup_{\tau\in[0,1]}\sup_{\bm{\eta}\in\mathbf{E}_{T}(r)}|\widetilde{G}_\tau(\bm{\eta})|>M\beta_T)\\
			&\leq &  \Pr(\sup_{\tau\in\mathcal{T}_{T,\kappa}}\sup_{\bm{\eta}\in\mathbf{E}_{T,\kappa}(r)}|\widetilde{G}_\tau(\bm{\eta})|>\beta_TM/2)\\
			&& + \Pr(\sup_{|\tau-\tau^\prime|\leq\kappa_T^{-1},|\bm{\eta}-\bm{\eta}^\prime|\leq\kappa_T^{-1}}|\widetilde{G}_\tau(\bm{\eta})-\widetilde{G}_{\tau^\prime}(\bm{\eta}^\prime)|>\beta_TM/2).
		\end{eqnarray*}
		Let $m_t(\tau,\bm{\eta},u):=\widehat{K}\left(\frac{\tau-u}{h}\right)g(\widetilde{\mathbf{y}}_t(\tau),\bm{\eta}_1+\bm{\eta}_2\cdot(\tau-u)/h)$, we have $\sup_{\tau \in [0,1]}\sup_{u,\bm{\eta}}\delta_{q}^{m(\tau,\bm{\eta},u)}(j) = O(j^{-(a+s)})$ and $\sup_{\tau \in [0,1]}\delta_{q}^{\sup_{u,\bm{\eta}}|m(\tau,\bm{\eta},u)|}(j) = O(j^{-(a+s)})$ for some $a \geq3/2$ by Lemma \ref{LemmaA.4}.2. Let $\alpha = 1/2$, we have
		$$
		W_{q,\alpha} := \max_{k\geq 0}(k+1)^\alpha \sup_{\tau \in [0,1]}\sum_{j=k}^{\infty}\delta_{q}^{\sup_{u,\bm{\eta}}|m(\tau,\bm{\eta},u)|}(j) \leq M \max_k k^{-(a-3/2+s)} < \infty
		$$ 
		and
		$$
		W_{2,\alpha} := \max_{k\geq 0}(k+1)^\alpha \sup_{\tau \in [0,1]}\sup_{u,\bm{\eta}}\sum_{j=k}^{\infty}\delta_{2}^{m(\tau,\bm{\eta},u)}(j) \leq M \max_k k^{-(a-3/2+s)} < \infty.
		$$
		Note that $ l = \min\{1,\log(\#\mathbf{E}_{T,\kappa}(r) \times \mathcal{T}_{T,\kappa})\} \leq 3(2d+1)\log(T)$ and $M\beta_T Th = M T^{1/2} (\log T)^{1/2} \geq \sqrt{Tl} W_{2,\alpha}+T^{1/q}l^{3/2}W_{q,\alpha}\geq T^{1/2} (\log T)^{1/2}+T^{1/q}(\log T)^{3/2}$ for some $M$ large enough. By using Theorem 6.2 of \cite{zhang2017gaussian} (the proof therein also works for the uniform functional dependence measure) with $q > 2$ and $\alpha = 1/2$ to $\{m_t(\tau,\bm{\eta},\tau_t) \}_{\tau \in \mathcal{T}_{T,\kappa},\bm{\eta}\in \mathbf{E}_{T,\kappa}(r)}$, we have
		\begin{eqnarray*}
			&&\Pr(\sup_{\tau\in\mathcal{T}_{T,\kappa}}\sup_{\bm{\eta}\in\mathbf{E}_{T,\kappa}(r)}|\widetilde{G}_\tau(\bm{\eta})|>\beta_TM/2)\\
			&\leq&\frac{M Tl^{q/2}}{(\beta_T Th)^q} + M\exp\left(-\frac{M(\beta_T Th)^2}{T} \right)\\
			&\leq& M \left(T^{-(q-2)/2} + \exp(-\log T)\right) \to 0.
		\end{eqnarray*}
		In addition, by the Markov inequality and Lemma \ref{LemmaA.3}.1, we have
		\begin{eqnarray*}
			\Pr(\sup_{|\tau-\tau^\prime|\leq\kappa_T^{-1},|\bm{\eta}-\bm{\eta}^\prime|\leq\kappa_T^{-1}}|\widetilde{G}_\tau(\bm{\eta})-\widetilde{G}_{\tau^\prime}(\bm{\eta}^\prime)|>\beta_TM/2) = O(h^{-2}T^{-3}/\beta_T)\to 0.
		\end{eqnarray*}
		The proof is now complete.
	\end{proof}

	\begin{proof}[Proof of Lemma \ref{LemmaA.7}]
		\item
		
		For notational simplicity, we let $ \bm{\eta}(\tau) = [\bm{\theta}(\tau) ,h\bm{\theta}^{(1)}(\tau )]$ in what follows, and define
		\begin{eqnarray*}
			\bm{\Gamma}(\tau)&:=&\gradient \widetilde{\mathcal{L}}_\tau(\bm{\eta} (\tau))-E[\gradient \widetilde{\mathcal{L}}_\tau(\bm{\eta}(\tau))]\\
			&&-\frac{1}{Th}\sum_{t=1}^{T}\widehat{\bm{K}}((\tau_t-\tau)/h)\otimes \gradient_{\bm{\vartheta}}\mathcal{l}(\widetilde{\bm{x}}_t(\tau_t),\widetilde{\bm{z}}_{t-1}(\tau_t);\bm{\theta}(\tau_t)).
		\end{eqnarray*}
		Due to $E(\gradient_{\bm{\vartheta}}\mathcal{l}(\widetilde{\bm{x}}_t(\tau_t),\widetilde{\bm{z}}_{t-1}(\tau_t);\bm{\theta}(\tau_t)))=0$ by Lemma \ref{LemmaA.1}, we have
		\begin{eqnarray*}
			\bm{\Gamma}(\tau)
			&=&\frac{1}{Th}\sum_{t=1}^{T}\widehat{\bm{K}}((\tau_t-\tau)/h)\otimes [\gradient_{\bm{\vartheta}}\mathcal{l}(\widetilde{\bm{x}}_t(\tau_t),\widetilde{\bm{z}}_{t-1}(\tau_t);\bm{\theta}(\tau)+\bm{\theta}^{(1)}(\tau)(\tau_t-\tau))\\
			&&-\gradient_{\bm{\vartheta}}\mathcal{l}(\widetilde{\bm{x}}_t(\tau_t),\widetilde{\bm{z}}_{t-1}(\tau_t);\bm{\theta}(\tau_t))]\\
			&&-\frac{1}{Th}\sum_{t=1}^{T}\widehat{\bm{K}}((\tau_t-\tau)/h)\otimes E[\gradient_{\bm{\vartheta}}\mathcal{l}(\widetilde{\bm{x}}_t(\tau_t),\widetilde{\bm{z}}_{t-1}(\tau_t);\bm{\theta}(\tau)+\bm{\theta}^{(1)}(\tau)(\tau_t-\tau))\\
			&& -\gradient_{\bm{\vartheta}}\mathcal{l}(\widetilde{\bm{x}}_t(\tau_t),\widetilde{\bm{z}}_{t-1}(\tau_t);\bm{\theta}(\tau_t)].
		\end{eqnarray*}
		
		By the Mean Value Theorem, we have
		\begin{eqnarray*}
			\bm{\Gamma}(\tau) = (Th)^{-1}\sum_{t=1}^{T}[\mathbf{M}_t^{(2)}(\tau,\tau_t) - E(\mathbf{M}_t^{(2)}(\tau,\tau_t))],
		\end{eqnarray*}
		where
		\begin{eqnarray*}
			&&\mathbf{M}_t^{(2)}(\tau,u):=\widehat{\bm{K}}((\tau_t-\tau)/h)\otimes\gradient_{\bm{\vartheta}}^2\mathcal{l}(\widetilde{\bm{x}}_t(u),\widetilde{\bm{z}}_{t-1}(u);\bm{\theta}(\tau)-v\mathbf{r}(u))\mathbf{r}(u) \text{ for some $v\in[0,1]$}, \\
			&& \mathbf{r}(u) = \frac{1}{2}\bm{\theta}^{(2)}(\tau)(u-\tau)^2+\frac{1}{6}\bm{\theta}^{(3)}(\overline{\tau})(u-\tau)^3 \text{ with $\overline{\tau}$ between $u$ and $\tau$}.
		\end{eqnarray*}
		
		We then use a similar argument as in the proof of Lemma \ref{LemmaA.5} to prove
		\begin{eqnarray*}
			\Pr(\sup_{\tau\in[0,1]} |\bm{\Gamma}(\tau)| > M \beta_T h^2)\to 0.
		\end{eqnarray*}
		Define $\kappa_T = T^5$ and $\mathcal{T}_{T,\kappa} = \{t/\kappa_T:t=1,2,...,\kappa_T\}$ as a discretization of $[0,1]$. For some constant $M > 0$, we have
		\begin{eqnarray*}
			\Pr(\sup_{\tau\in[0,1]}|\bm{\Gamma}(\tau)|>M\beta_Th^2)
			&\leq&  \Pr(\sup_{\tau\in\mathcal{T}_{T,\kappa}}|\bm{\Gamma}(\tau)|>\beta_Th^2M/2) \\
			&&+ \Pr(\sup_{|\tau-\tau^\prime|\leq\kappa_T^{-1}}|\bm{\Gamma}(\tau)-\bm{\Gamma}(\tau^\prime)|>\beta_Th^2M/2).
		\end{eqnarray*}
		By Lemma \ref{LemmaA.3}.3 and the Markov inequality, we have
		\begin{eqnarray*}
			\Pr(\sup_{|\tau-\tau^\prime|\leq\kappa_T^{-1}}|\bm{\Gamma}(\tau)-\bm{\Gamma}(\tau^\prime)|>\beta_Th^2M/2) =  O\left( \frac{h^{-2}\kappa_T^{-1}}{\beta_Th^2M/2} \right)\to 0.
		\end{eqnarray*}
		
		By Lemma \ref{LemmaA.4}.3, we have 
		\begin{eqnarray*}
			\sup_{u,\tau \in [0,1]}\delta_{q}^{|\mathbf{M}^{(2)}(\tau,u)|}(j) = O(h^2j^{-(a+s)})\quad \text{and}\quad \sup_{\tau \in [0,1]}\delta_{q}^{\sup_{u}|\mathbf{M}^{(2)}(\tau,u)|}(j) = O(h^2j^{-(a+s)})
		\end{eqnarray*}
		for some $a \geq3/2$ . Let $\alpha = 1/2$, we have
		\begin{eqnarray*}
			\widetilde{W}_{q,\alpha} := \max_{k\geq 0}(k+1)^\alpha \sup_{\tau \in [0,1]}\sum_{j=k}^{\infty}\delta_{q}^{\sup_{u}|\mathbf{M}^{(2)}(\tau,u)|}(j) = O(h^2)
		\end{eqnarray*}
		and
		\begin{eqnarray*}
			\widetilde{W}_{2,\alpha} := \max_{k\geq 0}(k+1)^\alpha \sup_{\tau,u \in [0,1]}\sum_{j=k}^{\infty}\delta_{2}^{|\mathbf{M}^{(2)}(\tau,u)|}(j)=O(h^2).
		\end{eqnarray*}
		Using Theorem 6.2 of \cite{zhang2017gaussian} with $q > 2$, $\alpha = 1/2$ and $ l = \min\{1,\log(\# \mathcal{T}_{T,\kappa})\} \leq 5\log(T)$ to $\{\mathbf{M}_t^{(2)}(\tau,\tau_t) \}_{\tau \in \mathcal{T}_{T,\kappa}}$, we have
		\begin{eqnarray*}
			&&\Pr(\sup_{\tau\in\mathcal{T}_{T,\kappa}}|\bm{\Gamma}(\tau)|>h^2\beta_TM/2)\\
			&\leq&\frac{M Tl^{q/2}\widetilde{W}_{q,\alpha}^q}{(\beta_Th^2 Th)^q} + M\exp\left(-\frac{M(\beta_Th^2 Th)^2}{T\widetilde{W}_{2,\alpha}^2} \right)\\
			&\leq& M \left(T^{-(q-2)/2} + \exp(-\log T)\right) \to 0.
		\end{eqnarray*}
		The proof is now complete.
	\end{proof}

	\begin{proof}[Proof of Lemma \ref{LemmaA.8}]
		\item
		\noindent (1). Let $\widehat{\bm{\eta}}(\tau):= [\widehat{\bm{\theta}}(\tau)^\top, \widehat{\bm{\theta}}^\star (\tau)^\top]^\top$ and $\bm{\eta} (\tau):= [\bm{\theta}(\tau)^\top,h\bm{\theta}^{(1)}(\tau)^\top]^\top$. By Lemma \ref{LemmaA.5} and the proof of Theorem 2.1, we have  
		\begin{eqnarray*}
			\sup_{\tau \in [0,1]}|\widehat{\bm{\eta}}(\tau) - \bm{\eta}(\tau)| = o_P(1).
		\end{eqnarray*}
		By the Taylor expansion, we have
		\begin{eqnarray*}
			\widehat{\bm{\eta}}(\tau) - \bm{\eta}(\tau) = -(\widetilde{\bm{\Sigma}}(\tau)+\mathbf{R}_T(\tau))^{-1}\gradient  \mathcal{L}_\tau(\bm{\eta}(\tau)),
		\end{eqnarray*}
		where $\mathbf{R}_T(\tau) := \gradient^2 \mathcal{L}_\tau(\overline{\bm{\eta}})-\widetilde{\bm{\Sigma}}(\tau)$ and $\widetilde{\bm{\Sigma}}(\tau):=\left[\begin{matrix}
			1 & 0\\
			0 & \widetilde{c}_2
		\end{matrix} \right]\otimes\bm{\Sigma}(\tau)$
		with $\overline{\bm{\eta}}$ between $\widehat{\bm{\eta}}(\tau)$ and $\bm{\eta}(\tau)$. By Lemma \ref{LemmaA.3} and Lemma \ref{LemmaA.5}, we have
		\begin{eqnarray*}
			&&\sup_{\tau \in [0,1],\bm{\eta}\in\mathbf{E}_T(r)}|\gradient^2 \mathcal{L}_\tau(\bm{\eta}) - \widetilde{\bm{\Sigma}}(\tau,\bm{\eta})|\\
			&=& \sup_{\tau \in [0,1],\bm{\eta}\in\mathbf{E}_T(r)}|\gradient^2 \widetilde{\mathcal{L}}_\tau(\bm{\eta}) - \widetilde{\bm{\Sigma}}(\tau,\bm{\eta})| + O_P((Th)^{-1})\\
			&=& \sup_{\tau \in [0,1],\bm{\eta}\in\mathbf{E}_T(r)}|E[\gradient^2 \widetilde{\mathcal{L}}_\tau(\bm{\eta})] - \widetilde{\bm{\Sigma}}(\tau,\bm{\eta})| + O_P((Th)^{-1} + \beta_T)\\
			&=&O_P(\beta_T + (Th)^{-1})+O(h),
		\end{eqnarray*}
		where $\widetilde{\bm{\Sigma}}(\tau,\bm{\eta}):= \int_{-\tau/h}^{(1-\tau)/h} K(u)\left[\begin{matrix}
			1 & u\\
			u & u^2
		\end{matrix} \right]\otimes\bm{\Sigma}(\tau,\bm{\eta}_1+\bm{\eta}_2u)\mathrm{d}u$ and 
		$$
		\bm{\Sigma}(\tau,\bm{\eta}_1+\bm{\eta}_2u):=E\left(\gradient_{\bm{\vartheta}}^2\mathcal{l}(\widetilde{\mathbf{x}}_t(\tau),\widetilde{\mathbf{z}}_{t-1}(\tau);\bm{\eta}_1+\bm{\eta}_2u) \right).
		$$
		
		By Lemma \ref{LemmaA.4}.3 and the condition 
		\begin{eqnarray*}
			\sup_{\tau \in [0,1]} [\alpha_j(\bm{\theta}(\tau)) + \beta_j(\bm{\theta}(\tau)) ]= O(j^{-(3+s)})
		\end{eqnarray*}
		for some $s > 0$, we have $\sup_{\tau \in [0,1]}\delta_{q}^{\gradient_{\vartheta_j}\mathcal{l}}(j) = O(j^{-(2+s)})$ for some $s > 0$. By Lemma \ref{LemmaB.3}, we have
		$$
		\sup_{\tau\in[0,1]}\left|T^{-1}\sum_{t=1}^{T}K_h(\tau_t-\tau)\gradient_{\bm{\vartheta}}\mathcal{l}(\widetilde{\mathbf{x}}_t(\tau_t),\widetilde{\mathbf{z}}_{t-1}(\tau_t);\bm{\theta}(\tau_t)) \right| = O_P((Th)^{-1/2}\log T).
		$$
		Since $E\left(\gradient_{\bm{\vartheta}}\mathcal{l}(\widetilde{\mathbf{x}}_t(\tau),\widetilde{\mathbf{z}}_{t-1}(\tau);\bm{\theta}(\tau))\right) = 0$, we further obtain that
		\begin{eqnarray*}
			&&\sup_{\tau\in[0,1]}|\gradient \widetilde{\mathcal{L}}_\tau(\bm{\eta}(\tau))-E[\gradient \widetilde{\mathcal{L}}_\tau(\bm{\eta}(\tau))]|\\
			&\leq&\sup_{\tau\in[0,1]}|\gradient \widetilde{\mathcal{L}}_\tau( \bm{\eta}(\tau))-E[\gradient \widetilde{\mathcal{L}}_\tau(\bm{\eta}(\tau))]\\
			&&-\frac{1}{Th}\sum_{t=1}^{T}\widehat{\bm{K}}((\tau_t-\tau)/h)\otimes \gradient_{\bm{\vartheta}}\mathcal{l}(\widetilde{\bm{x}}_t(\tau_t),\widetilde{\bm{z}}_{t-1}(\tau_t);\bm{\theta}(\tau_t))|\\
			&&+\sup_{\tau\in[0,1]}\left|T^{-1}\sum_{t=1}^{T}K_h(\tau_t-\tau)\gradient_{\bm{\vartheta}}\mathcal{l}(\widetilde{\mathbf{x}}_t(\tau_t),\widetilde{\mathbf{z}}_{t-1}(\tau_t);\bm{\theta}(\tau_t)) \right|\\
			&=&O_P(h^2\beta_T + (Th)^{-1/2}\log T)
		\end{eqnarray*}
		using Lemma \ref{LemmaA.7}.
		
		Hence, by Lemma \ref{LemmaA.3}.4, we have
		\begin{eqnarray*}
			&&\sup_{\tau \in [0,1]}|\gradient  \mathcal{L}_\tau(\bm{\eta}(\tau))|\\
			&\leq& \sup_{\tau \in [0,1]}|\gradient \mathcal{L}_\tau(\bm{\eta}(\tau))-\gradient \widetilde{\mathcal{L}}_\tau(\bm{\eta}(\tau))|  + \sup_{\tau \in [0,1]}|\gradient  \widetilde{\mathcal{L}}_\tau(\bm{\eta}(\tau))-E(\gradient \widetilde{\mathcal{L}}_\tau(\bm{\eta}(\tau)))|\\
			&& + \sup_{\tau \in [0,1]}|E(\gradient  \widetilde{\mathcal{L}}_\tau(\bm{\eta}(\tau)))|\\
			&=& \sup_{\tau \in [0,1]}|E(\gradient \widetilde{\mathcal{L}}_\tau(\bm{\eta}(\tau)))| + O_P(h^2\beta_T + (Th)^{-1/2}\log T+(Th)^{-1}).
		\end{eqnarray*}
		Since
		\begin{eqnarray*}
			&&E\left[\gradient \widetilde{\mathcal{L}}_\tau(\bm{\eta} (\tau))-\frac{1}{Th}\sum_{t=1}^{T}\widehat{\bm{K}}((\tau_t-\tau)/h)\otimes \gradient_{\bm{\vartheta}}\mathcal{l}(\widetilde{\bm{x}}_t(\tau_t),\widetilde{\bm{z}}_{t-1}(\tau_t);\bm{\theta}(\tau_t))\right]\\
			&=&-\frac{1}{2}h^2\frac{1}{Th}\sum_{t=1}^{T}\widehat{\bm{K}}((\tau_t-\tau)/h)\otimes\left[E[\gradient_{\bm{\vartheta}}^2\mathcal{l}(\widetilde{\bm{x}}_t(\tau),\widetilde{\bm{z}}_{t-1}(\tau);\bm{\theta}(\tau))]\bm{\theta}^{(2)}(\tau)\left(\frac{\tau_t-\tau}{h}\right)^2\right]+O(h^3)\\
			&=&\frac{1}{2}h^2\int_{-\tau/h}^{(1-\tau)/h}K(u)[u^2,u^3]^\top\mathrm{d}u\otimes \left(-\bm{\Sigma}(\tau) \bm{\theta}^{(2)}(\tau)\right)+O((Th)^{-1}+h^3),
		\end{eqnarray*}
		we have
		$$
		\sup_{\tau \in [0,1]}|\gradient_{\bm{\eta}_j} \mathcal{L}_\tau(\bm{\eta}(\tau)) | = O_P(h^2\beta_T + (Th)^{-1/2}\log T+(Th)^{-1} + h^{1+j}).
		$$
		for $j =1,2$. Hence, we have $\sup_{\tau \in [0,1]}|\widehat{\bm{\eta}}_j(\tau) - \bm{\eta}_j(\tau)| = O_P(h^2\beta_T + (Th)^{-1/2}\log T+(Th)^{-1} + h^{1+j})$ and $\sup_{\tau \in [0,1]}|\mathbf{R}_T(\tau)|=O_P(\beta_T + h + (Th)^{-1})$, where $\widehat{\bm{\eta}}_j(\tau) $ and $\bm{\eta}_j(\tau)$  are corresponding to the $j^{th}$ part in their definitions given in the beginning of this proof.
		
		Write
		\begin{eqnarray*}
			&&\left|-\widetilde{\bm{\Sigma}}(\tau)(\widehat{\bm{\eta}} (\tau)-\bm{\eta} (\tau))-\gradient \mathcal{L}_\tau(\bm{\eta}(\tau)) \right|\\
			&\leq& | [\mathbf{I}_{2d} +\widetilde{\bm{\Sigma}}^{-1}(\tau)\mathbf{R}_T(\tau) ]^{-1} - \mathbf{I}_{2d}^{-1}|\cdot |\gradient  \mathcal{L}_\tau(\bm{\eta}(\tau))|\\
			&\leq& | [\mathbf{I}_{2d} +\widetilde{\bm{\Sigma}}^{-1}(\tau)\mathbf{R}_T(\tau) ]^{-1}|\cdot |\widetilde{\bm{\Sigma}}^{-1}(\tau)\mathbf{R}_T(\tau)|\cdot |\gradient  \mathcal{L}_\tau(\bm{\eta}(\tau))|\\
			&=&O_P(\gamma_T).
		\end{eqnarray*}
		The proof of the first result is now complete.
		
		\noindent (2). By Lemma \ref{LemmaA.3}  and Lemma \ref{LemmaA.7}, we have
		\begin{eqnarray*}	&&\sup_{\tau\in[h,1-h]}\Big|\gradient_{\bm{\vartheta}}\mathcal{L}_\tau(\bm{\theta}(\tau),h\bm{\theta}^{(1)}(\tau)) + \frac{1}{2}h^2\widetilde{c}_2\bm{\Sigma}(\tau)\bm{\theta}^{(2)}(\tau)\\
			&&\quad -\frac{1}{T}\sum_{t=1}^{T}\gradient_{\bm{\vartheta}}\mathcal{l}(\widetilde{\mathbf{x}}_t(\tau_t),\widetilde{\mathbf{z}}_{t-1}(\tau_t);\bm{\theta}(\tau_t))K_h(\tau_t-\tau) \Big| \\
			&\leq& \sup_{\tau \in [h,1-h]}|\gradient_{\bm{\eta}_1} \mathcal{L}_\tau(\bm{\eta}(\tau))-\gradient_{\bm{\eta}_1} \widetilde{\mathcal{L}}_\tau(\bm{\eta}(\tau))| \\
			&& + \sup_{\tau \in [h,1-h]}|\gradient_{\bm{\eta}_1} \widetilde{\mathcal{L}}_\tau(\bm{\eta}(\tau))-E(\gradient_{\bm{\eta}_1} \widetilde{\mathcal{L}}_\tau(\bm{\eta}(\tau)))-\frac{1}{T}\sum_{t=1}^{T}\gradient_{\bm{\vartheta}}\mathcal{l}(\widetilde{\mathbf{x}}_t(\tau_t),\widetilde{\mathbf{z}}_{t-1}(\tau_t);\bm{\theta}(\tau_t))K_h(\tau_t-\tau)|\\
			&& + \sup_{\tau \in [h,1-h]}|E(\gradient_{\bm{\eta}_1} \widetilde{\mathcal{L}}_\tau(\bm{\eta}(\tau)))+\frac{1}{2}h^2\widetilde{c}_2\bm{\Sigma}(\tau)\bm{\theta}^{(2)}(\tau)|\\
			&=& \sup_{\tau \in [h,1-h]}|E(\gradient_{\bm{\eta}_1} \widetilde{\mathcal{L}}_\tau(\bm{\eta}(\tau)))+\frac{1}{2}h^2\widetilde{c}_2\bm{\Sigma}(\tau)\bm{\theta}^{(2)}(\tau)| + O_P((Th)^{-1} + \beta_Th^2).
		\end{eqnarray*}
		In addition, by the proof of the first result of this lemma, we have
		$$
		\sup_{\tau \in [h,1-h]}|E(\gradient_{\bm{\eta}_1} \widetilde{\mathcal{L}}_\tau(\bm{\eta}(\tau)))+\frac{1}{2}h^2\widetilde{c}_2\bm{\Sigma}(\tau)\bm{\theta}^{(2)}(\tau)| = O(h^3+(Th)^{-1}).
		$$
		The proof is now complete.
	\end{proof}
	
	\subsection{Computation of the Local Linear ML Estimates}\label{AppB.4}
	In our numerical studies, we use the function $\textit{fminunc}$ in programming language MATLAB to minimize the negative of log-likelihood function. The initial guess is important when using optimization functions because these optimizers are trying to find a local minimum, i.e. the one closest to the initial guess that can be achieved using derivatives. In this section, we give a possible choice of initial estimates.
	
	We could estimate the coefficients of time-varying VARMA($p,q$) model
	$$
	\mathbf{x}_t = \sum_{j=1}^{p}\mathbf{A}_j(\tau_t)\mathbf{x}_{t-j} + \bm{\eta}_t + \sum_{j=1}^{q}\mathbf{B}_j(\tau_t)\bm{\eta}_{t-j} \quad \text{with}\quad \bm{\eta}_t = \bm{\omega}(\tau_t)\bm{\varepsilon}_t, 
	$$
	by kernel-weighted least squares method if the lagged $\bm{\eta}_t$ were given. To obtain a preliminary estimator, we first fit a long VAR model and then use estimated residuals in place of true residuals. Consider the VAR$(p_T)$ model
	$$
	\mathbf{x}_t = \sum_{j=1}^{p_T}\bm{\Gamma}_j(\tau_t)\mathbf{x}_{t-j} + \bm{\eta}_t, 
	$$
	where $p_T$ is set to be $2(Th)^{1/3}$ in our numerical studies. Then, we compute $\widehat{\bm{\eta}}_t = \mathbf{x}_t - \sum_{j=1}^{p_T}\widehat{\bm{\Gamma}}_j(\tau_t)\mathbf{x}_{t-j}$, where $\{\widehat{\bm{\Gamma}}_j(\tau)\}$ are the local linear least squares estimators. Given $\widehat{\bm{\eta}}_t$, we are able to estimate $\{\mathbf{A}_j(\tau)\}$, $\{\mathbf{B}_j(\tau)\}$ and $\bm{\Omega}(\tau)$ as well as their derivatives by local linear least squares method.
	
	In order to achieve identifications, certain restrictions should be imposed on the coefficients of the VARMA model. Suppose there exists a known matrix $\mathbf{R}$ and a vector $\bm{\gamma}(\tau)$ satisfying
	$$
	\mathrm{vec}(\mathbf{A}_1(\tau),...,\mathbf{A}_p(\tau),\mathbf{B}_1(\tau),...,\mathbf{B}_q(\tau)) = \mathbf{R}\bm{\gamma}(\tau),
	$$
	which follows that
	$$
	\mathbf{x}_t \simeq (\mathbf{z}_{t-1}^\top\otimes \mathbf{I}_m) \mathbf{R}[\bm{\gamma}(\tau)+\bm{\gamma}^{(1)}(\tau)(\tau_t-\tau)] + \bm{\eta}_{t},
	$$
	where $\mathbf{z}_{t} = [\mathbf{x}_t^\top,...,\mathbf{x}_{t-p+1}^\top,\widehat{\bm{\eta}}_t^\top,...,\widehat{\bm{\eta}}_{t-q}^\top]^\top$. Then the local linear estimator of $(\bm{\gamma}(\tau),\bm{\gamma}^{(1)}(\tau))$ is given by
	$$
	\left(\begin{matrix}
		\widehat{\bm{\gamma}}(\tau)\\
		h\widehat{\bm{\gamma}}^{(1)}(\tau)
	\end{matrix}\right) = \left(\sum_{t=1}^{T}\mathbf{R}^\top \mathbf{Z}_{t-1}^* \mathbf{Z}_{t-1}^{*\top}\mathbf{R}K_h(\tau_t-\tau) \right)^{-1}\sum_{t=1}^{T}\mathbf{R}^\top\mathbf{Z}_{t-1}^*\mathbf{x}_tK_h(\tau_t-\tau),
	$$
	where $\mathbf{Z}_{t}^* = \mathbf{z}_{t} \otimes \mathbf{I}_m \otimes [1, \frac{\tau_{t+1}-\tau}{h}]^\top$. Similarly, the local linear estimator of $(\mathrm{vech}(\bm{\Omega}(\tau)),\mathrm{vech}(\bm{\Omega}^{(1)}(\tau)))$ is given by
	$$
	\left(\begin{matrix}
		\mathrm{vech}(\widehat{\bm{\Omega}}(\tau))\\
		h\mathrm{vech}(\widehat{\bm{\Omega}}^{(1)}(\tau))
	\end{matrix}\right) = \left(\sum_{t=1}^{T} \mathbf{Z}_{t} \mathbf{Z}_{t}^{\top}K_h(\tau_t-\tau) \right)^{-1}\sum_{t=1}^{T}\mathbf{Z}_{t}\mathrm{vech}(\widehat{\bm{\eta}}_t\widehat{\bm{\eta}}_t^\top)K_h(\tau_t-\tau),
	$$
	where $\mathbf{Z}_{t} = [1, \frac{\tau_{t}-\tau}{h}]^\top \otimes \mathbf{I}_{m(m+1)/2}$.
	
	We next consider the preliminary Estimation of Multivariate GARCH Models. Define $\mathbf{y}_t = \mathbf{x}_t \odot \mathbf{x}_t$ and $\mathbf{v}_t = \mathbf{y}_t - \mathbf{h}_t$. We can rewrite model (1.4) as
	$$
	\mathbf{y}_t = \mathbf{c}_0(\tau_t) + \sum_{j=1}^{\max(p,q)}(\mathbf{C}_j(\tau_t) + \mathbf{D}_j(\tau_t))\mathbf{y}_{t-j} + \mathbf{v}_t + \sum_{j=1}^{q}(-\mathbf{D}_j(\tau_t))\mathbf{v}_{t-j}
	$$
	with $E(\mathbf{v}_t | \mathcal{F}_{t-1} ) = 0$. Similar to the VARMA model, we are able to estimate $\mathbf{c}_0(\tau)$, $\{\mathbf{C}_j(\tau)\}$ and  $\{\mathbf{D}_j(\tau)\}$ as well as their derivatives by local linear least squares method. Consider the VAR$(p_T)$ model
	$$
	\mathbf{y}_t = \sum_{j=1}^{p_T}\bm{\Phi}_j(\tau_t)\mathbf{y}_{t-j} + \mathbf{v}_t, 
	$$
	where $p_T$ is set to be $2(Th)^{1/3}$ in our numerical studies. Then, we compute $\widehat{\mathbf{v}}_t = \mathbf{y}_t - \sum_{j=1}^{p_T}\widehat{\bm{\Phi}}_j(\tau_t)\mathbf{p}_{t-j}$, $\widehat{\mathbf{h}}_t = \mathbf{y}_t - \widehat{\mathbf{v}}_t$ and $\widehat{\bm{\eta}}_t = \mathrm{diag}^{-1/2}(\widehat{\mathbf{h}}_t)\mathbf{x}_t$. Hence, the local linear estimator of $(\mathrm{vechl}(\bm{\Omega}(\tau)),\mathrm{vechl}(\bm{\Omega}^{(1)}(\tau)))$ is given by
	$$
	\left(\begin{matrix}
		\mathrm{vechl}(\widehat{\bm{\Omega}}(\tau))\\
		h\mathrm{vechl}(\widehat{\bm{\Omega}}^{(1)}(\tau))
	\end{matrix}\right) = \left(\sum_{t=1}^{T} \mathbf{Z}_{t} \mathbf{Z}_{t}^{\top}K_h(\tau_t-\tau) \right)^{-1}\sum_{t=1}^{T}\mathbf{Z}_{t}\mathrm{vechl}(\widehat{\bm{\eta}}_t\widehat{\bm{\eta}}_t^\top)K_h(\tau_t-\tau),
	$$
	where $\mathbf{Z}_{t} = [1, \frac{\tau_{t}-\tau}{h}]^\top \otimes \mathbf{I}_{m(m-1)/2}$ and $\mathrm{vechl}(\cdot)$ stacks the lower triangular part of a square matrix excluding the diagonal.
	
	}
\end{document}